\documentclass[11pt,letterpaper]{article}
\usepackage[utf8]{inputenc}
\usepackage[margin=1in]{geometry}
\usepackage{graphicx}
\usepackage{appendix}
\usepackage{wrapfig}
\usepackage{amsthm}
\usepackage{enumitem}
\newtheorem{problem}{Problem}
\newtheorem{subproblem}[problem]{Subproblem}
\newtheorem{theorem}{Theorem}
\newtheorem{lemma}[theorem]{Lemma}
\newtheorem{definition}[theorem]{Definition}
\newtheorem{corollary}[theorem]{Corollary}
\newtheorem{observation}[theorem]{Observation}
\newtheorem{fact}[theorem]{Fact}
\usepackage{thm-restate}
\usepackage{amsmath}
\usepackage{amssymb}
\DeclareMathOperator\cis{cis}
\newcommand{\dominates}{>_d}
\newcommand{\undermines}{>_u}
\newcommand{\domseq}{>_{du}}
\usepackage{xcolor}
\usepackage{xspace}
\newcommand{\etal}{\textit{et al}.}
\newcommand{\problemtwo}{SC\xspace}
\usepackage{authblk}
\usepackage{color,soul}
\usepackage{placeins}

\begin{document}

\title{Cubic upper and lower bounds for subtrajectory clustering under the continuous Fr\'echet distance}
\author{Joachim Gudmundsson and Sampson Wong \\
\small{University of Sydney, Australia \\
joachim.gudmundsson@sydney.edu.au, swon7907@sydney.edu.au}}
\date{}

\maketitle

\begin{abstract}
Detecting commuting patterns or migration patterns in movement data is an important problem in computational movement analysis. Given a trajectory, or set of trajectories, this corresponds to clustering similar subtrajectories.

We study subtrajectory clustering under the continuous and discrete Fr\'echet distances. The most relevant theoretical result is by Buchin~\etal~(2011). They provide, in the continuous case, an $O(n^5)$ time algorithm\footnote{\label{quintic}Buchin~\etal~\cite{DBLP:conf/gis/BuchinBGHSSSSSW20} show an $O(|S| \cdot n^2)$ time algorithm, where $S$ is the set of (internal) critical points. In this paper, we show that~$|S| = \Theta(n^3)$, yielding an overall running time of $O(n^5)$.} and a 3SUM-hardness lower bound, and in the discrete case, an $O(n^3)$ time algorithm. We show, in the continuous case,  an $O(n^3 \log^2 n)$ time algorithm and a 3OV-hardness lower bound, and in the discrete case, an $O(n^2 \log n)$ time algorithm and a quadratic lower bound. Our bounds are almost tight unless SETH fails.

\end{abstract}

\section{Introduction}

The widespread use of the Global Positioning System (GPS) in location aware devices has led to an abundance of trajectory data. Although collecting and storing this data is cheaper and easier than ever, this rapid increase of data is making the problem of analysing this data more demanding. One way of extracting useful information from a large trajectory data set is to cluster the trajectories into groups of similar trajectories. However, focusing on clustering entire trajectories can overlook significant patterns that exist only for a small portion of their lifespan. Consequently, subtrajectory clustering is more appropriate if we are interested in similar portions of trajectories, rather than entire trajectories.

Subtrajectory clustering has been used to detect similar movement patterns in various applications. Gudmundsson and Wolle~\cite{DBLP:journals/urban/GudmundssonW14} applied it to football analysis. They reported common movements of the ball, common movements of football players, and correlations between players in the same team moving together. Buchin~\etal~\cite{DBLP:conf/gis/BuchinBGHSSSSSW20} applied subtrajectory clustering to map reconstruction. They reconstructed the location of roads, turns and crossings from urban vehicle trajectories, and the location of hiking trails from hiking trajectories. Many other applications have been considered in the Data Mining and the Geographic Information Systems communities, including behavioural ecology, computational biology and traffic analysis~\cite{DBLP:conf/pods/AgarwalFMNPT18,DBLP:conf/gis/BuchinBDFJSSSW17,DBLP:conf/icdm/ChangZ09,DBLP:conf/gis/GudmundssonTV12,DBLP:journals/tpds/GudmundssonV15,DBLP:conf/sigmod/LeeHW07,DBLP:journals/ijgi/LuoZXFR17,DBLP:conf/bigdataconf/TampakisPDT19}.

Despite considerable attention across multiple communities, the theoretical aspects of subtrajectory clustering are not well understood. The most closely related result is by Buchin~\etal~\cite{DBLP:journals/ijcga/BuchinBGLL11}. Their algorithm forms the basis of several implementations~\cite{DBLP:conf/gis/BuchinBDFJSSSW17,DBLP:conf/gis/BuchinBGHSSSSSW20,DBLP:conf/gis/GudmundssonTV12,DBLP:journals/tpds/GudmundssonV15,DBLP:journals/urban/GudmundssonW14}. Other models of subtrajectory clustering have also been considered~\cite{DBLP:conf/pods/AgarwalFMNPT18,DBLP:journals/corr/abs-2103-06040}, which we will briefly discuss in our related work section.

To measure the similarity between subtrajectories, numerous distance measures have been proposed in the literature~\cite{Ranacher14,tao2021comparative,Toohey15}. In this paper, we will use the (discrete and continuous) Fr\'echet distance, which is the most common and successful distance measure used for trajectories, and also the preferred distance measure in the theory community.

Given a trajectory $T$, the subtrajectory clustering (\problemtwo) problem considered by Buchin~\etal~\cite{DBLP:journals/ijcga/BuchinBGLL11} is to compute a subtrajectory cluster consisting of~$m$ non-overlapping subtrajectories of $T$, one of which is called the reference subtrajectory. The reference subtrajectory must have length at least~$\ell$, and the Fr\'echet distance between the reference subtrajectory and any of the other~$m-1$ subtrajectories is at most~$d$. We formally define the \problemtwo problem in Section~\ref{sec:new_preliminaries}. For \problemtwo under the continuous Fr\'echet distance,  Buchin~\etal~\cite{DBLP:journals/ijcga/BuchinBGLL11} provide an $O(n^5)$ time algorithm and a 3SUM-hardness lower bound. For \problemtwo under the discrete Fr\'echet distance, they provide an $O(n^3)$ time algorithm. Closing the gaps between the two upper and lower bounds have remained important open problems. 


In this paper, we provide an $O(n^3 \log^2 n)$ time algorithm for \problemtwo under the continuous Fr\'echet distance; a significant improvement over the previous algorithm~\cite{DBLP:journals/ijcga/BuchinBGLL11}. Along the way, we also show an $O(n^2 \log n)$ time algorithm for \problemtwo under the discrete Fr\'echet distance. 

We argue that our algorithms are essentially optimal. Our lower bounds are conditional on 
the Strong Exponential Time Hypothesis (SETH). 
Our main technical contribution is an intricate 3OV-hardness lower bound for \problemtwo under the continuous Fr\'echet distance. This implies that there is no $O(n^{3 - \varepsilon})$ time algorithm for any $\varepsilon > 0$, unless SETH fails. We also show, via a simple reduction, that Bringmann's~\cite{DBLP:conf/focs/Bringmann14} SETH-based quadratic lower bound applies to \problemtwo under the discrete Fr\'echet distance. These lower bounds show that our two algorithms are almost optimal, unless SETH fails. Interestingly, our results show that there is a provable separation between the discrete and continuous Fr\'echet distance for SC.

Next, we outline the structure of the paper. We discuss related work in Section~\ref{sec:related_work} and preliminaries in Section~\ref{sec:new_preliminaries}. In Section~\ref{sec:overview_algorithm}, we provide an overview of the key insights that lead to our improved algorithms. In Section~\ref{sec:overview_3ov}, we give a quadratic lower bound in the discrete case, and an overview of the key components of our cubic lower bound in the continuous case. Detailed descriptions and the full proofs are provided in Section~\ref{sec:dfd} for our algorithm under the discrete Fr\'echet distance, in Section~\ref{sec:cfd} for our algorithm under the continuous Fr\'echet distance, and in Section~\ref{sec:3ov} for our 3OV reduction.

\subsection{Related work}
\label{sec:related_work}
Recently, the closely related problem of clustering trajectories has received considerable attention, especially the $(k,\ell)$-center and $(k,\ell)$-median clustering problems. In these problems, entire trajectories are clustered. Given a set $\mathcal G$ of trajectories, and parameters $k$ and $\ell$, the problem is to find a set $\mathcal C$ of $k$ trajectories (not necessarily in $\mathcal G$), each of complexity at most $\ell$, so that the maximum Fr\'echet distance (center) or the sum of the Fr\'echet distances (median) over all trajectories in $\mathcal G$ to its closest trajectory in $\mathcal C$ is minimised. The set $\mathcal C$ is also known as the set of center curves, and the intuition behind restricting the complexity of the center curves is to avoid overfitting. 

Driemel~\etal~\cite{DBLP:conf/soda/DriemelKS16} were the first to consider $(k,\ell)$-center and $(k,\ell)$-medians clustering of trajectories. They showed that both problems are NP-hard when $k$ is part of the input, and provided $(1+\varepsilon)$-approximation algorithms if the trajectories are one-dimensional. Buchin~\etal~\cite{DBLP:conf/soda/BuchinDGHKLS19} showed that $(k,\ell)$-center clustering is NP-hard if $\ell$ is part of the input, and provided a 3-approximation algorithm for trajectories of any dimension. Buchin~\etal~\cite{DBLP:conf/soda/BuchinDR21} provided a randomised bicriteria-approximation algorithm with approximation factor $(1+\varepsilon)$ for trajectories of any dimension.

The idea of computing a set of center curves for trajectory clustering has been extended to subtrajectory clustering. Agarwal~\etal~\cite{DBLP:conf/pods/AgarwalFMNPT18} compute center curves (which they call pathlets) from a set of input trajectories. The key difference is that pathlets are similar to portions of the input trajectory, rather than the entire trajectories. Each trajectory is then modelled as a concatenation of pathlets, with possible gaps in between. Agarwal~\etal~\cite{DBLP:conf/pods/AgarwalFMNPT18} show that their problem is NP-hard, and provide an $O(\log n)$-approximation algorithm that runs in polynomial time. Akitaya~\etal~\cite{DBLP:journals/corr/abs-2103-06040} consider a very similar model, except that they compute a set of center curves with complexity at most $\ell$, and the concatenation of center curves covers the input trajectory without gaps. They also show that this version is NP-hard, and provide a polynomial-time $O(\log n)$-approximation.

The first fine-grained lower bound based on the Strong Exponential Time Hypothesis (SETH) to be applied to Fr\'echet distance problems was by Bringmann~\cite{DBLP:conf/focs/Bringmann14}, who showed a quadratic lower bound for computing the (discrete or continuous) Fr\'echet distance between trajectories of dimension two or higher. Bringmann and Mulzer~{\cite{DBLP:journals/jocg/BringmannM16}} extended the results to also hold for the discrete Fr\'echet distance of 1-dimensional trajectories. Buchin~\etal~{\cite{DBLP:conf/soda/BuchinOS19}} further extended the results to hold for the continuous Fr\'echet distance of 1-dimensional trajectories. Bringman~\etal~\cite{DBLP:conf/soda/BringmannKN19} showed a 4OV-hardness lower bound for computing the translation invariant discrete Fr\'echet distance between two trajectories of dimension two or higher.

\section{Preliminaries}
\label{sec:new_preliminaries}

In this paper, we study the problem of detecting a movement pattern that occurs frequently in a trajectory, or in a set of trajectories. The problem was first proposed by Buchin~et~al.~{\cite{DBLP:journals/ijcga/BuchinBGLL11}}, and we retain the existing convention by referring to this problem as the subtrajectory clustering (SC)~problem. In the SC problem, a trajectory $T$ of complexity $n$ is defined to be a sequence of points $v_1, v_2, \ldots, v_n$ in the $c$-dimensional Euclidean space $\mathbb R^c$, connected by segments.


\begin{problem}[SC problem]
\label{def:p2}
Given a trajectory~$T$ of complexity~$n$, a positive integer~$m$, and positive real values~$\ell$ and~$d$, decide if there exists a subtrajectory cluster of~$T$ such that:
\begin{itemize} [noitemsep]
    \item the cluster consists of one reference subtrajectory and $m-1$ other subtrajectories of $T$,
    \item the reference subtrajectory has Euclidean length at least $\ell$,
    \item the Fr\'echet distance between the reference subtrajectory and any other subtrajectory is~$\leq d$,
    \item any pair of subtrajectories in the cluster overlap in at most one point.
\end{itemize}
\end{problem}

Buchin~\etal~\cite{DBLP:journals/ijcga/BuchinBGLL11} show how the case where the input is a set of trajectories can be reduced to the case when a single trajectory is given as input. See Figure~\ref{fig:sc01_clustering}. Since our two algorithms build upon the algorithm by Buchin~\etal~\cite{DBLP:journals/ijcga/BuchinBGLL11} we will briefly describe their algorithm. First, we discuss their algorithm for the discrete Fr\'echet distance, then for the continuous Fr\'echet distance.

\begin{figure}[ht]
    \centering
    \includegraphics[width=0.8\textwidth]{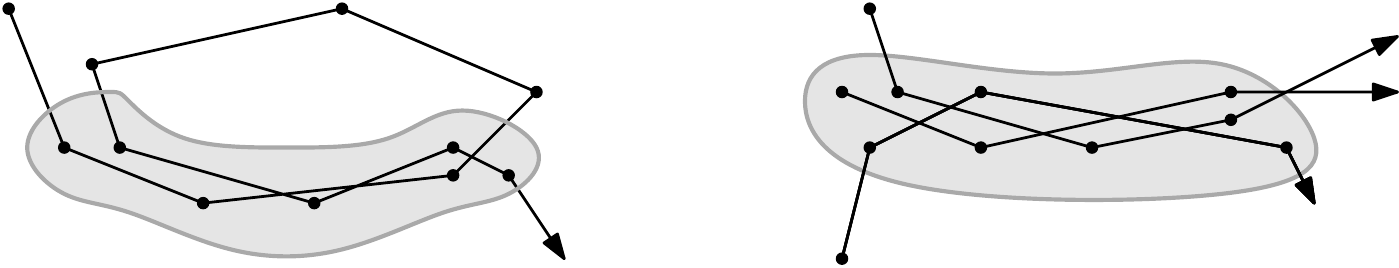}
    \caption{A subtrajectory cluster, for a single trajectory (left), or for a set of trajectories (right).}
    \label{fig:sc01_clustering}
\end{figure}

The first step is to transform \problemtwo into a problem in the discrete free space diagram. Let $F_d(T,T)$ be the discrete Fr\'echet free space diagram between two copies of $T$ and with distance value $d$. For the definition and a formal discussion of the free space diagram, see Appendix~\ref{sec:appendix_frechet_freespace}. Suppose the conditions of \problemtwo hold for some reference subtrajectory starting at vertex~$s$ and ending at vertex~$t$. Let $l_s$ and $l_t$ be the vertical lines in $F_d(T,T)$ representing $s$ and $t$. The conditions of \problemtwo state that there are $m-1$ other subtrajectories so that the Fr\'echet distance between the reference subtrajectory and each of the other subtrajectories is at most $d$. Therefore, the \problemtwo problem reduces to deciding if there are vertical lines $l_s$ and~$l_t$ so that the reference subtrajectory from $s$ to $t$ has length at least $\ell$, and there are $m-1$ monotone paths in $F_d(T,T)$ starting at $l_s$ and ending at~$l_t$. Moreover, since these $m-1$ subtrajectories overlap with each other in at most one point, the $y$-coordinates of our monotone paths overlap in at most one point. Similarly, the $y$-coordinates of our monotone paths overlap with the $y$-interval from $s$ to $t$ in at most one point. See Figure~\ref{fig:sc_02_monotone_paths}, left.

\begin{figure}[ht]
    \centering
    \includegraphics[width=0.7\textwidth]{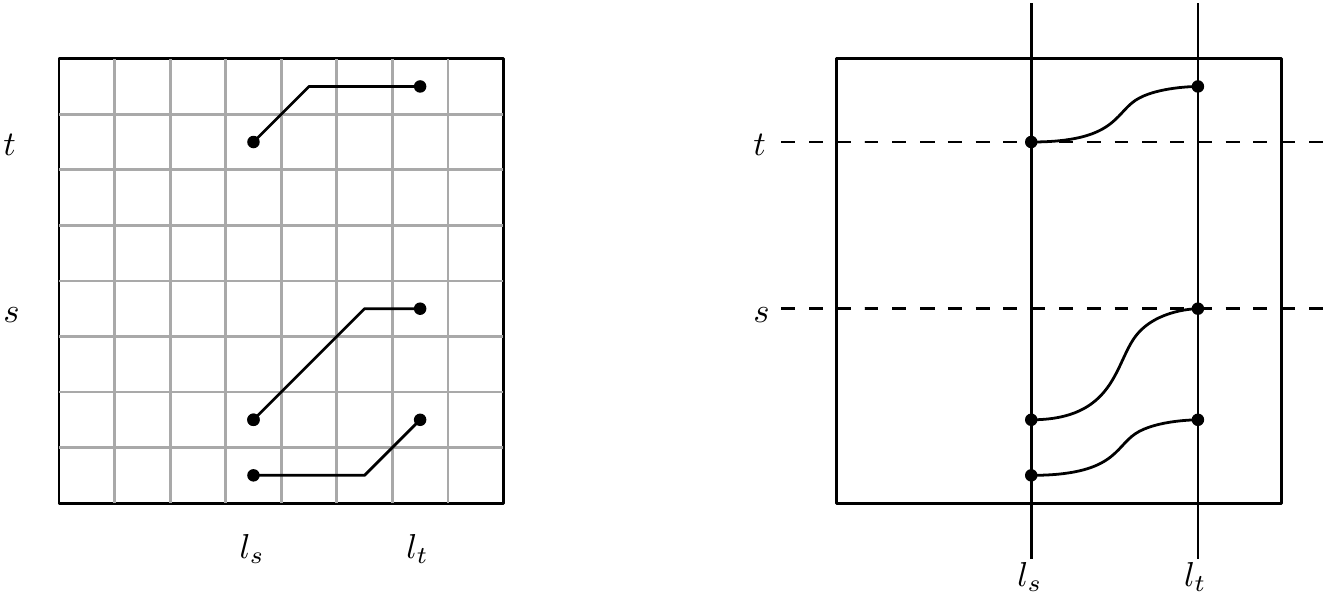}
    \caption{An example of three monotone paths from $l_s$ to $l_t$ in the discrete (left) and continuous (right) free space diagrams. The monotone paths overlap with each other in at most one point, and overlap with the $y$-interval from $s$ to $t$ in at most one point.}
    \label{fig:sc_02_monotone_paths}
\end{figure}

The second step is to iterate over all reference subtrajectories with length at least $\ell$. Buchin~\etal~\cite{DBLP:journals/ijcga/BuchinBGLL11} show that, for \problemtwo under the discrete Fr\'echet distance, there are only $O(n)$ candidate reference subtrajectories to consider, and of these, no reference subtrajectory is a subtrajectory of any other reference subtrajectory.

The third step is, given a candidate subtrajectory starting at $s$ and ending at $t$, deciding whether there is a subtrajectory cluster satisfying the conditions of \problemtwo with this candidate subtrajectory as its reference subtrajectory. The third step is an important subproblem in both the algorithm of Buchin~\etal~\cite{DBLP:journals/ijcga/BuchinBGLL11} and our algorithm. We state the subproblem formally.

\begin{subproblem}
\label{def:p3}
Given a trajectory~$T$ of complexity~$n$, a positive integer~$m$, a positive real value~$d$, and a reference subtrajectory of~$T$ starting at vertex~$s$ and ending at vertex~$t$, let $l_s$ and $l_t$ be two vertical lines in $F_d(T,T)$ representing the vertices $s$ and $t$. Decide if there exist:
\begin{itemize}[noitemsep]
    \item $m-1$ distinct monotone paths starting at $l_s$ and ending at $l_t$, such that 
    \item the $y$-coordinate of any two monotone paths overlap in at most one point, and
    \item the $y$-coordinate of any monotone path overlaps the $y$-interval from $s$ to $t$ in at most one point.
\end{itemize}
\end{subproblem}


Buchin~\etal~\cite{DBLP:journals/ijcga/BuchinBGLL11} solve each instance of Subproblem~\ref{def:p3} individually in $O(n + ml)$ time, where $l$ is the maximum complexity of the reference subtrajectory. As there are $O(n)$ reference subtrajectories, there are $O(n)$ subproblems to solve. Therefore, the total running time of the algorithm is $O(n^2 + nml)$ which in the worst case is $O(n^3)$.

Next, we describe the algorithm of Buchin~\etal~\cite{DBLP:journals/ijcga/BuchinBGLL11} under the continuous Fr\'echet distance. 
The first step is exactly the same as the discrete case, except that we substitute the discrete free space diagram with the continuous free space diagram. See Figure~\ref{fig:sc_02_monotone_paths}, right.

The second step is significantly different between the discrete and continuous case. In the continuous case, the starting point $s$ and ending point $t$ may be any arbitrary points on the trajectory $T$, not just vertices of $T$. Buchin~\etal~\cite{DBLP:journals/ijcga/BuchinBGLL11} claim that there are $O(n^2)$ critical points in the free space diagram, and the vertical line $l_s$ representing the starting point of the reference subtrajectory must pass through one of these critical points. We show in Section~\ref{sec:overview_3ov} that there are $\Omega(n^3)$ critical points, and we show in Lemma~\ref{lem:cfd.a2a.n^3_critical_points} that there are $O(n^3)$ critical points. Hence, in the general case, there are $\Theta(n^3)$ possible reference subtrajectories to consider.

The third step is to solve Subproblem~\ref{def:p3}. Instead of $l_s$ and $l_t$ representing vertices that are the starting and ending points of the reference subtrajectories, we consider $l_s$ and $l_t$ representing arbitrary points that are the starting and ending points. Buchin~\etal~\cite{DBLP:journals/ijcga/BuchinBGLL11} solve each instance of Subproblem~\ref{def:p3} in $O(nm)$ time. As there are $O(n^3)$ critical points and therefore $O(n^3)$ reference subtrajectories to consider, the overall running time of Buchin~\etal's~\cite{DBLP:journals/ijcga/BuchinBGLL11} algorithm is $O(n^4 m)$, which in the worst case is $O(n^5)$.

\section{Technical Overview}

Our main technical contributions in this paper are cubic upper and lower bounds for Subtrajectory Clustering (SC) under the continuous Fr\'echet distance. 

As a stepping stone towards our cubic upper bound, we study two special cases. The first special case is the \problemtwo problem under the discrete Fr\'echet distance. The second special case is the \problemtwo problem under the continuous Fr\'echet distance, but under the restriction that the reference subtrajectory must be a vertex-to-vertex subtrajectory. The final case is the general case. In Section~\ref{sec:overview_algorithm}, we provide an overview of our key insights in each case. Detailed descriptions of the algorithms and their analyses can be found in Sections~\ref{sec:dfd} and~\ref{sec:cfd}.

For our lower bounds, assuming SETH, we show in Section~\ref{sec:overview_3ov} a simple reduction that proves that there is no $O(n^{2-\varepsilon})$ algorithm for \problemtwo for any $\varepsilon > 0$ under the discrete Fr\'echet distance. Then, we show that there is no $O(n^{3-\varepsilon})$ algorithm for \problemtwo for $\varepsilon > 0$ under the continuous Fr\'echet distance. We provide an overview of our reduction from 3OV to \problemtwo in Section~\ref{sec:overview_3ov}, and a detailed analysis of our construction can be found in Section~\ref{sec:3ov}.

\subsection{Algorithm Overview}
\label{sec:overview_algorithm}

Our algorithms, both in the special and general cases, build on the work of Buchin~\etal~\cite{DBLP:journals/ijcga/BuchinBGLL11}. We make several key insights that lead to improvements over previous work. In this section, we present our main technical contributions, and defer relevant proofs to Sections~\ref{sec:dfd} and~\ref{sec:cfd}.

\subsubsection*{Key Insight 1: Reusing monotone paths between different subproblems}

As previously stated, the algorithm of Buchin~\etal~\cite{DBLP:journals/ijcga/BuchinBGLL11} divides \problemtwo into one subproblem per candidate reference subtrajectory. 
For the discrete Fr\'echet distance, there are $O(n)$ candidate reference subtrajectories, so \problemtwo is divided into $O(n)$ instances of Subproblem~\ref{def:p3}. Buchin~\etal~\cite{DBLP:journals/ijcga/BuchinBGLL11} showed that each instance of Subproblem~\ref{def:p3} can be solved in $O(n + ml) \approx O(n^2)$ time individually, so all $O(n)$ subproblems can be solved in $O(n^2 + nml) \approx O(n^3)$ time. 

Our first key insight is that we need not handle each subproblem individually. If we can reuse monotone paths between the $O(n)$ subproblems, this could significantly speed up the algorithm. For example, suppose that a subproblem starts at vertex $s$ and ends at vertex $t$, whereas an adjacent subproblem starts at vertex $s+1$ and ends at vertex $t+1$. Suppose that we solve the subproblem $(s,t)$ first, and we solve the subproblem $(s+1,t+1)$ next. Our key insight is to reuse the monotone paths that we computed in subproblem $(s,t)$ to guide our search in subproblem $(s+1,t+1)$. In fact, since almost all grid cells in subproblem $(s+1,t+1)$ have already appeared in the subproblem $(s,t)$, the speedup from using previously computed paths could be quite large. See Figure~\ref{fig:to_ki1_1}, left.

\begin{figure}[ht]
    \centering
    \includegraphics[width=0.8 \textwidth]{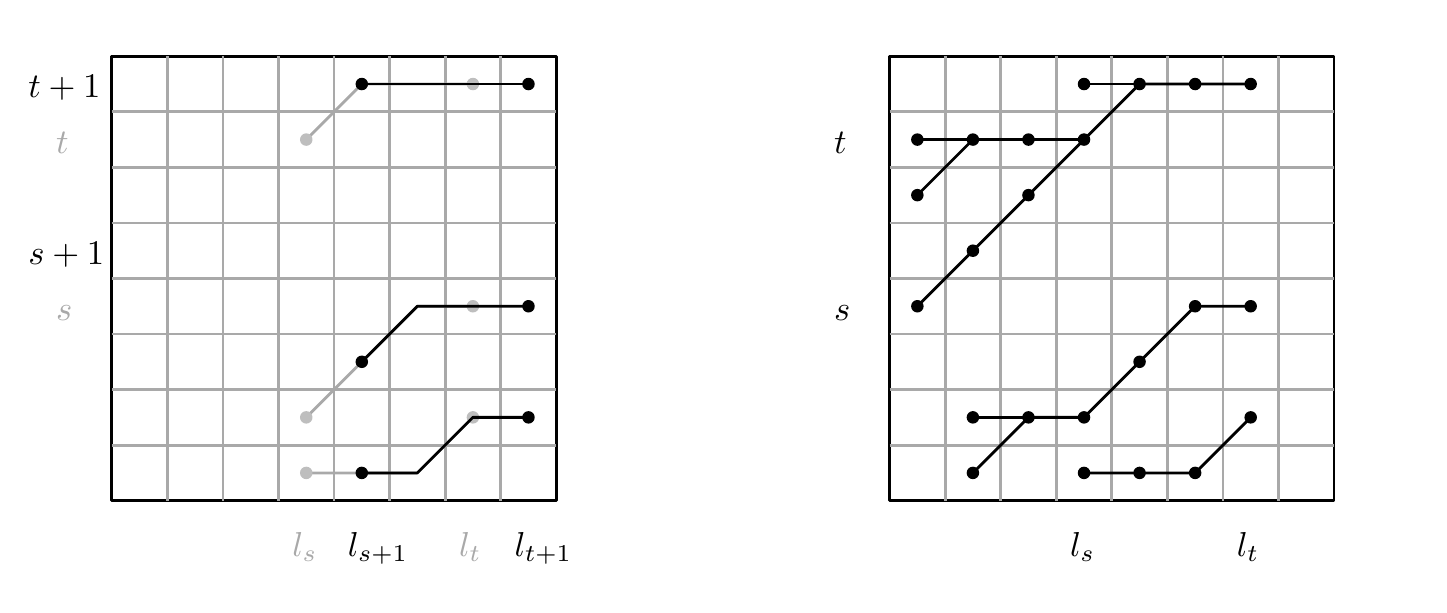}
    \caption{Monotone paths from $l_s$ to $l_t$ that are similar to monotone paths from $l_{s+1}$ to $l_{t+1}$ (left). The link-cut tree data structure retaining monotone paths from previous values of $s$ and $t$ (right).}
    \label{fig:to_ki1_1}
\end{figure}

In short, our approach is as follows. We sort the $O(n)$ subproblems by the $x$-coordinates of their vertical lines $l_s$ and $l_t$. We consider those with the smallest $x$-coordinates first. For each subproblem, we perform a greedy depth first search to find $m-1$ non-overlapping monotone paths. Our search algorithm is very similar to and has the same running time as the original search algorithm by Buchin~\etal~\cite{DBLP:journals/ijcga/BuchinBGLL11}. 

While performing our greedy depth first search, we maintain a dynamic tree data structure to store our monotone paths as we compute them. In particular, we use a link-cut tree~\cite{DBLP:journals/jcss/SleatorT83}, to store a set of rooted trees. The invariant maintained by our data structure is that every node has a monotone path to the root of its link-cut tree, as shown in Figure~\ref{fig:to_ki1_1}, right. The link-cut data structure is maintained via edge insertions and deletions, which require $O(\log n)$ time per update. Using this data structure, we can significantly reuse the monotone paths between subproblems. Whenever we visit a node that has been considered by a previous subproblem, we can simply query for its root in the link-cut data structure, without needing to recompute the monotone path.

In Section~\ref{sec:dfd.dfs}, we provide details of our greedy depth first search. In Section~\ref{sec:dfd.jump} we provide details of how we apply the link-cut tree data structure to our greedy depth first search. We then use amortised analysis to bound the running time. Putting this together yields:

\begin{restatable}{theorem}{theoremone}
\label{thm:main_dfd}
There is an $O(n^2 \log n)$ time algorithm for \problemtwo under the discrete Fr\'echet distance.
\end{restatable}

\subsubsection*{Key Insight 2: Transforming continuous free space reachability into graph reachability}

As we transition from the discrete case of \problemtwo to the continuous case, the free space diagram becomes significantly more complex. In particular, reachability in continuous free space is calculated in a fundamentally different way to reachability in the discrete free space. The most common approach for computing the continuous Fr\'echet distance is to compute the reachable space. Reachable space is defined to be the subset of free space that is reachable via monotone paths from the bottom left corner, as shown in Figure~\ref{fig:to_ki2_1}, left. For example, this is the approach used by Alt and Godau~\cite{DBLP:journals/ijcga/AltG95} in their original algorithm, and also by Buchin~\etal~\cite{DBLP:journals/dcg/BuchinBMM17} in their improved algorithm.

\begin{figure}[ht]
    \centering
    \includegraphics[width=0.7\textwidth]{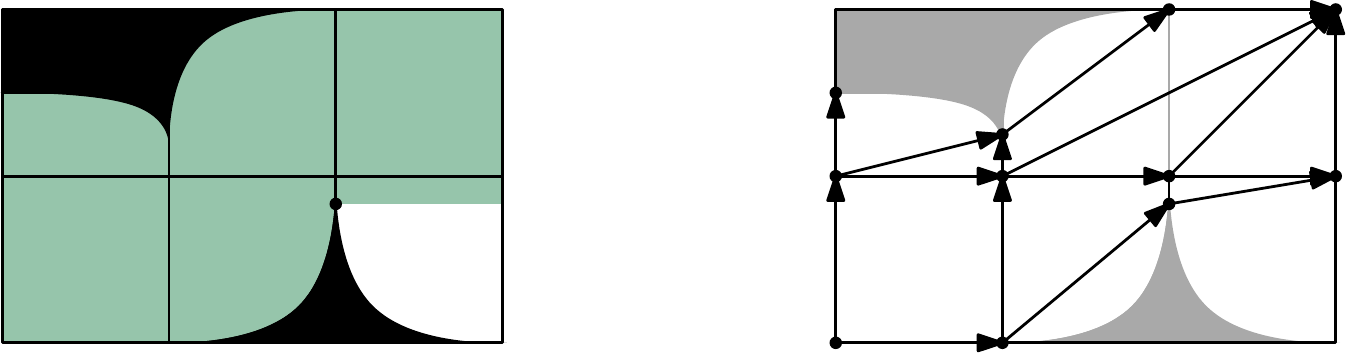}
    \caption{The reachable free space from the bottom left corner, shaded in green (left). The directed graph $G=(V,E)$ where there is a path in the graph if and only if there is a monotone path in the continuous free space diagram (right).}
    \label{fig:to_ki2_1}
\end{figure}

Maintaining the reachable space is fundamentally incompatible with our algorithm. The reason is that the reachable space only considers monotone paths that start at the bottom left corner of the free space diagram. In our problem, the starting point of our monotone path constantly changes, and the reachable space for one starting point may not be used for computing monotone paths from any other starting point.

We propose an alternative to computing the reachable space in the continuous free space diagram. A critical point in the continuous free space diagram is the intersection of the boundary of a cell with the boundary of the free space for that cell (ellipse), or is a cell corner. We build a directed graph $G=(V,E)$ of size $O(n^2 \log n)$, where $V$ is the set of $O(n^2)$ critical points in the free space diagram, so that there is a monotone path between two points in the free space diagram if and only if there is a path between the same two points in the directed graph $G$. See Figure~\ref{fig:to_ki2_1}, right. Note that reachability in this directed graph works for any pair of critical points, not just those where the starting point is the bottom left corner.

In Section~\ref{sec:cfd.graph}, we establish the equivalence between continuous free space reachability and graph reachability in $G$. In Section~\ref{sec:cfd.jump.v2v}, we combine this with our first key insight to obtain an $O(n^2 \log^2 n)$ time algorithm for \problemtwo under the continuous Fr\'echet distance, in the special case where the reference subtrajectory is vertex-to-vertex. 

\subsubsection*{Key Insight 3: Handling additional critical points and reference subtrajectories}

In all special cases considered so far, there are $O(n)$ candidate reference subtrajectories, and hence $O(n)$ instances of Subproblem~\ref{def:p3}. However, in the continuous case where the reference subtrajectory may be any arbitrary subtrajectory, there are significantly more candidate reference subtrajectories to consider. We call a potential starting point for the reference subtrajectory either an internal or external critical point. The difference is that an internal critical point lies in the interior of a free space cell, whereas an external critical point lies on the boundary of a free space cell. 

We show that there are $O(n^3)$ internal critical points, and therefore $O(n^3)$ instances of Subproblem~\ref{def:p3} in the general case. The number of internal critical points is dominated by critical points of the following type: a point that is on the boundary between free and non-free space, and shares a $y$-coordinate with an external critical point. For an example, see Figure~\ref{fig:to_ki3}, left. Surprisingly, these $O(n^3)$ internal critical points are necessary, and form one of the key components of our lower bound.

\begin{figure}[ht]
    \centering
    \includegraphics[width=0.8\textwidth]{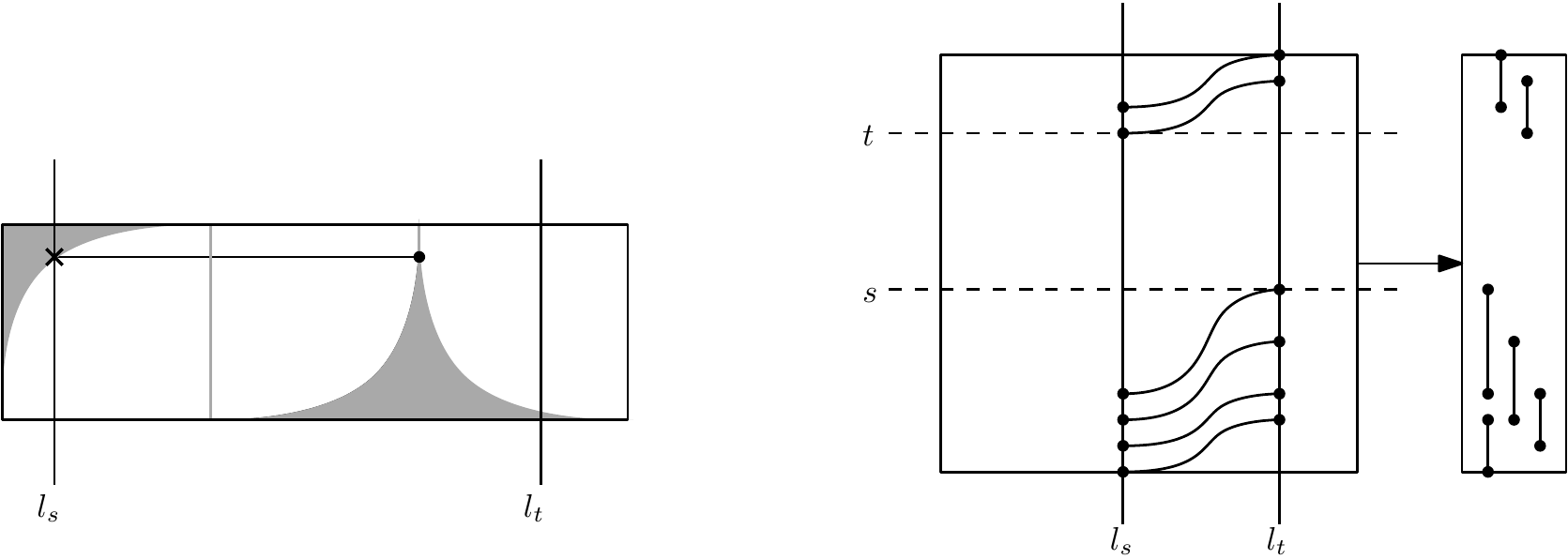}
    \caption{An example of an internal critical point, marked with a cross (left). The overlapping $y$-intervals maintained by the dynamic monotone interval data structure (right).}
    \label{fig:to_ki3}
\end{figure}

Given this increase in critical points and reference subtrajectories in the worst case, if we were to na\"ively apply our first two key insights, we would obtain an $O(n^3 m)$ time algorithm for the general case. The difference is that, previously, our algorithm's running times were dominated by maintaining the link-cut tree data structure over our set of critical points. However, now that there are significantly more reference subtrajectories, the relatively simple process of enumerating up to $m-1$ non-overlapping monotone paths per reference subtrajectory is the new bottleneck. To handle this bottleneck, we require another data structure. The dynamic monotone interval data structure~\cite{DBLP:journals/tcs/GavruskinKKL15} reports whether there exist $m-1$ non-overlapping $y$-intervals, which represent our $m-1$ non-overlapping monotone paths, without explicitly computing them. See Figure~\ref{fig:to_ki3}, right. Putting this together with our first two key insights we obtain the following theorem.

\begin{restatable}{theorem}{theoremtwo}
\label{thm:main_cfd}
There is an $O(n^3 \log^2 n)$ time algorithm for \problemtwo under the continuous Fr\'echet distance.
\end{restatable}

\subsection{Lower Bound Overview}
\label{sec:overview_3ov}

Our aim is to show that the algorithms in Theorems~\ref{thm:main_dfd} and~\ref{thm:main_cfd} are essentially optimal, assuming SETH. We start by showing that there is no strongly subquadratic time algorithm for \problemtwo under the discrete Fr\'echet distance. We achieve this by reducing from Bringmann's~\cite{DBLP:conf/focs/Bringmann14} lower bound for computing the discrete Fr\'echet distance. 

\begin{theorem}
\label{theorem:bringmann_reduction}
There is no $O(n^{2-\varepsilon})$ time algorithm for \problemtwo under the discrete Fr\'echet distance, for any~$\varepsilon>0$, unless SETH fails.
\end{theorem}

\begin{proof}
Assuming SETH, Bringmann~\cite{DBLP:conf/focs/Bringmann14} constructs a pair of trajectories, $T_1$ and $T_2$, so that deciding whether the discrete Fr\'echet distance between $T_1$ and $T_2$ is at most one has no $O(n^{2 - \varepsilon})$ time algorithm for any $\varepsilon > 0$. Let the trajectories $T_1$ and $T_2$ lie within a ball of radius $r$ centered at the origin, and let~$\ell$ be the sum of the lengths of the trajectories $T_1$ and $T_2$. Let $\lambda = \ell + 2r$, and place point $A$ at $(-4\lambda,0)$ and point $B$ at $(4 \lambda,0)$. Construct the trajectory $T = A \circ T_1 \circ B \circ A \circ T_2 \circ B$. We will show that there is a subtrajectory cluster of $T$ with parameters $m=2$, $\ell=8\lambda$ and $d=1$ if and only if the discrete Fr\'echet distance between $T_1$ and $T_2$ is at most one.

Our key observation is that the subtrajectories $A \circ T_1 \circ B$ and $A \circ T_2 \circ B$ have discrete Fr\'echet distance of at most one if and only if $T_1$ and $T_2$ have a discrete Fr\'echet distance of at most one.

For the if direction, the pair of subtrajectories $A \circ T_1 \circ B$ and $A \circ T_2 \circ B$ each have length at least $8\lambda$, and have discrete Fr\'echet distance at most one from one another.

For the only if direction, suppose there is a subtrajectory cluster of size two. Note that subtrajectory $A \circ T_1$, or any other subtrajectory that contains at most one copy of $A$ or $B$, will have length strictly less than $8 \lambda$. Therefore, $A$ and $B$ must occur in both subtrajectories in the cluster. The corresponding $A$'s and $B$'s must match to one another for the discrete Fr\'echet distance to be at most one, so they must be in the same order in their respective subtrajectory. Without loss of generality, the first subtrajectory contains $A \circ T_1 \circ B$ and the second subtrajectory contains $A \circ T_2 \circ B$. So the discrete Fr\'echet distance between $T_1$ and $T_2$ must be at most one, as required.

As there is no $O(n^{2 - \varepsilon})$ time algorithm for deciding if the discrete Fr\'echet distance is at most one for any $\varepsilon > 0$, there is no $O(n^{2 - \varepsilon})$ time algorithm for \problemtwo under the discrete Fr\'echet distance.
\end{proof}

Next, we show that there is no strongly subcubic time algorithm for \problemtwo under the continuous Fr\'echet distance. We achieve this by reducing the three orthogonal vectors problem (3OV) to SC. 

\begin{problem}[3OV]
We are given three sets of vectors $\mathcal X = \{X_1, X_2, \ldots, X_n\}$, $\mathcal Y = \{Y_1, Y_2, \ldots, Y_n\}$ and $\mathcal Z = \{Z_1, Z_2, \ldots, Z_n\}$. For $1 \leq i,j,k \leq n$, each of the vectors $X_i$, $Y_j$ and $Z_k$ are binary vectors of length $W$. Our problem is to decide whether there exists a triple of integers $1 \leq i,j,k \leq n$ such that $X_i$, $Y_j$ and $Z_k$ are orthogonal. The three vectors are orthogonal if
$X_i[h] \cdot Y_j[h] \cdot Z_k[h] = 0$ for all $1 \leq h \leq W$.
\end{problem}

We employ a three step process in our reduction in Section~\ref{sec:3ov}. First, given a 3OV instance $(\mathcal X, \mathcal Y, \mathcal Z)$, we construct in $O(nW)$ time an \problemtwo instance $(T,m,\ell,d)$ of complexity $O(nW)$. Second, for this instance, we consider the free space diagram $F_d(T,T)$ and prove various properties of it. Third, we use the properties of $F_d(T,T)$ to prove that $(\mathcal X, \mathcal Y, \mathcal Z)$ is a YES instance if and only if $(T,m,\ell,d)$ is a YES instance.

Our reduction implies that there is no $O(n^{3-\varepsilon})$ time algorithm for SC for any $\varepsilon > 0$, unless SETH fails. If such an algorithm for \problemtwo were to exist, then by our reduction we would obtain an $O(n^{3-\varepsilon}W^{3 - \varepsilon})$ time algorithm for 3OV. But under the Strong Exponential Time Hypothesis (SETH), there is no $O(n^{3-\varepsilon}W^{O(1)})$ time algorithm for 3OV, for any $\varepsilon > 0$~\cite{williams2018some}.

Next, we present the three key components of our reduction. For the full reduction see Section~\ref{sec:3ov}. 

\subsubsection*{Key Component 1: Diamonds in continuous free space}

As a stepping stone towards the full reduction, we construct a trajectory $T$ so that $F_d(T,T)$ has $\Theta(n^3)$ internal critical points. This weaker result shows that the analysis of our algorithm in Section~\ref{sec:cfd.jump.a2a} is essentially tight, up to polylogarithmic factors.

To obtain these $\Theta(n^3)$ internal critical points, we introduce the first key component of our reduction. It is a method to generate two curves $T_1$ and $T_2$, so that $F_d(T_2,T_1)$ consists predominantly of free space, with small regions of diamond-shaped non-free space. By varying the positions of the vertices on $T_1$ and $T_2$, we can change both the position and sizes of these small diamonds in $F_d(T_2,T_1)$.

To construct $T_1$ and $T_2,$ we use the polar coordinates in the complex plane. Recall that $r \cis \theta$ has distance~$r$ from the origin and is at an anticlockwise angle of $\theta$ from the positive real axis. The vertices of~$T_1$ will be on the ball of radius $r$ centered at the origin, as illustrated in Figure~\ref{fig:to_kc1_1}, left. The vertices of $T_2$ will be on the ball of radius~$r'$ centered at the origin, where $r' > 10r$. For now, define $\phi = \frac {\pi} 2$, although a much smaller value of $\phi$ will be used in the full reduction.

Let $O$ be the origin. Place the vertices $A$, $B$, $C$, $D$, respectively, at $r \cis 0$, $r \cis \phi$, $r \cis (2 \phi)$, and $r \cis (3 \phi)$. Place the vertices $A^+$, $B^+$, $C^+$, $D^+$, respectively, at $r' \cis \pi$, $r' \cis (\pi + \phi)$, $r' \cis (\pi + 2 \phi)$, and $r' \cis (\pi + 3 \phi)$. Note that $A$ and $A^+$ are diametrically opposite from one another, but on differently sized circles. Let $||\cdot||$ denote the Euclidean norm.  Define $d = ||A^+B||$. Then $A^+$ is within distance $d$ of $B$, $C$ and $D$, but not within distance~$d$ of $A$. Similarly, $A$ is within distance~$d$ of $B^+$, $C^+$ and $D^+$, but not within distance~$d$ of $A^+$. We call the pairs $(A^+,A)$, $(B^+,B)$, $(C^+,C)$ and $(D^+,D)$ antipodes. 

\begin{figure}[ht]
    \centering
    \includegraphics[width=0.7\textwidth]{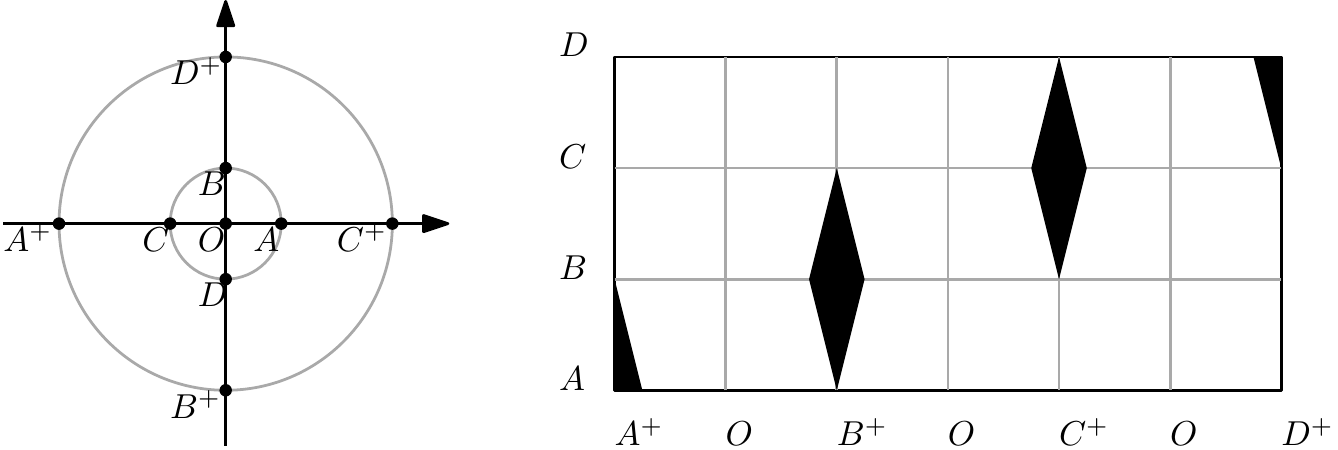}
    \caption{The vertices of $T_1$ and $T_2$ (left), and the free space diagram $F_d(T_2,T_1)$ (right).}
    \label{fig:to_kc1_1}
\end{figure}

Our key insight is that, if the only vertices that appear in $T_1$ are $A$, $B$, $C$ and $D$, and the only vertices that appear in $T_2$ are $A^+$, $B^+$, $C^+$ and $D^+$, then the free space diagram $F_d(T_2,T_1)$ will consist predominantly of free space, with small regions of non-free space. The centers of the regions of non-free space will have $x$-coordinate $v^+$ and $y$-coordinate $v$, where $(v^+,v)$ are antipodes. Section~\ref{sec:3ov_antipodal_property} is dedicated to properties of antipodal pairs, and Section~\ref{sec:3ov_free_space_diagram} is dedicated to how these small regions associated with the antipodes fit in with the rest of the free space diagram.

Let $A \circ B \circ C \circ D$ denote the trajectory formed by joining, with straight segments, the vertices $A$, $B$, $C$, and $D$ in order. Define:
\[
    \begin{array}{rcl}
        T_1 &=& A \circ B \circ C \circ D\\
        T_2 &=& A^+ \circ O \circ B^+ \circ O \circ C^+ \circ O \circ D^+.
    \end{array}
\]

The vertices of $T_1$ and $T_2$ and the free space diagram $F_d(T_2,T_1)$ are shown in Figure~\ref{fig:to_kc1_1}, left. It is not too difficult to see that by swapping the order of the vertices in $T_2$, or by inserting additional vertices into $T_2$, we can change the placements of our diamonds, or partial diamonds, in the free space in Figure~\ref{fig:to_kc1_1}, right. 

We would also like to vary the sizes of our diamonds. To do this, we introduce points that are approximately distance $r'$ from the origin. Let $B^-$ be on $BB^+$ so that $||BB^-|| = d$, that is, $B^- = (d-r) \cis \phi$. Similarly, define $C^- = (d-r) \cis (2 \phi)$. Next, let $B_1^+, B_2^+, \ldots, B_n^+$ be points on segment $B^-B^+$ so that $B^-,B_1^+,\ldots,B_n^+,B^+$ are evenly spaced. Then the antipodal pair $(B^+_i,B)$ would generate a different sized diamond to the antipodal pair $(B^+,B)$. This is because $||B^+_iB|| < d$, 
so the non-free space it generates will be smaller. 

We are ready to construct the trajectory $T$ with $n^3$ internal critical points. Define 

\[
    \begin{array}{rcl}
        T_1 &=& A \circ B \circ C \circ D, \\
        T_2 &=& \bigcirc_{1 \leq i \leq n} (B_i^+ \circ O) \circ \bigcirc_{1 \leq i \leq n} (C^+ \circ O),
    \end{array}
\]
and set $T = \bigcirc_{1 \leq i \leq n}(T_1) \circ T_2$. Note that $T$ has linear complexity.

\begin{figure}[ht]
    \centering
    \includegraphics{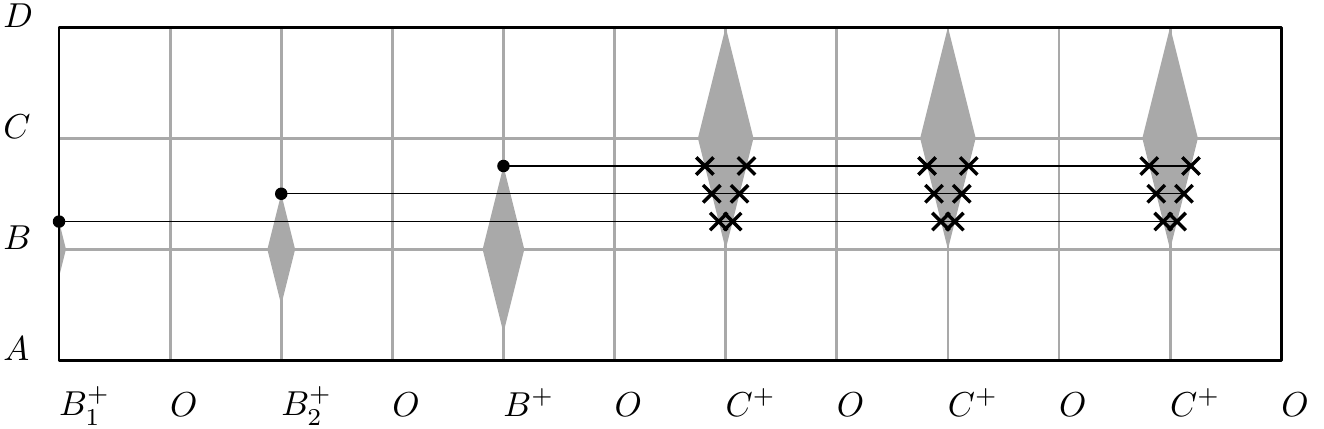}
    \caption{A free space diagram $F_d(T_2,T_1)$ with $2n^2$ critical points.}
    \label{fig:to_kc1_2}
\end{figure}

 In Figure~\ref{fig:to_kc1_2}, on the left, we have the diamonds associated with the antipodal pairs $(B^+_i,B)$. On the right, we have the diamonds associated with the antipodal pairs $(C^+,C)$. Each diamond on the left generates an external critical point for its topmost point (marked with dots). Moreover, all these external critical points have different $y$-coordinates. Each of these distinct $y$-coordinates generates two internal critical points on each diamond on the right (marked with crosses). Therefore, $F_d(T_2,T_1)$ has $2n^2$ internal critical points. Since $F_d(T,T)$ contains $n$ copies of $F_d(T_2,T_1)$ in it, we have that $F_d(T,T)$ contains $\Theta(n^3)$ critical points, as required.

\subsubsection*{Key Component 2: Combining diamonds to form gadgets}

Recall that in our first component, we generate pairs of curves $T_1$ and $T_2$, so that the only regions of non-free space in $F_d(T_2,T_1)$ are small diamonds. We can change $T_1$ and $T_2$ to vary the number, positions and sizes of the diamonds in $F_d(T_2,T_1)$. Our second key component is to position the diamonds in a way that encodes boolean formulas. To help simplify the description of these boolean formulas, we only consider the two sets $(\mathcal X, \mathcal Z)$, making the input a 2OV instance for now.

Our first gadget is an OR gadget and checks if one of two booleans is zero. Our second gadget is an AND gadget and checks if a pair of vectors are orthogonal.

Our OR gadget receives as input two booleans, $X$ and $Z$, and constructs a pair of curves $T_1$ and $T_2$. Let~$s$ and~$t$ be the starting and ending points of~$T_2$, and~$l_s$ and~$l_t$ be the vertical lines corresponding to~$s$ and~$t$. The trajectories $T_1$ and $T_2$ are constructed in such a way that, if $X \cdot Z = 0$ then there is a monotone path from $l_s$ to $l_t$, otherwise, there is no such monotone path.

We use the same definitions of vertices as in the first key component. The pairs $(A^+, A)$, $(B^+,B)$, $(C^+,C)$ and $(D^+,D)$ are antipodes.  Define $B^+_X = O$ if $X = 0$, and $B^+_X = B^+$ otherwise. Similarly, define $D^+_Z = O$ if $Y=0$, and $D^+_Y = D^+$ otherwise.

Now we are ready to construct the curves $T_1$ and $T_2$.

\[
    \begin{array}{rcl}
        T_1 &=& A \circ B \circ C \circ D\\
        T_2 &=& B^+ \circ O \circ C^+ \circ O \circ D^+ \circ O \\
            && \circ A^+ \circ O \circ B^+_X \circ O \circ C^+ \circ O \circ D^+_Z.
    \end{array}
\]

\begin{figure}[ht]
    \centering
    \includegraphics[width=0.7\textwidth]{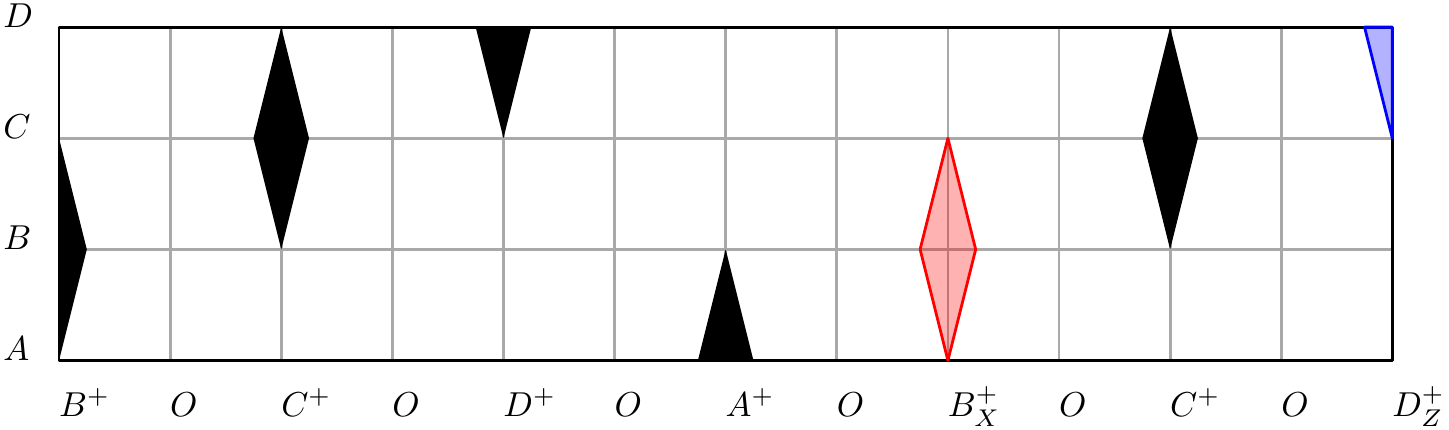}
    \caption{The free space diagram $F_d(T_2,T_1)$ for the OR gadget. The red diamond in column $B_X^+$ disappears if $X=0$, and the blue diamond in column $D_Z^+$ disappears if $Z=0$.}
    \label{fig:to_kc2_1}
\end{figure}

The free space diagram $F_d(T_2,T_1)$ is shown in Figure~\ref{fig:to_kc2_1}. The red diamond in column $B_X^+$ disappears if $X=0$, whereas the blue diamond in column $D_Z^+$ disappears if $Z=0$. If either one is zero, there is a gap for there to be a monotone path from $l_s$ to $l_t$. Otherwise, there is no gap, and no monotone path. This completes the description of the OR gadget.

Our AND gadget receives as input a pair of binary vectors $X = (X[1], X[2], \ldots, X[W])$ and $Z = (Z[1],Z[2],\ldots,Z[W])$, and constructs a pair of trajectories $T_1$ and $T_2$. Let $s$ and $t$ be the starting and ending points of $T_2$, and let $l_s$ and $l_t$ be the vertical lines corresponding to $s$ and $t$. If $X$ and $Z$ are orthogonal, in other words, if $X[h] \cdot Z[h] = 0$ for all $1 \leq h \leq W$, then the maximum number of monotone paths from $l_s$ to $l_t$ is $W$. If $X$ and $Z$ are not orthogonal, in other words, if $X[h] \cdot Z[h] = 1$ for some $1 \leq h \leq W$, then the maximum number of monotone paths from $l_s$ to $l_t$ is $W-1$. Let $r$ be a positive real and $r' > 10r$. Now, let $\phi = \frac {2\pi} {2W+4}$, let $d = ||r \cis \phi - r' \cis \phi||$ and, let $\varepsilon = r' + r - d$. 

\begin{wrapfigure}[5]{r}{0.49\textwidth}
    \centering
    \includegraphics[width=0.45\textwidth]{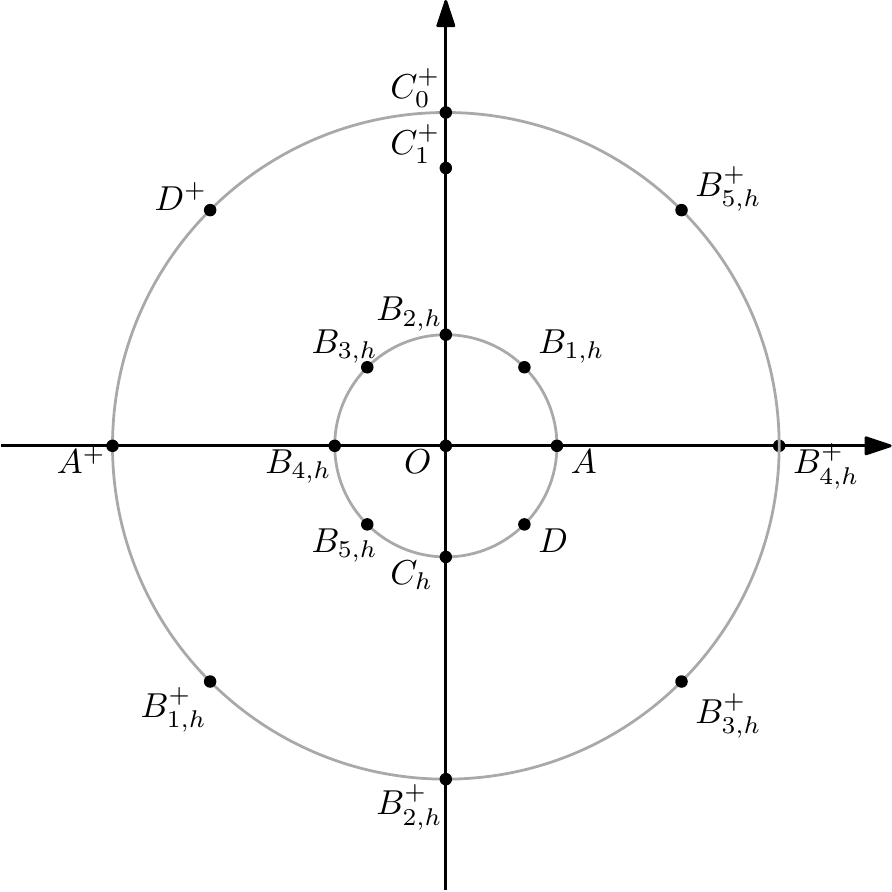}
    \caption{Vertices of $T_1$ and $T_2$, for $W=2$.}
    \label{fig:to_kc2_2}
\end{wrapfigure}

Next, we define the vertices of $T_1$ and $T_2$.
\[
    \begin{array}{rcl}
        A &=& r \cis 0 \\
        B_{u,h} &=& r \cis (u \cdot \phi), \text{ if $u \neq 2h$}\\
        B_{u,h} &=& (r - \frac 2 3 \varepsilon) \cis (u \cdot \phi), \text{ if $u = 2h$}\\
        C_h &=& r \cis ((2W+2) \cdot \phi), \text{ if $Z[h] = 1$}\\
        C_h &=& (r - \frac 2 3 \varepsilon) \cis ((2W+2) \cdot \phi), \text{ if $Z[h] = 0$}\\        
        D &=& r \cis ((2W+3) \cdot \phi) \\
        O &=& 0 \cis 0\\
        A^+ &=& r' \cis \pi \\
        B_{u,0}^+ &=& r' \cis (\pi + u \cdot \phi) \\
        B_{u,1}^+ &=& r' \cis (\pi + u \cdot \phi), \text{ if $u$ is odd}\\
        B_{u,1}^+ &=& r' \cis (\pi + u \cdot \phi), \text{ if $u$ is even and $X[\frac u 2] = 1$}\\        
        B_{u,1}^+ &=& (r' - \frac 2 3 \varepsilon) \cis (\pi + u \cdot \phi) \\&& \quad \text{if $u$ is even and $X[\frac u 2] = 0$}\\
        C_0^+ &=& r' \cis (\pi + (2W+2) \cdot \phi) \\
        C_1^+ &=& (r'-\frac 2 3 \varepsilon) \cis (\pi + (2W+2) \cdot \phi) \\
        D^+ &=& r' \cis (\pi + (2W+3) \cdot \phi) \\        
    \end{array}
\]

Now we can construct $T_1$ and $T_2$.
\[
    \begin{array}{rcl}
        T_1 &=& \bigcirc_{1 \leq h \leq W} \big( A \circ \bigcirc_{1 \leq u \leq 2W+1}(B_{u,h}) \circ C_h \circ D\big) \circ A\\
        T_2 &=& \bigcirc_{1 \leq u \leq 2W+1}(B^+_{u,0} \circ O) \circ C_0^+ \circ O \\
            && \circ A \circ O \circ \bigcirc_{1 \leq u \leq 2W+1}(B^+_{u,1} \circ O) \circ C_1^+ \circ O \\
            && \circ A \circ O \circ \bigcirc_{1 \leq u \leq 2W+1}(B^+_{u,0} \circ O) \circ D^+ \circ O \circ A \circ O \circ C^+_0 \circ O \circ D \circ O
    \end{array}
\]

\begin{figure}[ht]
    \centering
    \includegraphics[width=0.55\textwidth]{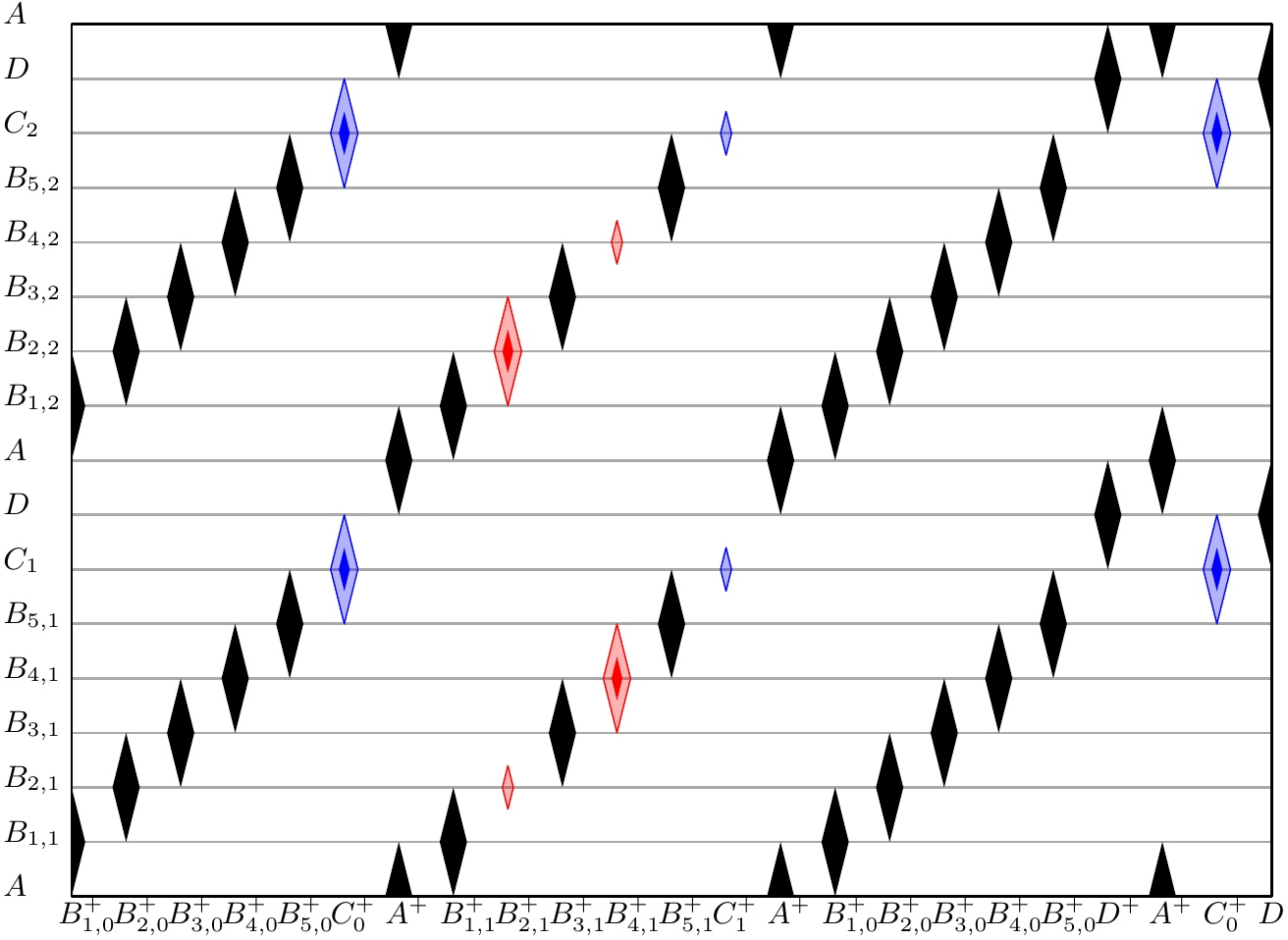}
    \caption{The free space diagram $F_d(T_2,T_1)$ for the AND gadget, for $W=2$.}
    \label{fig:to_kc2_3}
\end{figure}

The free space diagram $F_d(T_2,T_1)$ is shown in Figure~\ref{fig:to_kc2_3}, for $W=2$. The labels for the repeated $O$'s are omitted from the $x$-axis of $F_d(T_2,T_1)$. For $W=2$, as we can see, there are four red diamonds, in columns $B_{2,1}^+$ and $B_{4,1}^+$, and six blue diamonds, in columns $C_0^+$, $C_1^+$ and $C_0^+$. The red diamond in column $B_{2,1}^+$ and row $B_{2,1}$ disappears if and only if $X[1] = 0$. The red diamond in column $B_{4,2}^+$ and row $B_{4,2}$ disappears if and only if $X[2] = 0$. The other red diamonds may shrink, but do not completely disappear. The blue diamond in column $C_1^+$ and row $C_1$ disappears if and only if $Z[1] = 0$. The blue diamond in column $C_1^+$ and row $C_2$ disappears if and only if $Z[2] = 0$. The other blue diamonds may shrink, but do not completely disappear.

The bottom half of the free space diagram is essentially an OR gadget for $X[0]$ and $Z[0]$, and the top half is essentially an OR gadget for $X[1]$ and $Z[1]$. In Figure~\ref{fig:to_kc2_4} left, the vectors $X$ and $Z$ are orthogonal. In this case, there are two monotone paths, one for the OR gadget $X[0] \cdot Z[0] = 0$ and one for the OR gadget $X[1] \cdot Z[1] = 0$. In Figure~\ref{fig:to_kc2_4} right, the vectors $X$ and $Z$ are not orthogonal. In this case, there is a maximum of one monotone path, and this monotone path passes through the ``gap" between the two OR gadgets. 

\begin{figure}[ht]
    \centering
    \includegraphics[width=0.49\textwidth]{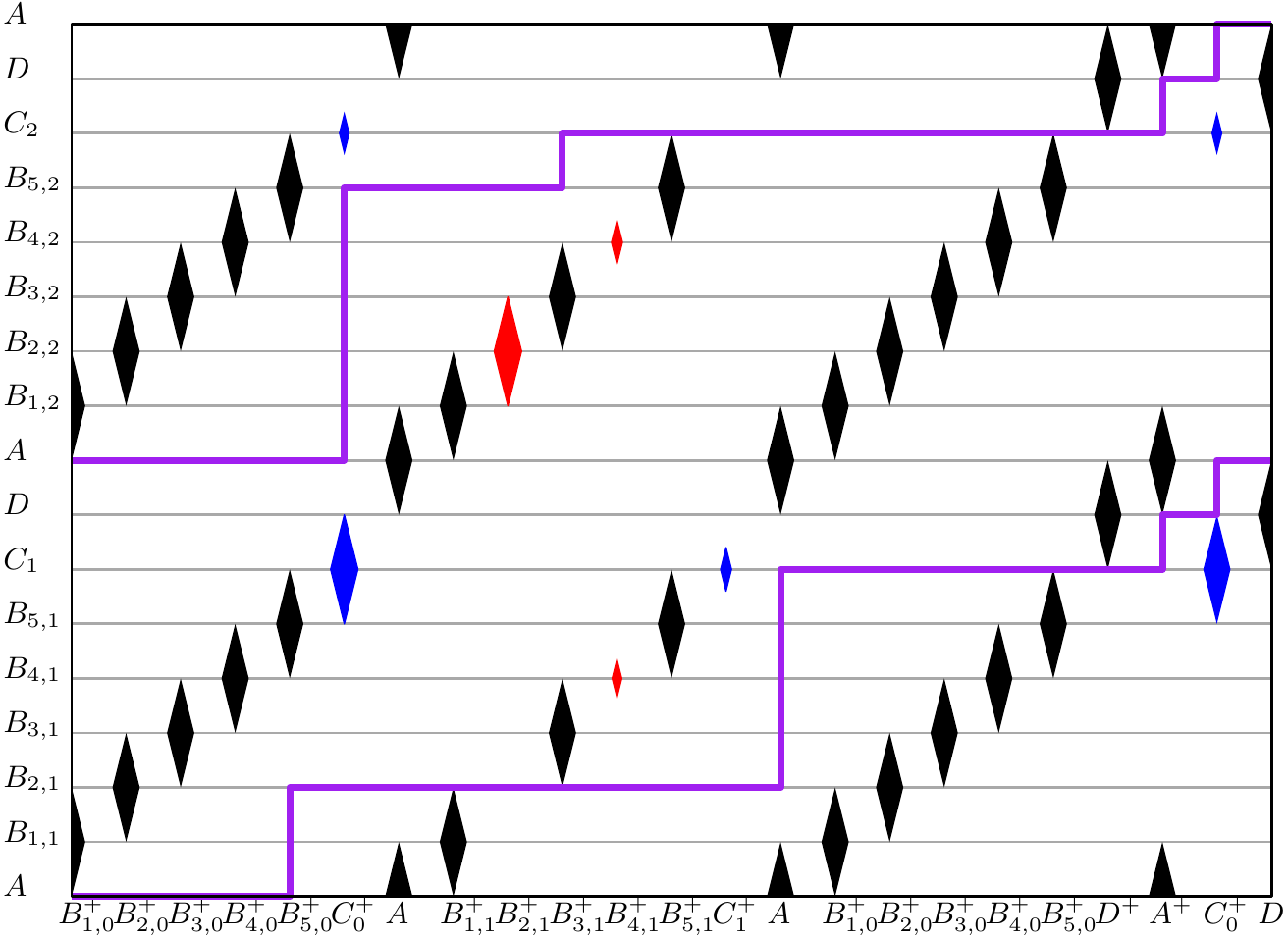}
    \includegraphics[width=0.49\textwidth]{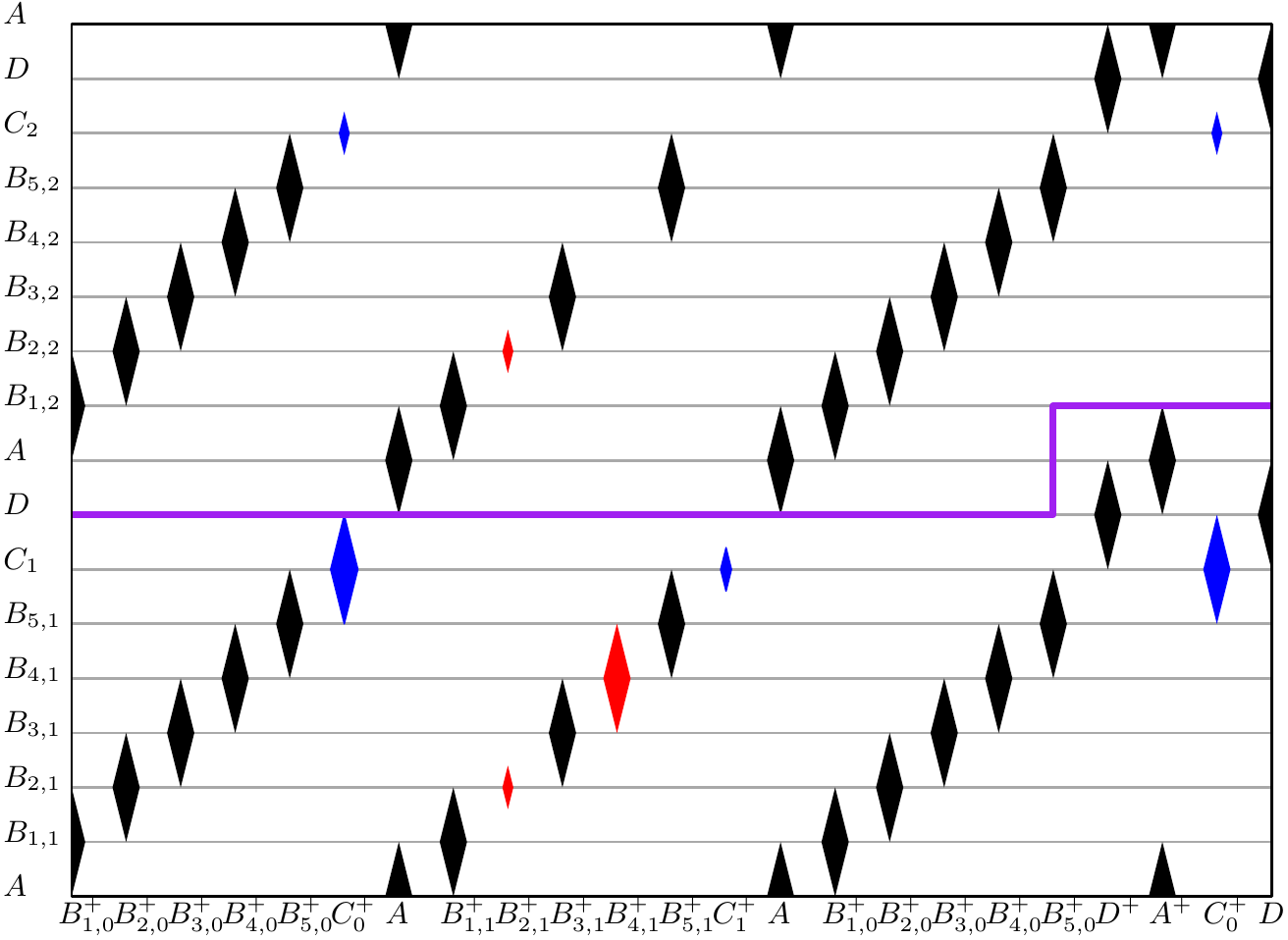}
    \caption{For $W=2$, if $X$ and $Z$ are orthogonal, there are $W$ paths (left), whereas if $X$ and $Z$ are not orthogonal, there are $W-1$ paths (right).}
    \label{fig:to_kc2_4}
\end{figure}

For general values of $W$, the AND gadget is a stack of OR gadgets for $X[h]$ and $Z[h]$ for $1 \leq h \leq W$. If $X[h] \cdot Z[h] = 0$ for all $1 \leq h \leq W$, we obtain $W$ monotone paths, one per OR gadget. Otherwise, we obtain a maximum of $W-1$ monotone paths, one for each gap between two consecutive OR gadgets. This completes the description of the AND gadget.

\subsubsection*{Key Component 3: Combining gadgets to form the full reduction}

The gadgets in our full construction are inspired by the gadgets in our second key component. However, the gadgets in our full construction are more sophisticated in two important ways. 

First, in the gadgets constructed so far, we only consider a single pair of vertical lines, that is, the pair of vertical lines that correspond to the start and end points of $T_2$. In our full construction, we consider all vertical lines that start and end at internal critical points. In particular, we consider $n^2$ pairs of vertical lines that correspond to a pair of integers $(i,j)$ where $1 \leq i,j \leq n$. Each of these vertical lines passes through $nW$ internal critical points. The $nW$ internal critical points correspond to pairs of integers $(k,h)$ where $1 \leq k \leq n$ and $1 \leq h \leq W$. 

Second, in the gadgets we have constructed so far, we only consider two sets of binary vectors. In our full construction, we consider three sets of binary vectors, $\mathcal X$, $\mathcal Y$ and $\mathcal Z$. 

With these two key differences in mind, we can describe the gadgets in our full construction. We receive as input a 3OV instance, in other words, we are given three sets $\mathcal X$, $\mathcal Y$ and $\mathcal Z$, each containing $n$ binary vectors of length $W$. 

First, we describe the OR gadget that appears in our full construction. There will be $nW$ copies of the OR gadget. Given $1 \leq k \leq n$ and $1 \leq h \leq W$, the OR gadget for the pair $(k,h)$ satisfies the following property for all $1 \leq i,j \leq n$: there are two monotone paths from $l_s$ to $l_t$ if $X_i[h] \cdot Y_j[h] \cdot Z_k[h] = 0$, whereas there is a maximum of one monotone path from $l_s$ to $l_t$ if $X_i[h] \cdot Y_j[h] \cdot Z_k[h] = 1$, where $l_s$ and $l_t$ are the vertical lines corresponding to the pair $(i,j)$. 

In Figure~\ref{fig:to_kc3_1}, there is one monotone path from $l_s$ to $l_t$, and this is the only monotone path in the OR gadget in the case that $X_i[h] \cdot Y_j[h] \cdot Z_k[h] = 1$. On the other hand, if $X_i[h] \cdot Y_j[h] \cdot Z_k[h] = 0$, then there are two monotone paths in the OR gadget.

\begin{figure}[ht]
    \centering
    \includegraphics[width=\textwidth]{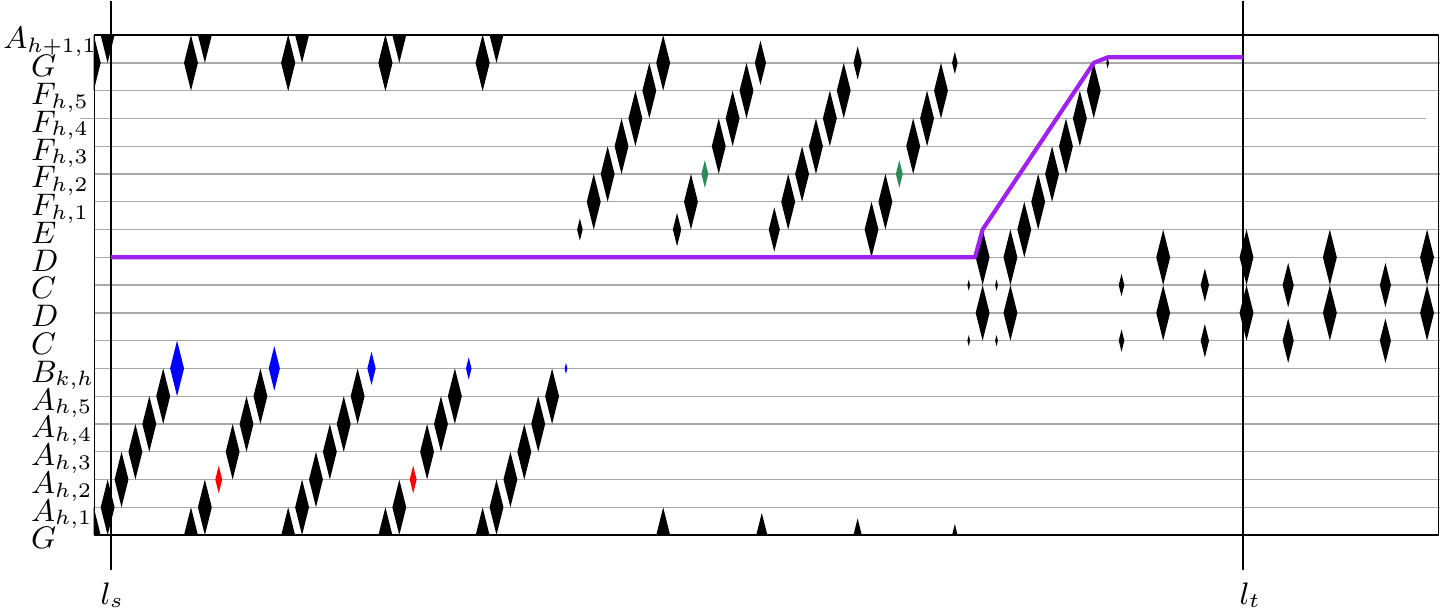}
    \caption{The OR gadget for booleans $X_i[h]$, $Y_j[h]$ and $Z_k[h]$.}
    \label{fig:to_kc3_1}
\end{figure}

In Figure~\ref{fig:to_kc3_X}, we show the two monotone paths from $l_s$ to $l_t$ in three separate cases: for $X_i[h] = 0$, $Y_j[h] = 0$ and $Z_k[h] = 0$. We briefly describe the behaviour in these three cases. If $X_i[h] = 0$, a red diamond in the bottom right disappears and one of the two monotone paths passes through this gap. If $Y_j[h] = 0$, a green diamond disappears and one of the two monotone paths passes through this gap. Finally, if $Z_k[h]=0$, all blue diamonds shrink in size, allowing the lower monotone path to have a much smaller maximum $y$-coordinate.

\begin{figure}[tb]
    \centering
    \includegraphics[width=0.9\textwidth]{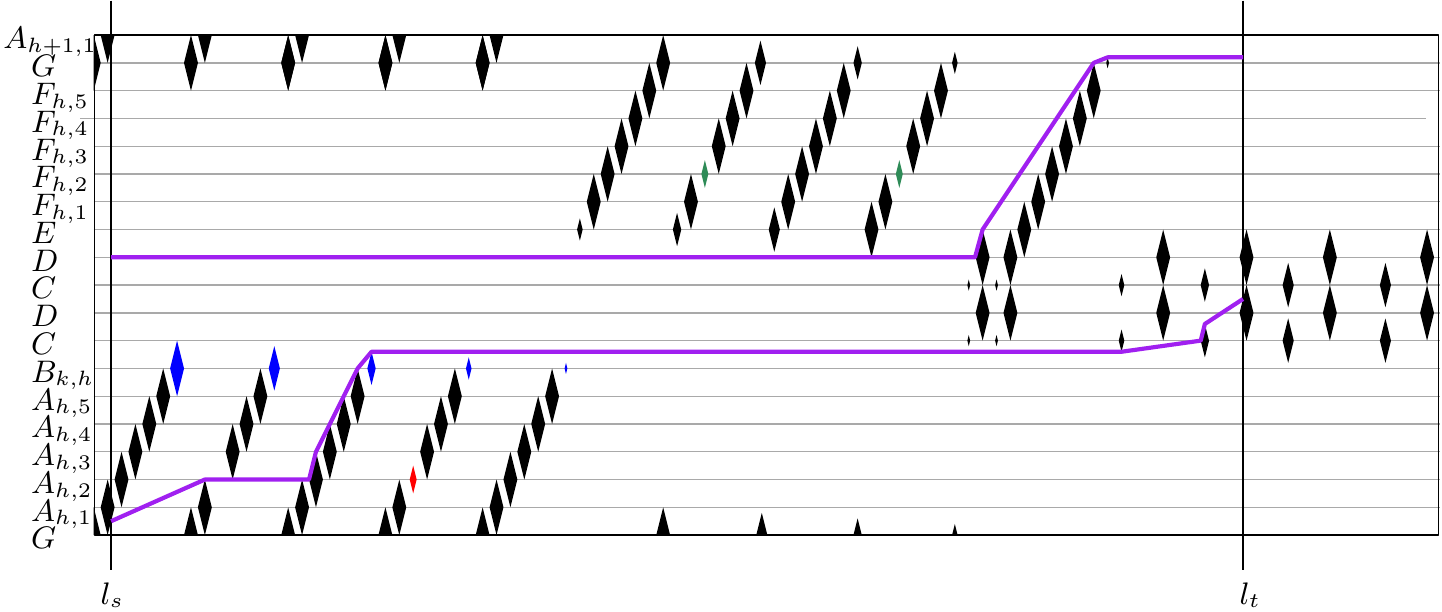}
    \includegraphics[width=0.9\textwidth]{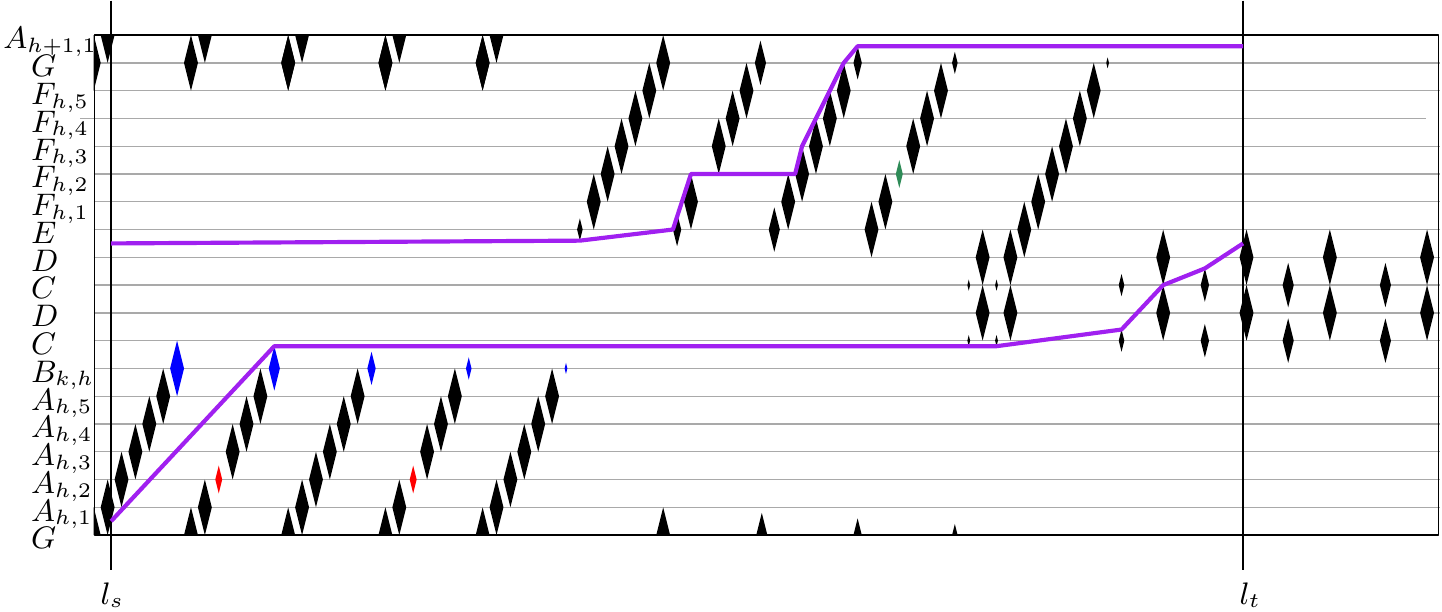}
    \includegraphics[width=0.9\textwidth]{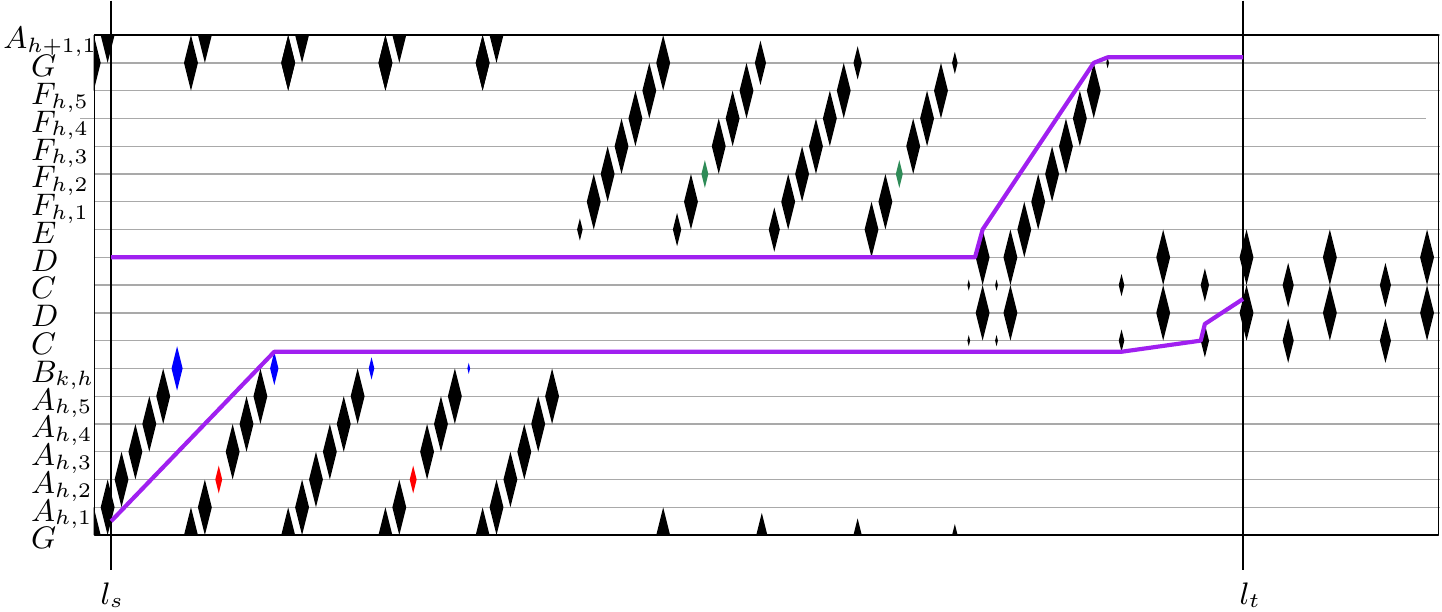}
    \caption{The two monotone paths from $l_s$ to $l_t$ in the cases where $X_i[h]=0$, $Y_j[h] = 0$ and $Z_k[h]=0$, respectively}
    \label{fig:to_kc3_X}
\end{figure}

It is worth noting the connection between these OR gadgets and internal critical points. In all three cases, the starting point of the first monotone path is an internal critical point. If $Y_j[h]=0$, all starting and ending points of the two monotone paths are either internal critical points, or share a $y$-coordinate with an internal critical point. 

Similarly to the AND gadget in the second key component, we stack $W$ copies of the OR gadget on top of each other. Each stack of OR gadgets checks if a triple of vectors $(X_i, Y_j, Z_k)$ are orthogonal. There are $2W+1$ monotone paths from~$l_s$ to~$l_t$ if $X_i[h] \cdot Y_j[h] \cdot Z_k[h] = 0$ for all $1 \leq h \leq W$, whereas there are $2W$ paths from~$l_s$ to~$l_t$ if $X_i[h] \cdot Y_j[h] \cdot Z_k[h] = 1$ for some $1 \leq h \leq W$. 

\FloatBarrier

Finally, we combine these AND gadgets to form the full construction. We give a high level overview. We stack, one on top of another, $n$ copies of the AND gadgets. We add non-free space between consecutive AND gadgets, so that monotone paths in one AND gadget cannot interact with monotone paths in another AND gadget. We let the two curves that generate this free space be $T_1$ and $T_2$. We join $T_1$ and~$T_2$ end to end to form the final curve and set $m = 2nW+2$. The intuition behind setting $m=2nW+2$ is that, as there are $n$ copies of the AND gadget, if there are $m-1$ monotone paths, then by the pigeonhole principle, we must have one AND gadget with $2W+1$ monotone paths, which implies that $X_i$, $Y_j$ and $Z_k$ are orthogonal for some triple $(i,j,k)$. Otherwise, all AND gadgets have $2W$ monotone paths, so for all triples $(i,j,k)$, there is some $h$ so that $X_i[h] \cdot Y_j[h] \cdot Z_k[h] = 1$. Putting this all together yields the following theorem.

\begin{restatable}{theorem}{theoremfour}
\label{theoremfour}
There is no $O(n^{3-\varepsilon})$ time algorithm for \problemtwo under the continuous Fr\'echet distance, for any $\varepsilon>0$, unless SETH fails.
\end{restatable}

\section{Discrete Fr\'echet Distance}
\label{sec:dfd}

The main theorem that we will prove in this section is:

\theoremone*

We will be using the discrete free space diagram extensively in our algorithm. Recall that for the discrete Fr\'echet distance, the free space diagram $F_d(T,T)$ consists of $n^2$ grid points. A grid point $(x,y)$ is free if vertices $x$ and $y$ of the trajectory $T$ are within distance $d$ of one another. A monotone path is a sequence of free grid points where a grid point $(x,y)$ is followed by $(x+1,y)$, $(x,y+1)$, or $(x+1,y+1)$. See Appendix~\ref{sec:appendix_frechet_freespace} for a formal discussion of the discrete free space diagram.

In Section~\ref{sec:dfd.dfs}, we show how to solve Subproblem~\ref{def:p3} in $O(nl)$ time, where $l = t-s$, under the discrete Fr\'echet distance. Then, in Section~\ref{sec:dfd.jump}, we show how to extend this to an algorithm that solves \problemtwo in $O(n^2 \log n)$ time, under the discrete Fr\'echet distance.

\subsection{Subproblem~\ref{def:p3} under the discrete Fr\'echet distance}
\label{sec:dfd.dfs}

We inductively define our algorithm for Subproblem~\ref{def:p3} under the discrete Fr\'echet distance. In the base case, we compute the monotone path $P_1$ from $l_s$ to $l_t$ that minimises its maximum $y$-coordinate. In the inductive case, we compute the monotone path $P_i$ from $l_s$ to $l_t$ that minimises its maximum $y$-coordinate, under the condition that $P_i$ does not overlap in $y$-coordinate with $P_1, P_2, \ldots, P_{i-1}$ or the $y$-interval corresponding to the reference subtrajectory.

To compute each of the paths $P_i$, we begin by picking its starting point of $l_s$. Initially, we mark all free grid points as valid, and all non-free grid points as invalid. For $i=1$, we pick the lowest valid grid point on $l_s$ as the initial starting point. For $i>1$, we pick the lowest valid grid point on $l_s$ that has $y$-coordinate at least the maximum $y$-coordinate on $P_{i-1}$. From this initial starting point we begin a greedy depth first search to compute $P_i$.

Any grid point has up to three neighbouring grid points to explore: $(x+1,y), (x+1,y+1)$ and $(x,y+1)$, as long as these grid points are valid. Similar to the standard depth first search, we explore each branch of the search tree as far as possible first before backtracking. The greedy aspect of our depth first search is to first explore the neighbour $(x+1,y)$, then $(x+1,y+1)$, and finally $(x,y+1)$. The intuition is that we would like to minimise the $y$-coordinate of our search. If we are forced to backtrack, i.e. if all neighbours lead to dead ends, we mark the grid point as invalid and backtrack to its original parent. This is the only way a free grid point becomes invalid. Once a grid point is marked as invalid, it is never reverted back to valid, and can never be used in any monotone paths. In Figure~\ref{fig:dfd_1}, backtracking occurred on the red paths, so these cells will be marked as invalid.

\begin{figure}[ht]
    \centering
    \includegraphics{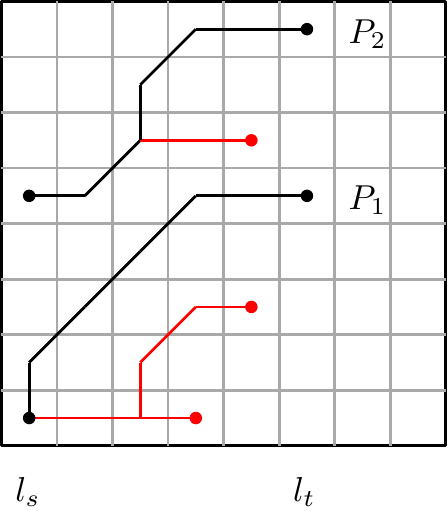}
    \caption{Our greedy depth first search algorithm searches for the monotone paths $P_i$. The red paths are dead ends, so we backtrack and then continue searching.}
    \label{fig:dfd_1}
\end{figure}

The greedy depth first search halts if one of the following three conditions are met. First, if our algorithm reaches $l_t$, we have computed $P_i$, and therefore we halt the search. We continue by computing the next monotone path $P_{i+1}$. Second, if our algorithm backtracks to our initial starting point, then our algorithm cannot find a monotone path starting at this point, and therefore we halt the search. We continue by trying to compute $P_i$, but starting from a higher valid grid point on $l_s$. Third, if our algorithm moves to a grid point with $y$-coordinate strictly between~$s$ and~$t$, then the monotone path intersects the reference trajectory at more than one point, which contradicts the conditions of Subproblem~\ref{def:p3}, and therefore we halt the search. We continue by trying to compute $P_i$ again, but we start at the lowest valid grid point on $l_s$ that has $y$-coordinate at least~$t$.

If our algorithm computes $P_1, P_2, \ldots, P_{m-1}$, then our set of $m-1$ monotone paths are returned. Otherwise, if all valid grid points on $l_s$ are exhausted, then our algorithm returns that there is no set of $m-1$ monotone paths. This completes the statement of our algorithm. 

Next we argue the correctness of our algorithm. We make three observations. Each observation follows immediately from the greedy depth first search and the following ordering on the outgoing edges: first $(x+1,y)$, then $(x+1,y+1)$, and finally $(x,y+1)$. For our third observation, we define $l_{\geq g}$ be the set of grid points with the same $x$-coordinate as $g$ and have $y$-coordinate greater than or equal to the $y$-coordinate of $g$.

\begin{observation}
\label{lem:dfd.dfs.always_valid}
Let $P$ be a monotone path from $l_s$ to $l_t$. Throughout the execution of the greedy depth first search, all grid points on $P$ will remain valid.
\end{observation}

\begin{observation}
\label{lem:dfd.dfs.no_path}
Suppose our algorithm starts at a grid point on $l_s$ and does not find a monotone path to $l_t$. Then there is no monotone path from $l_s$ to $l_t$ starting at that grid point.
\end{observation}

\begin{observation}
\label{lem:dfd.dfs.ell_geq}
Let $P$ be a monotone path from $l_s$ to $l_t$ computed by our algorithm. Let the initial starting point of $P$ be $g$ and the final grid point of $P$ be $r$. Then any monotone path that starts on~$l_{\geq g}$ and ends on $l_t$ must end on $l_{\geq r}$. 
\end{observation}

The third observation formalises the intuition that if a monotone path is found, then a lower monotone path cannot exist. With these three observations in mind, we are now ready to prove the correctness of our algorithm.

\begin{lemma}
\label{lem:dfd.dfs.correctness}
There exist $m-1$ monotone paths satisfying the conditions of Subproblem~\ref{def:p3} under the discrete Fr\'echet distance if and only if our algorithm returns a set of $m-1$ monotone paths.
\end{lemma}

\begin{proof}
Suppose our algorithm returns a set of $m-1$ monotone paths. It is straightforward to check from the definition of our algorithm that our set of monotone paths all start on $l_s$, all end on $l_t$, are distinct, and overlap in at most one $y$-coordinate. Therefore, our $m-1$ monotone paths satisfy the conditions of Subproblem~\ref{def:p3}.

Next, we prove the converse. We assume that there exist monotone paths $Q_1, Q_2, \ldots, Q_{m-1}$ that satisfy the conditions of Subproblem~\ref{def:p3}. We will prove that our algorithm computes a set of monotone paths $P_1, P_2, \ldots, P_{m-1}$ that also satisfy the conditions of Subproblem~\ref{def:p3}.

We prove by induction that, for all $1 \leq i \leq m-1$,  our algorithm computes a monotone path~$P_i$ so that the maximum $y$-coordinate of $P_i$ is at most the maximum $y$-coordinate of $Q_i$. We will focus on the inductive case, since the base case follows similarly. Our inductive hypothesis implies that the maximum $y$-coordinate of $P_{i-1}$ is at most the maximum $y$-coordinate of $Q_{i-1}$. So the maximum $y$-coordinate of $P_{i-1}$ is at most the minimum $y$-coordinate of $Q_i$. 

Next, our algorithm attempts to compute $P_i$. It starts at a $y$-coordinate that is at most the minimum $y$-coordinate of $Q_i$. By the contrapositive of Lemma~\ref{lem:dfd.dfs.no_path}, our algorithm computes a monotone path when considering the starting point on $Q_i$, if not earlier. Hence, the minimum $y$-coordinate of $P_i$ is at most the minimum $y$-coordinate of $Q_i$. By Lemma~\ref{lem:dfd.dfs.ell_geq}, the maximum $y$-coordinate of $P_i$ is at most the maximum $y$-coordinate of~$Q_i$, completing the induction and the proof of the lemma. So our algorithm computes a set of $m-1$ monotone paths that satisfy the conditions of Subproblem~\ref{def:p3}.
\end{proof}

Finally, we analyse the running time of our algorithm. Whenever our algorithm visits a grid point, a constant number of operations are performed. The operations involve checking if its three neighbours are valid, moving to a neighbour, or backtracking. Each operation takes constant time.

It suffices to count the number of grid points visited by our algorithm. There are $O(nl)$ grid points in total between the vertical lines $l_s$ and $l_t$ where $l = t-s$. When computing a monotone path, for example when computing $P_i$ we use a depth first search. We never visit the same grid point more than once when we compute $P_i$. Hence, a grid point can only be revisited when computing different monotone paths, for example, when computing $P_i$ and $P_j$. However, this cannot happen very often, as we show next. 

\begin{lemma}
\label{lem:dfd.dfs.rt1}
There are at most $2ml$ instances where $P_j$ revisits a grid point that has previously been visited by $P_i$ for some $i < j$.
\end{lemma}

\begin{proof}
First we prove that $j \leq i + 2$. Then we use this bound to show that there are at most $2ml$ revisiting instances, as claimed.

Suppose for the sake of contradiction that a grid point is visited by $P_i$ and $P_j$, and $j > i + 2$. So $P_j$ and $P_i$ must share a $y$-coordinate at the shared grid point. But the monotone paths $P_{i+1}$ and $P_{j-1}$ have $y$-coordinates that are at least the maximum $y$-coordinate of $P_i$ and at most the minimum $y$-coordinate of $P_j$. Therefore, $P_{i+1}$ and $P_{j-1}$ must be horizontal paths. But these paths are no longer unique, which is a contradiction. This proves that $j \leq i + 2$. 

There are at most~$2m$ pairs $(i,j)$ where $j - i \leq 2$, and for each such pair $(i,j)$ there is at most one $y$-coordinate, i.e. $l$ cells, where $P_i$ and $P_j$ can visit the same cell. In total, there are at most~$2ml$ cells where both $P_i$ and $P_j$ can visit, for some pair $(i,j)$.
\end{proof}

There are $O(nl)$ initial visits to grid points and $O(ml)$ revisits, where $l = t-s$. This yields:

\begin{theorem}
\label{thm:dfd.dfs}
There is an $O(nl)$ time algorithm that solves Subproblem~\ref{def:p3} under the discrete Fr\'echet distance.
\end{theorem}

\subsection{\problemtwo under the discrete Fr\'echet distance}
\label{sec:dfd.jump}

Similar to the previous algorithm of Buchin~\etal~\cite{DBLP:journals/ijcga/BuchinBGLL11}, our algorithm for \problemtwo under the discrete Fr\'echet distance involves solving~$O(n)$ instances of Subproblem~\ref{def:p3} with a sweepline approach. We maintain a link-cut data structure~\cite{DBLP:journals/jcss/SleatorT83}
 that allows us to reuse monotone paths during our sweep. Our data structure maintains a set of rooted trees, where the nodes of the trees are grid points in the free space diagram.

\begin{fact}[\cite{DBLP:journals/jcss/SleatorT83}]
\label{lem:linkcut}
A link-cut tree maintains a set of rooted trees, and offers the following four operations. Each operation can be performed in $O(\log n)$ amortised time.
\begin{itemize}[noitemsep]
    \item Add a tree consisting of a single node,
    \item Attach a node (and its subtree) to another node as its child,
    \item Disconnect a node (and its subtree) from its current tree,
    \item Given a node, find the root of its current tree.
\end{itemize}
\end{fact}

The invariant maintained by our data structure is that there is a monotone path from any node to the root of its current tree. We will first describe our algorithm and how we update our data structure. Then we will prove our invariant in Lemma~\ref{lem:dfd.jump.valid_path}.

We use a sweepline approach, starting with $s=0$ and incrementing $s$. For each $s$, we let~$t$ be the final vertex of the shortest subtrajectory that starts at~$s$ and has length~$\geq \ell$. For each pair~$(s,t)$, we decide whether there exists a set of $m-1$ monotone paths between~$l_s$ and~$l_t$ that do not overlap in $y$-coordinate. We perform a modified version of the greedy depth first search in Section~\ref{sec:dfd.dfs}. Our modifications include updating the link-cut data structure whenever we explore a node, and querying the link-cut data structure to reuse paths.

Whenever our greedy depth first search moves from a current grid point, which we call $g_c$, to an outgoing neighbour, which we call $g_n$, we update our link-cut data structure. We link $g_c$ as a child of $g_n$, thus attaching the subtree rooted at $g_c$ to the tree containing $g_n$. Under the algorithm described in Section~\ref{sec:dfd.dfs}, we would continue the depth first search from the neighbour $g_n$. However, we already know that there is a monotone path from $g_n$ to the root of its current link-cut tree. Hence, we set the new current node, $g_{c'}$, to be the root of the tree containing $g_n$ and continue our greedy depth first search from there. See Figure~\ref{fig:dfd_2}, left.

\begin{figure}[ht]
    \centering
    \includegraphics{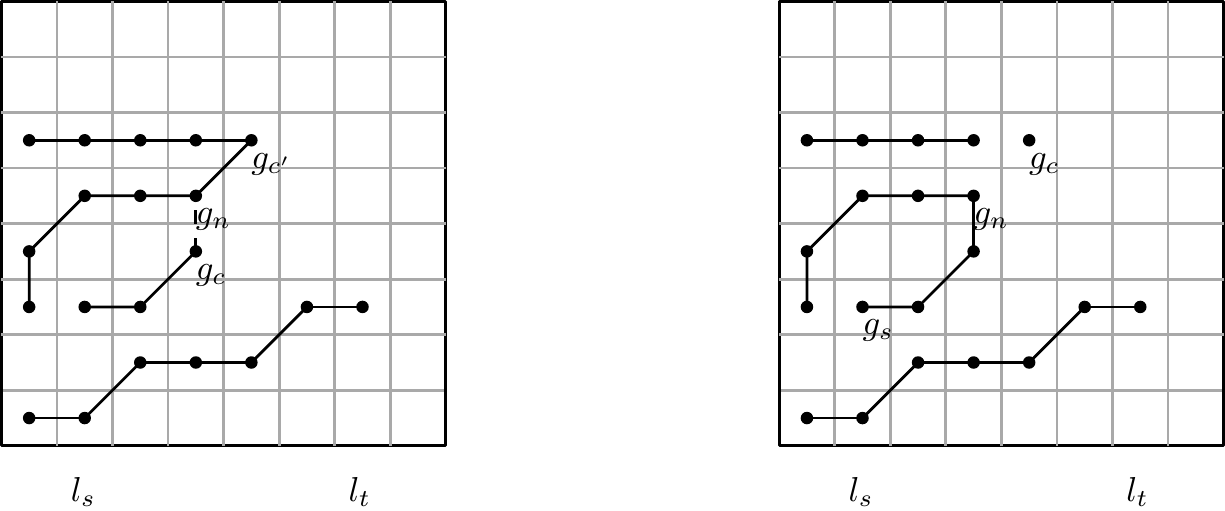}
    \caption{Adding a link from $g_c$ to $g_n$ when exploring a neighbour (left), and removing links from~$g_c$ to its children when backtracking (right).}
    \label{fig:dfd_2}
\end{figure}

Whenever our greedy depth first search backtracks from a current grid point, which we call $g_c$, we also update our link-cut data structure. An invariant maintained by our algorithm is that the current grid point is always the root of its link-cut tree. Hence, its valid incoming neighbours are its children in the link-cut data structure. We disconnect $g_c$ from each of its children. We backtrack to the incoming neighbour $g_n$ with the following property. If we are currently searching for the monotone path $P_i$, and the initial node of $P_i$ on $l_s$ is $g_s$, then we choose $g_n$ to be the root of the tree containing $g_s$. See Figure~\ref{fig:dfd_2}, right.

The greedy depth first search halts if any of the same three conditions as described in Section~\ref{sec:dfd.dfs} are met. This completes the description of our algorithm. Next, we argue its correctness.

\begin{lemma}
\label{lem:dfd.jump.valid_path}
Let~$g$ be a grid point in the free space diagram and~$r$ be a grid point corresponding to the root of the link-cut tree containing~$g$. Then there is a monotone path from~$g$ to~$r$. 
\end{lemma}

\begin{proof}
Each link in the link-cut data structure is from a grid point $(x,y)$ and to a grid point $(x,y+1)$, $(x+1,y+1)$ or $(x+1,y)$. Since~$r$ is an ancestor of $g$ in the link-cut tree data structure, there must be a monotone path from $g$ to $r$.
\end{proof}

Next, we prove an analogous lemma to Lemma~\ref{lem:dfd.dfs.ell_geq}. Recall that $l_{\geq g}$ is the set of grid points with the same $x$-coordinate as $g$ and have $y$-coordinate greater than or equal to the $y$-coordinate of $g$. 

\begin{lemma}
\label{lem:dfd.jump.ell_geq}
For a fixed pair $(s,t)$, suppose $g$ is a grid point on $l_s$ and suppose that its root $r$ lies on $l_t$. Then any monotone path starting on $l_{\geq g}$ that ends on $l_t$ must end on $l_{\geq r}$. 
\end{lemma}

\begin{proof}
By Lemma~\ref{lem:dfd.jump.valid_path}, there exists a monotone path $P$ from $g$ to $r$. The remainder of the proof is identical to the proof of Lemma~\ref{lem:dfd.dfs.ell_geq}, except we replace the claim that our algorithm prefers to search its lower neighbours with our algorithm prefers to link to its lower neighbours.
\end{proof}

Now we show that our algorithm solves \problemtwo.

\begin{lemma}
\label{lem:dfd.jump.correctness}
There exist $m-1$ monotone paths satisfying the conditions of \problemtwo under the discrete Fr\'echet distance if and only if our algorithm returns a set of $m-1$ monotone paths.
\end{lemma}

\begin{proof}
First we prove the if direction. If our algorithm returns a set of $m-1$ monotone paths from $l_s$ to $l_t$, then the conditions of Subproblem~\ref{def:p3} are satisfied for this fixed pair of vertices $(s,t)$. With the subtrajectory from $s$ to $t$ acting as the reference trajectory, and the $m-1$ monotone paths acting as the other $m-1$ subtrajectories, this cluster of subtrajectories satisfies the conditions of \problemtwo. 

Next we prove the only if direction. Suppose there exists a set of $m-1$ monotone paths that satisfy \problemtwo. Let the reference subtrajectory start at $s$ and end at $u$. Then the $m-1$ monotone paths between $l_s$ and $l_u$ corresponding to the $m-1$ subtrajectories that satisfy the conditions of Subproblem~\ref{def:p3}. Let the shortest subtrajectory starting at $s$ with length $\geq \ell$ end at the vertex $t$. Shorten the $m-1$ monotone paths to be between $l_s$ and $l_t$. There exists a set of $m-1$ monotone paths from $l_s$ to $l_t$ that satisfy the conditions of Subproblem~\ref{def:p3}. The remainder of the proof is identical to the proof of Lemma~\ref{lem:dfd.dfs.correctness}, except we replace our reference to Lemma~\ref{lem:dfd.dfs.ell_geq} with a reference to Lemma~\ref{lem:dfd.jump.ell_geq}. Therefore, our algorithm returns a set of $m-1$ monotone paths from $l_s$ to $l_t$, as required.
\end{proof}

Finally, we analyse the overall running time of our algorithm to obtain the main theorem of Section~\ref{sec:dfd}. The running time is dominated by the greedy depth first search and updating the link-cut data structure.

\theoremone*

\begin{proof}
We will bound the number of steps in the greedy depth first search, and hence the number of updates made on the data structure. Since each grid point has degree at most three, there are at most $3n^2$ pairs of vertices between which there could be a link. At each search step, our algorithm either moves to a neighbour, in which case a link is added, or our algorithm backtracks, in which case a link is removed. Once a link is removed it cannot be added again. Hence, there are at most $O(n^2)$ updates made on the data structure, being either links or cuts, and each update takes $O(\log n)$ amortized time. Hence, the overall running time is $O(n^2 \log n)$. 
\end{proof}

By Theorem~\ref{theorem:bringmann_reduction}, there is no $O(n^{2-\varepsilon})$ time algorithm for \problemtwo under the discrete Fr\'echet distance, for any $\varepsilon > 0$, assuming SETH. Therefore, our algorithm is almost tight, unless SETH fails.

\section{Continuous Fr\'echet Distance}
\label{sec:cfd}

The main theorem that we will prove in this section is the following:

\theoremtwo*

We will be using the continuous free space diagram extensively in this section. Recall that for the continuous Fr\'echet distance, the free space diagram $F_d(T,T)$ consists of $n^2$ cells. The free space within a single cell is the intersection of an ellipse with the cell. A monotone path is a continuous monotone path in the free space. For each cell, we define its critical points to be the intersection of the boundary of the free space with the boundary of the cell. A cell corner in free space is considered a critical point. There are at most eight critical points per cell. See Appendix~\ref{sec:appendix_frechet_freespace} for a formal discussion of the continuous free space diagram.

In Section~\ref{sec:cfd.graph} we provide a modified version of the algorithm by Alt and Godau~\cite{DBLP:journals/ijcga/AltG95} that decides if the continuous Fr\'echet distance between two trajectories is at most~$d$. In Section~\ref{sec:cfd.jump} we extend this algorithm to solve \problemtwo under the continuous Fr\'echet distance in $O(n^2 \log^2 n)$ time.

\subsection{The continuous Fr\'echet distance decision problem}
\label{sec:cfd.graph}

The problem we focus on in this section is to decide if the continuous Fr\'echet distance between a pair of trajectories is at most~$d$. We let the complexities of our two trajectories be $n_1$ and $n_2$. Note that the complexities are usually denoted with $m$ and $n$, however, we use $n_1$ and $n_2$ to avoid confusion with the size of the subtrajectory cluster. 

Similar to the original algorithm by Alt and Godau~\cite{DBLP:journals/ijcga/AltG95}, we decide whether there is a monotone path from the bottom left to the top right corner of the free space diagram. The running time of original algorithm requires $O(n_1n_2)$ time~\cite{DBLP:journals/ijcga/AltG95}, whereas ours requires $O(n_1 n_2 \log (n_1 + n_2))$ time. The original algorithm computes reachability intervals, which are horizontal or vertical propagations of critical points. We avoid computing reachability intervals. Instead, our algorithm decomposes long monotone paths into shorter ones, which we call basic monotone paths.

Recall that a critical point in the continuous free space diagram is either a cell corner, or the intersection point of a cell boundary with an elliptical boundary between free space and non-free space. For a detailed discussion of the free space diagram see Appendix~\ref{sec:appendix_frechet_freespace}. 

\begin{definition}
\label{def:cfd.graph.basic_monotone_path}
A basic monotone path is a monotone path that is contained entirely in a single row (resp. column) of the free space diagram, starts at a critical point on a vertical (resp. horizontal) cell boundary, and ends on a point on a horizontal (resp. vertical) cell boundary. See Figure~\ref{fig:cts_basic_monotone_path}.
\end{definition}

\begin{figure}[ht]
    \centering
    \includegraphics{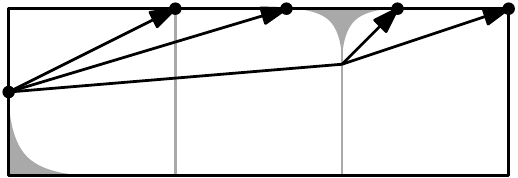}
    \caption{Example of basic monotone paths in a row.}
    \label{fig:cts_basic_monotone_path}
\end{figure}

The next lemma decomposes a monotone path from the bottom left to the top right corner into basic monotone paths. 

\begin{lemma}
\label{lem:cfd.graph.c1}
Given a pair of critical points $a$ and $b$ in the free space diagram, there is a monotone path from $a$ to $b$ if and only if there is a sequence of critical points $p_0, p_1, p_2, \ldots, p_k$ so that $p_0=a$, $p_k=b$, and there is a basic monotone path from $p_i$ to $p_{i+1}$ for every $i = 0,1,\ldots,k-1$.
\end{lemma}

\begin{proof}
For the ``if'' direction, there is a monotone path from $p_i$ to $p_{i+1}$ for all $i=0,1,\ldots,k-1$. Concatenating these paths yields a monotone path from $a$ to $b$.

For the ``only if'' direction, suppose there is a monotone path~$Q$ starts at~$a$ and ends at~$b$. Define $q_0 = a$. We define~$q_i$ inductively for $i \geq 1$. If~$q_i$ is on a vertical (resp. horizontal) boundary, then we define~$q_{i+1}$ as the first intersection of~$Q$ with a horizontal (resp. vertical) boundary occurring after~$q_i$. Between $q_i$ and $q_{i+1}$ there is a monotone path, and the sequence alternates between being on horizontal and vertical boundaries. Eventually we have $q_k = b$ for some $k$.

The monotone path from $q_i$ to $q_{i+1}$ may not be basic if~$q_i$ and~$q_{i+1}$ are not critical points. For $1 \leq i \leq k-1$, if $q_i$ is on a vertical (resp. horizontal) boundary, we define~$p_i$ to be the critical point below (resp. left of)~$q_i$ and on the same cell boundary segment as~$q_i$. Then $p_i$ is either the lowest free point on a vertical boundary, or the leftmost free point on a horizontal boundary. Finally, we set $p_0 = q_0$ and $p_k = q_k$. This completes the construction of the sequence of critical points $p_0, p_1, p_2, \ldots, p_k$. It suffices to show that there is a basic monotone path from~$p_i$ to~$p_{i+1}$. See Figure~\ref{fig:cts_decompose_basic_monotone}.

\begin{figure}[ht]
    \centering
    \includegraphics{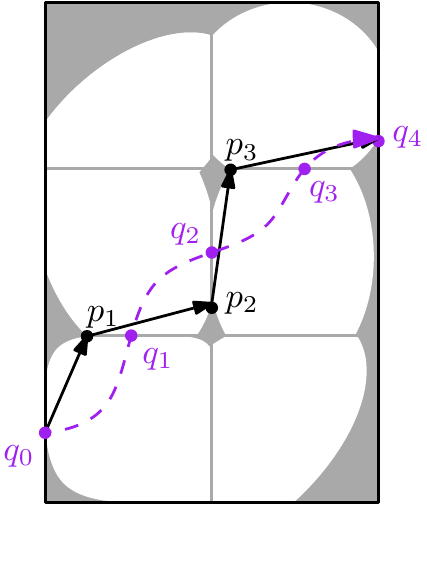}
    \caption{Constructing a sequence of basic monotone paths $\{p_i\}$, from any monotone path $\{q_i\}$.}
    \label{fig:cts_decompose_basic_monotone}
\end{figure}

There is a monotone path from $p_i$ to $q_i$, since $p_i$ is in the same cell and either directly below or directly to the left of $q_i$. Next, we show that there is a monotone path from $q_i$ to $p_{i+1}$. Consider the monotone path from $q_i$ to $q_{i+1}$, which is a subpath of~$Q$. Recall that if $q_i$ is on a vertical (resp. horizontal) boundary, then $q_{i+1}$ is the first intersection of $Q$ with a horizontal (resp. vertical) boundary. Therefore, the monotone path $Q$ must intersect the left (resp. bottom) boundary of the cell that has $q_{i+1}$ on its top (resp. right) boundary. Let the intersection of $Q$ with this left (resp. bottom) boundary be $r_i$. Now, we have a monotone path from $q_i$ to $r_i$, and there is a monotone path from $r_i$ to $p_{i+1}$. Thus, there is a monotone path from $p_i$ to $q_i$, to $r_i$, to $p_{i+1}$. Moreover, $p_i$ is on a vertical (resp. horizontal) boundary, and $p_{i+1}$ is on a horizontal (resp. vertical) boundary, and their cells share a $y$-coordinate (resp. $x$-coordinate). We have shown that there is a basic monotone path from $p_i$ to $p_{i+1}$, completing our proof.
\end{proof}

Lemma~\ref{lem:cfd.graph.c1} motivates us to build the following directed graph. Let $G=(V,E)$ be a graph where~$V$ is the set of critical points in the free space diagram, and~$E$ is the set of all pairs~$(p,q)$ such that there is a basic monotone path from $p$ to $q$. Unfortunately, there are cases where a critical point has $O(n)$ outgoing edges, so that $|E| = O(n^3)$. Our goal will be to reduce the size of this graph, or rather, build essentially the same graph, but implicitly. To do this, we observe the following property of outgoing neighbours of a critical point.

\begin{lemma}
\label{lem:cfd.graph.c2}
Let $p$ be the lowest free point on a vertical cell boundary. Let $q$ be the rightmost critical point in the same row as $p$ such that there is a basic monotone path from $p$ to $q$. Then for any critical point~$r$, there is a basic monotone path from $p$ to $r$ if and only if~$r$ is to the right of $p$, to the left of $q$, and has the same $y$-coordinate as $q$. See Figure~\ref{fig:cts_basic_monotone_path_pqr}.
\end{lemma}

\begin{figure}[ht]
    \centering
    \includegraphics{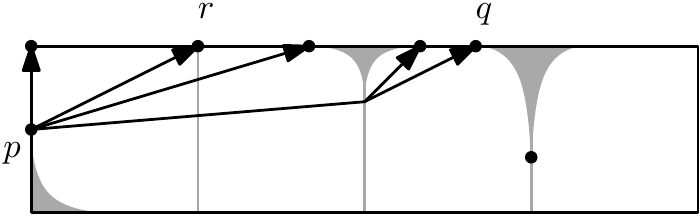}
    \caption{There is a basic monotone path $pr$ if and only if $r$ is to the right of $p$, to the left of $q$, and $r$ has the same $y$-coordinate as $q$.}
    \label{fig:cts_basic_monotone_path_pqr}
\end{figure}

\begin{proof}
We first prove the ``only if" direction. Suppose there is a basic monotone path from $p$ to~$r$. Then clearly $r$ is to the right of $p$. Moreover, $q$ is to the right of $r$ since $q$ is the rightmost critical point so that there is a basic monotone path from $p$ to $r$. Finally, $p$, $q$ and $r$ are all on boundaries of cells that share a $y$-coordinate. Both $q$ and $r$ must be on the top boundary of their respective cells as there are basic monotone paths from $p$ to $q$ and $r$. Hence, $q$ and $r$ share the same $y$-coordinate. This completes the ``only if'' direction.

Next we prove the ``if'' direction. Suppose $r$ shares the same $y$-coordinate as $q$, is to the right of $p$ and to the left of $q$. Let $L$ be the left boundary of the cell with $r$ on its top boundary. Since $p$ and $q$ are in the same row of the free space diagram, and $L$ is between $p$ and $q$, the basic monotone path from $p$ to $q$ must intersect $L$ at some point, which we will call $s$. So there is a monotone path from $p$ to $s$ and $s$ to $r$, since $s$ is on the left boundary and $r$ is on the top boundary of the same cell. Moreover, this monotone path is basic since $p$ is on a vertical boundary, $r$ is on a horizontal boundary, and their cells share a $y$-coordinate. This completes the ``if'' direction.
\end{proof}

We also have a corollary for this lemma where the $x$ and $y$-coordinates are switched.

\begin{corollary}
\label{lem:cfd.graph.c2'}
Let $p$ be the leftmost free point on a horizontal cell boundary. Let $q$ be the topmost critical point in the same column as $p$ such that there is a basic monotone path from $p$ to $q$. Then for any critical point~$r$, there is a basic monotone path from $p$ to $r$ if and only if~$r$ is above $p$, is below $q$, and has the same $x$-coordinate as $q$.
\end{corollary}

We leverage Lemma~\ref{lem:cfd.graph.c2} and Corollary~\ref{lem:cfd.graph.c2'} to build an improved graph that has fewer edges. We use binary trees as an intermediary between the start and end point of the edges in~$E$. We construct a (directed) binary tree $B_i$ for each horizontal line $h_0, h_1, \ldots h_{n_2}$ in the free space diagram. Every parent in $B_i$ has a (directed) edge to its two children. The leaves of the binary tree are a sorted list of the critical points on $h_i$. Analogous binary trees are constructed for the vertical lines in the free space diagram. 
Now we describe the improved graph $G' = (V', E')$. Let $V'$ be the union of the set of critical points in the free space diagram plus the set of internal vertices in the binary trees~$B_i$ for the horizontal and vertical lines in the free space diagram. For each critical point $p$ on a vertical boundary, compute the rightmost critical point $q$ so that there is a basic monotone path from~$p$ to~$q$. In $G=(V,E)$, there is a directed edge from $p$ to every critical point that is between~$p$ and~$q$ and on the horizontal line through $q$. Let the binary tree that corresponds to this horizontal line be $B_i$. We construct a directed edge from $p$ to nodes of $B_i$ so that the union of their descendants in $B_i$ matches this set of contiguous critical points on $h_i$. We do so similarly for the critical points $p$ on horizontal boundaries. This completes the construction of the graph $G' = (V', E')$. 

To decide whether the Fr\'echet distance between our two trajectories is at most $d$, we perform a depth first search in $G'$ to decide whether there is a directed path from the bottom left corner to the top right corner. This completes the statement of the algorithm. Now we prove its correctness.

\begin{lemma}
\label{lem:cfd.graph.c3}
Given a pair of critical points $a$ and $b$, there is a monotone path from $a$ to $b$ if and only if there is a directed path from $a$ to $b$ in the graph $G' = (V', E')$. 
\end{lemma}

\begin{proof}
By Lemma~\ref{lem:cfd.graph.c1} and the definition of the graph $G$, there is a monotone path from $a$ to~$b$ if and only if there is a directed path from $a$ to $b$ in the graph $G = (V,E)$. By Lemma~\ref{lem:cfd.graph.c2} and Corollary~\ref{lem:cfd.graph.c2'}, the outgoing neighbours of a critical point form a set of contiguous critical points on either a horizontal or vertical line in the free space diagram. Without loss of generality, suppose the set of contiguous critical points lie on $h_i$. By the definition of the binary trees $B_i$, we can convert a set of edges to this contiguous set of critical points into a set of paths down the binary tree $B_i$. Hence, there is a directed edge $(p,q)$ in $G=(V,E)$ if and only if there is a directed path $(p_0, p_1, \ldots, p_k)$ in $G'=(V',E')$ where $p_0 = p$, $p_k = q$ and $p_j$ is in a binary tree $B_i$ for all $1 \leq j \leq k-1$. Hence, there is a monotone path from $a$ to $b$ in the free space diagram if and only if there is a directed path from $a$ to $b$ in the directed graph $G' = (V',E')$.
\end{proof}

Finally, we analyse the running time of our algorithm. First we consider the running time of constructing the graph $G' = (V', E')$. Constructing the set of vertices~$V'$ takes $O(n_1n_2)$ time. It remains to construct the set of edges~$E'$. We first compute, for each critical point $p$, the rightmost (resp. topmost) critical point $q$ where there is a basic monotone path from $p$ to $q$.

\begin{lemma}
\label{lem:cfd.graph.c4}
Given a row of $n_1$ cells in a free space diagram, let $p_1, p_2, \ldots, p_{n_1+1}$ be the lowest free points on its vertical boundaries. Then we can compute, in $O(n_1 \log n_1)$ time, the set of critical points $q_1, q_2, \ldots, q_{n_1+1}$, so that for all $1 \leq i \leq k$, $q_i$ is the rightmost critical point on a horizontal cell boundary such that there is a basic monotone path from $p_i$ to $q_i$. 
\end{lemma}

\begin{proof}
Let the lowest free points on the vertical boundaries from left to right be $p_1,p_2,\ldots,p_{n_1+1}$. Let the corresponding highest free points on the same vertical boundaries be $r_1, r_2, \ldots, r_{n_1+1}$. Our algorithm is a dynamic program that considers the critical points $p_i$ and $r_i$ for decreasing values of~$i$.

While performing the dynamic program on decreasing values of $i$, we maintain two lists, one for $p_i$ and one for $r_i$. The list for $p_i$ is of all $p_j$ that have $y$-coordinate greater than the $y$-coordinates of $p_{i+1},\ldots,p_{j-1}$. The list for $r_i$ is of all $r_j$ that have $y$-coordinate less than the $y$-coordinates of $r_{i+1},\ldots,r_{j-1}$. 

Both lists can be maintained in amortized constant time per update by storing the list as a stack. When a new vertical boundary is considered, we will add $p_i$ and $r_i$ to the top of the stack. To maintain the invariant that all $p_j$ in the list must have greater $y$-coordinates than $p_i$, we pop off all elements on the top of the stack that have $y$-coordinate less than or equal to the $y$-coordinate of $p_i$ before adding $p_i$. We maintain $r_i$ analogously, but we check if the $y$-coordinate is greater than or equal. 

Next, we use this pair of stacks to compute our dynamic program. The idea is that we construct a horizontal path starting at the critical point $p_i$ and find the first non-free point it intersects. There are three cases. Either it reaches the rightmost vertical boundary of the row, it intersects a point that is below $p_j$ for some $j>i$, or it intersects a point that is above $r_j$ for some $j>i$. 

In the first case, the horizontal path intersects the rightmost vertical boundary. If $i < n_1+1$, then $q_i$ is on the top boundary of the rightmost cell in the row. If $i = n_1 + 1$, then $q_i$ does not exist.

In the second case, the horizontal line intersects a point that is below~$p_j$. By our invariant, the critical point~$p_j$ must be in our stack. We locate~$p_j$ by computing the smallest index $j$ such that the $y$-coordinate of $p_j$ is at least the $y$-coordinate of $p_i$. For points to the right of $p_j$, our monotone path starting at $p_i$ can only reach points where the monotone path starting at $p_j$ can reach. Hence, we set $q_i$ to $q_j$, which has previously been computing. 

The third case is that the horizontal line intersects a point that is above $r_j$. By our invariant, the critical point~$r_j$ must be in our stack. Again, we locate~$r_j$ by computing the smallest index~$j$ such that the $y$-coordinate of $r_j$ is at most the $y$-coordinate of~$p_i$. Our monotone path starting at $p_i$ can reach $r_j$, but cannot reach any points to the right of $r_j$. Hence, we can set $q_j$ to be the leftmost free point on the top boundary of the cell that has $r_j$ on its right boundary.

Maintaining the stacks takes $O(n_1)$ time. Performing the binary searches to find $p_j$ in the first case and $r_j$ in the second case takes $O(n_1 \log n_1)$ time in total. Hence, the overall running time for computing $q_1, q_2, \ldots, q_{n_1+1}$ is $O(n_1 \log n_1)$ time, as required.
\end{proof}

Lemma~\ref{lem:cfd.graph.c4} allows us to compute all the edges $E$ in $O(n_1 n_2 \log (n_1 + n_2))$ time. Next, we convert the edges $E$ into the edges $E'$. We show that in the graph $G' = (V', E')$, there are at most $O(\log n_1)$ outgoing neighbours of $p$, and that these neighbours can be computed in $O(\log n_1)$ time.

\begin{fact}[Chapter~5.1 of~\cite{DBLP:books/lib/BergCKO08}]
\label{lem:cfd.graph.logn_binary_tree}
Let $B$ be a binary search tree with size $O(n)$. Suppose the leaves of $B$, from left to right, are a sorted list of real numbers. Then we can preprocess $B$ in $O(n)$ time, so that given a pair of real numbers $s$ and $t$, we can select $O(\log n)$ nodes of~$B$ in $O(\log n)$ time so that their descendants are those that lie in the interval $[s,t]$.
\end{fact}

Applying Lemma~\ref{lem:cfd.graph.c4} and Fact~\ref{lem:cfd.graph.logn_binary_tree} to each of the $O(n_1n_2)$ critical points in the free space diagram leads to an $O(n_1n_2 \log (n_1+n_2))$ time algorithm for constructing all edges in the graph $G' = (V', E')$. Moreover, the size of $E'$ is $O(n_1n_2 \log (n_1+n_2))$. Finally, running the depth first search takes $O(|V'| + |E'|) = O(n_1n_2 \log (n_1+n_2))$. This yields the following theorem.

\begin{theorem}
\label{thm:cfd.graph}
Given a pair of trajectories of complexities $n_1$ and $n_2$, there is an $O(n_1n_2 \log (n_1+n_2))$ time algorithm that solves the Fr\'echet distance decision problem by running a depth first search algorithm on the set of critical points in the free space diagram.
\end{theorem}

\subsection{Reference subtrajectory is vertex-to-vertex}
\label{sec:cfd.jump}
\label{sec:cfd.jump.v2v}

Our approach is to run the sweepline algorithm in Section~\ref{sec:dfd.jump}, but we replace the discrete free space diagram (i.e. the $n^2$ grid points) with the graph $G'=(V',E')$ defined in Section~\ref{sec:cfd.graph}. For this we require three modifications to the sweepline algorithm.

The first modification is to generalise the greedy aspect of the depth first search to the new graph. In the discrete free space diagram, we first explore $(x+1,y)$, then $(x+1,y+1)$, and finally $(x,y+1)$. In the graph $G'=(V',E')$, we explore the neighbours with minimum $y$-coordinate first, and of those with the same $y$-coordinate, we explore those with maximum $x$-coordinate first.

The second modification is to create additional nodes in $G'=(V',E')$ for the ending points of the monotone paths $P_i$. The ending point is the lowest point on $l_t$ such that there is a monotone path to that point. However, this lowest point on $l_t$ may not be a node in $G'$. We can detect this case by checking if the last critical point before reaching $l_t$, say $p$, has a basic monotone path through $l_t$. In this case the ending point is simply the intersection of $l_t$ with a horizontal line through $p$. We add this intersection to the graph $G'$, and calculate its outgoing neighbours.

The third modification is to create additional nodes in $G'=(V',E')$ for the starting points of the monotone paths $P_i$. As usual, the starting point of $P_{i+1}$ is the lowest free point on $l_s$ that has a $y$-coordinate greater than or equal to the maximum $y$-coordinate of $P_i$. However, this lowest free point may not be a node of $G'$. If it is not, we add it to $G'$ and calculate its outgoing neighbours. 

All additional nodes created in the second and third modifications are not initially part of the graph $G'=(V',E')$, and are only added to $G'$ when necessary. Hence, the graph $G'$ increases in size as the sweep line algorithm is performed.

A special case that is closely related to the second and third modifications is to detect if there are infinitely many horizontal monotone paths between $l_s$ and $l_t$. We use the same method as Buchin~\etal~\cite{DBLP:journals/ijcga/BuchinBGLL11} to detect if adding any of these additional nodes to $G'=(V',E')$ creates infinitely many horizontal monotone paths, in which case we return a set of $m-1$ monotone paths.

This completes the statement of our algorithm. Next, we argue its correctness.

\begin{lemma}
There exist $m-1$ monotone paths satisfying the conditions of \problemtwo under the continuous Fr\'echet distance in the case that the reference subtrajectory is vertex-to-vertex if and only if our algorithm returns a set of $m-1$ monotone paths. 
\end{lemma}

\begin{proof}
Observe that Lemma~\ref{lem:dfd.jump.valid_path} generalises from the discrete free space diagram to the graph $G' = (V',E')$. In particular, we add a link between a pair of critical points in the continuous free space diagram only if there is a directed path between them in $G'$. So there is always a monotone path from any critical point to the root of its link-cut tree.

Observe that Lemma~\ref{lem:dfd.jump.ell_geq} generalises to $G'$. In particular, by our first modification, our algorithm prefers to link to its lower neighbours first, and the remainder of the proof is identical to the proof of Lemma~\ref{lem:dfd.jump.ell_geq}.

Finally, we observe that Lemma~\ref{lem:dfd.jump.correctness} generalises to $G'$. The proof of both the if and only if directions are identical, so long as we take into account the additional nodes from the second and third modifications. These additional nodes can be treated exactly the same as any other critical point in $G'$, other than that they require additional time to compute. By generalising Lemma~\ref{lem:dfd.jump.correctness} to the continuous free space diagram, we yield the Lemma.
\end{proof}

It remains only to analyse the running time.

\begin{theorem}
There is an $O(n^2 \log^2 n)$ time algorithm that solves \problemtwo under the continuous Fr\'echet distance in the case that the reference subtrajectory is vertex-to-vertex.
\end{theorem}

\begin{proof}
First, we analyse the running time of constructing $G'$. By Lemmas~\ref{lem:cfd.graph.c4} and~\ref{lem:cfd.graph.logn_binary_tree}, we can construct $G'$ in $O(n^2 \log n)$ time. Next, we analyse the running time of the sweepline algorithm. This running time is dominated by two processes, maintaining the link-cut tree and inserting the additional nodes for the second and third modifications.

First, we bound the number of additional nodes, and the time required to insert them. There are at most $m-1$ paths per reference subtrajectory we consider, and we consider $O(n)$ reference subtrajectories. Therefore, we add at most $O(mn)$ additional points to the graph $G'$. Inserting these additional points requires $O(mn)$ time. Inserting the outgoing edges for each of these points requires amortised $O(\log n)$ time per edge by Lemmas~\ref{lem:cfd.graph.c4} and~\ref{lem:cfd.graph.logn_binary_tree},  which is $O(mn \log n)$ time in total

Next, we analyse the running time of maintaining the link-cut tree. For this, the running time is dominated by updating the data structure. There are $O(n^2 \log n)$ edges in $G'$, and $O(mn \log n)$ edges for the additional points. Each update, which is either a link or cut operation, requires $O(\log n)$ amortised time. We observe that, similarly to in Section~\ref{sec:dfd.jump}, once a link is removed it cannot be added again. So all link and cut updates can be performed in $O(n^2 \log^2 n)$ time. This dominates the running time of our algorithm, yielding the theorem.
\end{proof}

\subsection{Reference subtrajectory is arbitrary}
\label{sec:cfd.jump.a2a}

Next, we handle the case where the reference subtrajectory may start and end at arbitrary points on the trajectory. Although there are infinitely many possible starting and ending points, we only consider starting and ending points associated to either vertices of the trajectory, or to additional critical points that we call internal critical points.

We define an external critical point to be critical points that lie on the boundary of a free space cell. In contrast, an internal critical point is in the interior of a cell, and is defined as follows:

\begin{definition}
\label{defn:cfd.a2a.internal_critical_point}
We define an internal critical point to be a point that is in the interior of a cell, on the boundary between free and non-free space, and satisfies one of the following three conditions:

\begin{itemize}[noitemsep]
    \item it is the leftmost or rightmost free point in its cell, or
    \item it shares a $y$-coordinate with an external critical point, or
    \item it is $\ell$ units horizontally to the right of a point that is on the boundary between free and non-free space.
\end{itemize}

\begin{figure}[ht]
    \centering
    \includegraphics[width=\textwidth]{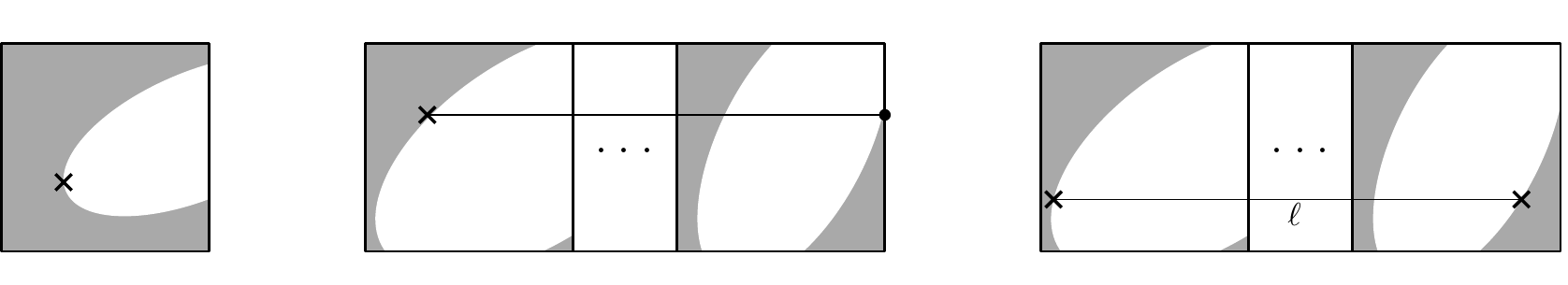}
    \caption{The three types of internal critical points.}
    \label{fig:cfd_3}
\end{figure}

\end{definition}

In Figure~\ref{fig:cfd_3} we show the three types of interior critical points. Both the vertices of the trajectory and the internal critical points are candidate starting points for the reference subtrajectory. To decide whether any such subtrajectory satisfies the properties in \problemtwo, we would like to perform the same sweepline algorithm as the one in Section~\ref{sec:cfd.jump.v2v}. Recall that this algorithm iterates through all reference subtrajectories, and reuses monotone paths between these subproblems using a link-cut data structure. We provided the details for maintaining the link-cut data structure in Section~\ref{sec:dfd.jump}. Recall that three modifications were applied to our sweepline algorithm so that it may be applied to the graph $G' = (V',E')$. We provided the details of these three modifications in Section~\ref{sec:cfd.jump.v2v}.

However, if we were to perform the same sweepline algorithm as in  Section~\ref{sec:cfd.jump.v2v}, our running time would increase to accommodate the additional internal critical points and their reference subtrajectories. In Lemma~\ref{lem:cfd.a2a.n^3_critical_points}, we show that there are $O(n^3)$ internal critical points. Therefore, the running time for maintaining the link-cut tree would increase to $O(n^3 \log^2 n)$. The running time for inserting the additional nodes for the second and third modifications would increase to $O(n^3 m)$. 

To perform the sweepline algorithm in $O(n^3 \log^2 n)$ time, we avoid computing the $O(n^3 m)$ additional nodes for the second and third modifications entirely. We use a completely different approach. Our intuition is that if a monotone path exists for a reference subtrajectory, then either the same or a very similar monotone path is likely to exist for the next reference subtrajectory. We divide the process of computing $m-1$ non-overlapping monotone paths into two steps. The first step is to maintain a large set of overlapping monotone paths, so that any monotone path is represented within this set. The second step is to query this set of overlapping monotone paths to decide whether there are $m-1$ elements that do not overlap. Surprisingly, building and maintaining this set of overlapping monotone paths is more efficient than recomputing the non-overlapping monotone paths for each reference subtrajectory.

Our set of overlapping monotone paths is maintained by the following dynamic data structure.

\begin{fact}[Theorem~16 in~\cite{DBLP:journals/tcs/GavruskinKKL15}]
A dynamic monotonic interval data structure maintains a set of monotonic intervals. A set of intervals is monotonic if no interval contains another. The data structure offers the following three operations. Each operation can be performed in $O(\log n)$ amortised time.
\begin{itemize}[noitemsep]
    \item Insert an interval, so long as the monotonic property is maintained,
    \item Remove an interval, 
    \item Report the maximum number of non-overlapping intervals in the data structure.
\end{itemize}
\end{fact}

With this data structure in mind, we are now ready to state our algorithm in full. Compute all external critical points and build the graph $G' = (V',E')$ defined in Section~\ref{sec:cfd.graph}. Our sweepline algorithm starts with~$s = 0$. For each external critical point~$g$ on~$l_s$, we perform a greedy depth first search to find the lowest point~$r$ on~$l_t$ so that there is a monotone path from $g$ to $r$. We maintain the link-cut data structure throughout the greedy depth first search, so that $r$ is the root of the tree containing $g$. For $s=0$, we only compute a set of additional external critical points, that is, all points on $l_s$ that share a $y$-coordinate with another external critical point. For these external critical points, we compute its lowest monotone path to $l_t$, so that the root of the tree containing the external critical point is on $l_t$.

For $s=0$, we now have a set of link-cut trees, each of which has their root on $l_t$. For each root $r$ on~$l_t$, compute its highest descendant $g$ on $l_s$. This highest descendant can be maintained by each link-cut tree individually, so that when two link-cut trees are merged, we simply take the higher descendant for the merged tree. Finally, for each root $r$ on $l_t$, and its highest descendant~$g$ on $l_t$, we insert the interval $(y(g),y(r))$ into the dynamic monotonic interval data structure, where $y(g)$ and $y(r)$ denote the $y$-coordinates of~$g$ and $r$ respectively. See Figure~\ref{fig:dynamic_interval_scheduling}. This completes the base case of $s=0$ in the sweepline algorithm.

\begin{figure}[ht]
    \centering
    \includegraphics{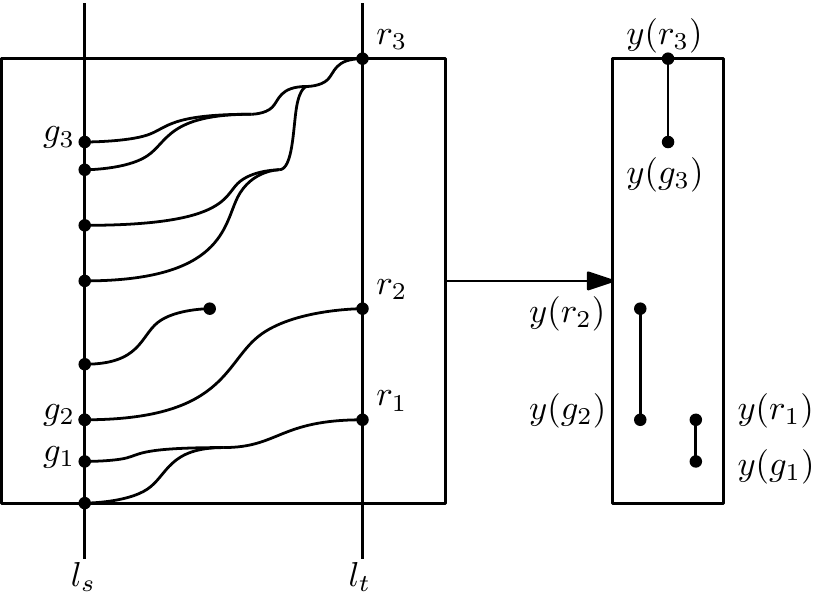}
    \caption{The link cut tree rooted at $r_i$, with its highest descendant $g_i$ on $l_s$, shown on the left. The $y$-intervals $(y(g_i),y(r_i))$ stored by the dynamic monotonic interval data structure, shown on the right.}
    \label{fig:dynamic_interval_scheduling}
\end{figure}

Next, compute all internal critical points defined in Definition~\ref{defn:cfd.a2a.internal_critical_point}, and sort them by $x$-coordinate. We sweep the vertical lines $l_s$ and $l_t$ from left to right, and whenever $l_s$ or $l_t$ pass through an internal critical point, we process it as an event. There are five types of events, depending on the type of the internal critical point and whether it passes through $l_s$ or $l_t$. On all five of these events, we maintain the invariant that the intervals inserted into our data structure are monotonic. We also maintain the invariant that any monotone path is represented by an interval in the data structure. So that we do not need to re-compute the exact $y$-coordinates of these monotone paths at every event point, we only store the relative positions of the intervals with respect to one another (i.e. whether they are overlapping or non-overlapping). By only storing the relative positions of the $y$-coordinates, we can avoid computing the $O(n^3m)$ starting and ending points of monotone paths in the second and third modifications that are required for the algorithm in Section~\ref{sec:cfd.jump.v2v}. We split our analysis of our five types of events into five cases:

\begin{itemize}[noitemsep]
    \item The first type of event is if $l_s$ passes through an internal critical point $g$ that is the leftmost free point in its cell. For this event, insert $g$ into the graph $G'$, and compute the lowest monotone path from $g$ to the $l_t$. If a jump operation was performed using the link-cut data structure, then the root~$r$ of the tree containing $g$ already exists in the dynamic monotonic interval data structure. We only update this interval to $(y(g),y(r))$ if $g$ is the highest descendant of $r$ on $l_s$. If no jump operation was performed, then $r$ is a new root on $l_t$, so we simply insert the interval $(y(g),y(r))$ into our data structure. 
    \item The second type of event is if $l_s$ passes through an internal critical point $g$ that shares a $y$-coordinate with an external critical point. Similarly to the first event, we insert the~$g$ into the graph~$G'$, and compute its root~$r$ on the line~$l_t$. We insert a new interval $(y(g),y(r))$ into the dynamic monotonic interval data structure if $r$ is new, or replace an existing interval if $r$ is not new, but $g$ is the highest descendant on~$l_s$. After this, we consider whether $l_s$ passing through $g$ causes a pair of monotone paths that were previously overlapping to now be non-overlapping. In particular, suppose that $g$ shares a $y$-coordinate with an exterior critical point, which in turn shares a $y$-coordinate with a root $r'$ on $l_t$. In other words, $y(g) = y(r')$. Then the pair of intervals $(y(g),y(r))$ and $(y(g'),y(r'))$ may switch from overlapping to non-overlapping, or vice versa. If this is the case, we remove the interval $(y(g),y(r))$ and replace it with a new interval so that their relative positions are $g',r',g,r$ instead of $g',g,r',r$, or vice versa.
    \item The third type of event is if $l_t$ passes through an internal critical point $r$ that is the rightmost free point in its cell. For this event, simply remove the interval $(g,r)$ from the dynamic monotone interval data structure, where $g$ is the highest ancestor of $r$ that is on $l_s$. 
    \item The fourth type of event is if $l_t$ passes through an internal critical point $r$ that shares a $y$-coordinate with an external critical point. Let the highest descendant of $r$ on $l_s$ be $g$. We consider whether $l_t$ passing through $r$ makes the monotone path from $g$ to $r$ invalid. This could be the case if the last segment in this monotone path is horizontal, and if $r$ lies on a boundary between free and non-free space that has a negative gradient. If this is the case, then we remove the invalid path from $r$ to its child, and recompute its new highest descendant on $g$, if one exists. Next, similar to the second event, we consider whether $l_t$ passing through $r$ causes a pair of monotone paths that were previously non-overlapping to now be overlapping. In particular, suppose that $r$ shares a $y$-coordinate with an exterior critical point, which in turn shares a $y$-coordinate with a root $g'$ on $l_s$. In other words, $y(r) = y(g')$. Then the pair of intervals $(y(g),y(r))$ and $(y(g'),y(r'))$ may switch from overlapping to non-overlapping, or vice versa. If this the case, we remove the interval $(y(g),y(r))$ and $(y(g'),y(r;))$ and replace it with a new interval so that their relative positions are $g,r,g',r'$ instead of $g,g',r'r'$, or vice versa. 
    \item The fifth type of event is if $l_s$ passes through a point $g'$ while $l_t$ simultaneously passes through a point~$r$, so that both $g'$ and~$r$ are on the boundaries between free and non-free space, and $y(g') = y(r)$. We consider whether $l_s$ and $l_t$ passing through $g'$ and~$r$ causes a pair of monotone paths that were previously non-overlapping to now be overlapping. In particular, the pair of intervals $(y(g),y(r))$ and $(y(g'),y(r'))$ may switch from overlapping to non-overlapping, or vice versa. If this is the case, we remove the interval $(y(g),y(r))$, and replace it with a new interval so that their relative positions of $g',r',g,r$ instead of $g',g,r',r$, or vice versa.
\end{itemize}

After processing an event, we report whether there are $m-1$ monotone paths for this event and its associated reference subtrajectory. To do this, we query the dynamic monotone interval data structure to report the maximum number of non-overlapping intervals in the data structure. If there are $m-1$ or more non-overlapping intervals, we report these $m-1$ intervals. We can retrieve the original monotone paths by storing the monotone paths with the intervals when we insert them. We can modify the monotone paths to start and end at $l_s$ and $l_t$ in constant time per monotone path. 

This completes the statement of our algorithm. Next, we prove its correctness. 

\begin{definition}
Given $s$ and $t$, a monotone path from $g$ on $l_s$ to $r$ on $l_t$ is called minimal if there does not exist a monotone path from $g'$ on $l_s$ to $r'$ on $l_t$ such that the interval $[y(g'),y(r')]$ is a strict subset of the interval $[y(g),y(r)]$.
\end{definition}

\begin{lemma}
\label{lem:cfd.a2a.exists_minimal}
Given $s$ and $t$, and a monotone path from $g$ on $l_s$ to $r$ on $l_t$, there exists a minimal monotone path from $g'$ on $l_s$ to $r'$ on $l_t$ so that $[(y(g'),y(r')] \subseteq [(y(g),y(r))]$.
\end{lemma}

\begin{proof}
Let $r'$ be the lowest point on $l_t$ so that there is a monotone path from $g$ to $r'$. We know $r'$ exists since we can compute it with a greedy depth first search. Let $g'$ be the highest point on $l_s$ so that there is a monotone path from $g'$ to $r'$. We have by definition that $y(g) \leq y(g') \leq y(r') \leq y(r)$, so $[(y(g'),y(r')] \subseteq [(y(g),y(r))]$. It suffices to show that the monotone path from $g'$ to $r'$ is minimal.

Suppose the monotone path from $g'$ to $r'$ is not minimal. Then there exists a monotone path from $g''$ to $r''$ so that $[(y(g''),y(r'')] \subset [(y(g'),y(r'))]$. But now, the monotone paths from $g''$ to $r''$ and from $g'$ to $r'$ must cross at some point $u$. Construct the monotone path from $g''$ to $u$ to $r'$. But $g'$ is the highest point on $l_s$ so that there is a monotone path from $g'$ to $r'$, so $y(g'') \leq y(g')$. But $y(g') \leq y(g'')$ since $[(y(g''),y(r'')] \subset [(y(g'),y(r'))]$. So $y(g'') = y(g')$. 

Therefore, there exists a monotone path from $g'$ to $r''$ such that $y(r'') < y(r')$. But now, the monotone path from $g'$ to $r''$ and the monotone path from $g$ to $r$ must cross at some point $u$. Construct the monotone path $g$ to $u$ to $r''$. Since $y(r'') < y(r')$, this contradicts the construction that $r'$ is the lowest point on $l_t$ such that there is a monotone path from $g$ to $r'$. Hence, our initial assumption that $g'$ to $r'$ is not minimal cannot hold, and we are done.
\end{proof}

\begin{lemma}
\label{lem:cfd.a2a.order_preserving_bijection}
Given $s$ and $t$, suppose we run our sweepline algorithm until our pair of sweeplines reach $l_s$ and $l_t$ respectively. Then there is an order preserving bijection from the $y$-intervals of the minimal monotone paths to the set of intervals in our dynamic monotonic interval data structure.
\end{lemma}

\begin{proof}
Our proof is divided into three parts. First we prove an order preserving bijection in the base case, where $s=0$. Next, we prove that in the inductive case, if the sweepline does not pass through any internal critical points, then the order preserving bijection is preserved. Finally, we prove that as the sweepline passes through an internal critical point, the order preserving bijection is preserved.

For the base case, consider when $s=0$. Let $g$ be on $l_s$ and $r$ be on $l_t$ so that the path from $g$ to $r$ is a minimal monotone path. Suppose $g$ is not an external critical point. Then there is a segment of free space directly above $g$. We show that the monotone path from $g$ to $r$ consists of a horizontal path from $g$ to an external critical point. Suppose the contrary, that the initial path from $g$ is not horizontal. Let the initial path be from $g$ to $h$, so that $y(h) > y(g)$. Then if we selected the point on $l_s$ with $y$-coordinate $y(g) + \varepsilon$ for a sufficiently small $\varepsilon$, then there would still have $y(h) > y(g) + \varepsilon$, and we would maintain our monotone path. But now, the new path would have a $y$-interval from $y(g) + \varepsilon$ to $y(r)$, which is a strict subset of the $y$-interval from $y(g)$ to $y(r)$, contradicting the fact that the monotone path from $g$ to $r$ is minimal. Therefore, if $g$ is not an external critical point, then there is a horizontal path from $g$ to its next point, which must be an external critical point. Therefore, all minimal monotone paths start at either external critical points or these additional points that share a $y$-coordinate with an external critical point. By the definition of our base case in our algorithm, we insert all $y$-intervals for these potential starting points, so we have all minimal monotone paths in our initial data structure. Hence, we have shown the base case of our induction.

For the inductive case where the sweepline does not pass through an internal critical point, suppose the inductive hypothesis that $l_s$ and $l_t$ are sweeplines for which there is an order preserving bijection. Let $l_{s'}$ and $l_{t'}$ be another pair of sweeplines to the right of $l_s$ and $l_t$, so that there are no internal critical points between $l_s$ and $l_{s'}$, and similarly between $l_t$ and $l_{t'}$. We will show that there is a bijection between the $y$-intervals of the minimal monotone paths between $l_s$ and $l_t$, and $l_{s'}$ and $l_{t'}$. By composing this bijection with the bijection between the minimal monotone paths between $l_s$ and $l_t$ and the data structure, we yield a bijection between the minimal monotone paths between $l_{s'}$ and $l_{t'}$ and the data structure.

We construct our bijection between the two sets of minimal monotone paths as follows. For each minimal monotone path starting at $g$ on $l_s$ and ending at $r$ on $l_t$, we perform an operation to yield a minimal monotone path between $l_{s'}$ and $l_{t'}$. First, we let $g'$ be the highest point on $l_{s'}$ such that there is a monotone path from $g$ to $g'$ and a monotone path from $g'$ to $r$. Second, we extend the monotone path horizontally from $r$ to its right, until it either hits $l_{t'}$, or the boundary between free and non-free space. We extend the monotone path along the boundary between free and non-free space until we reach $l_{t'}$. When this path reaches $l_{t'}$, we define this point to be $r'$. This completes the description of the mapping from minimal monotone paths of $l_s$ and $l_t$ to paths between $l_{s'}$ and $l_{t'}$. 

Our proof of this bijection is divided into five parts. First, we show that the mapping is well defined. Second, we show that the mapping is from minimal monotone paths to minimal monotone paths. Third, we show that the mapping is injective. Fourth, we show that the mapping is surjective. Fifth, we show that the mapping is order preserving. 

First, we show that the mapping is well defined. The point $g'$ is well defined. Extending $r$ horizontally is well defined, however, it may be that $r$ cannot be extended along the boundary between free and non-free space until we reach $l_{t'}$. There are two ways this can occur. First, the free space may stop at some $x$-coordinate before $l_{t'}$. However, in this case, we have an internal critical point that is the rightmost free point in its cell. But we assumed there were no such critical points between $l_t$ and $l_{t'}$. Second, the boundary between free space and non-free space may have a negative gradient, so that the monotone path cannot travel along it. In this case, we must have a horizontal path starting at $r$ that is extended towards the boundary with a negative gradient. If the point $r$ does not share a $y$-coordinate with an external critical point, then we can reduce the $y$-coordinate of the horizontal path containing $r$, contradicting the fact that $g$ to $r$ is minimal. Hence, $r$ shares a $y$-coordinate with an external critical point, so does our horizontal path that intersects the boundary between free and non-free space. Therefore, this point is an internal critical point, which again contradicts our assumption that there are no such critical points between $l_t$ and $l_{t'}$.

Second, we show that the mapping is from minimal monotone paths to minimal monotone paths. By definition, $r'$ is the lowest point on $l_{t'}$ such that there is a monotone path from $g'$ to $r'$. We will show that $g'$ is the highest point on $l_s$ so that there is a monotone path from $g'$ to $r'$. Putting these two facts together yields that the monotone path from $g'$ to $r'$ is minimal. Suppose for the sake of contradiction that there is another point, $g''$, that has a larger $y$-coordinate than $g'$, and there is a monotone path from $g''$ to $r'$. Extend the monotone path starting at $g''$ horizontally to the left, and then along the boundary between free and non-free space, until it reaches $l_s$. For the same reason as extending $r$ horizontally to the right to reach $l_{t'}$, we must be able to extend $g''$ horizontally to the left to reach $l_s$, as there are no internal critical points between $l_s$ and $l_{s'}$. We show that the extension of $g''$ meets $l_s$ at $g$. If $y(g'') > y(g)$, this would contradict the minimality of the monotone path from $g$ to $r$. If $y(g'') < y(g)$, then the extension of $g''$ crosses the path from $g'$ to $g$, which is impossible. Hence, we have monotone paths from $g$ to $g''$ to $r$ where $g''$ has a larger $y$-coordinate of $g'$. This contradicts the definition of $g'$, so the point $g''$ cannot exist, as required.

Third, we show the mapping is injective. Suppose for the sake of contradiction that there are two monotone paths, $g$ to $r$ and $g'$ to $r'$, from $l_s$ to $l_t$, that both map to a monotone path $g''$ to $r''$, from $l_{s'}$ to $l_{t'}$. Then there are monotone paths from $g$ to $g''$ to $r$ and from $g'$ to $g''$ to $r'$. Without loss of generality, suppose $y(g') \geq y(g)$. Then by the minimality of the path from $g$ to $r$, we get that $y(r') \geq y(r)$. But now, consider the monotone path from $g$ to $g''$ to $r'$. By the minimality of $g$ to $r$, we get that $y(r') \leq y(r)$, so $y(r') = y(r)$. By the minimality of $g'$ to $r'$, we get that $y(g') \leq y(g)$, so $y(g') = y(g)$. Hence, $g = g'$ and $r = r'$, so our mapping is injective.

Fourth, we show the mapping is surjective. Suppose for the sake of contradiction that there is a monotone path between $l_{s'}$ and $l_{t'}$ that is not mapped to. Let this minimal monotone path be from $g'$ to $r'$. Let $r$ be the highest point on $l_t$ such that there is a monotone path from $g'$ to $r$ to $r'$. Let $g$ be the highest point on $l_s$ such that there is a monotone path from $g$ to $g'$ to $r$. We claim that our construction maps the path $g$ to $r$ to the path $g'$ to $r'$. First, we show that $g'$ is the highest point on $l_{s'}$ such that there is a monotone path from $g$ to $g'$ to $r$. Suppose there is a higher point $g''$ so that there is a monotone path from $g''$ to $r$. Then we have a monotone path from $g''$ to $r'$, contradicting the minimality of the path from $g'$ to $r'$. Next, we show that if we extend the monotone path from $r$ horizontally to the right, then we reach the point $r'$ on $l_{t'}$. Suppose that instead we reach another point $r''$. By the definition of our extension, $r''$ is the lowest point on $l_{t'}$ such that there is a monotone path from $r$ to $r''$. So $y(r'') \leq y(r')$. But by the minimality of the path $g'$ to $r'$, we get that $y(r'') \geq y(r')$. Hence, $r'=r''$ as required, and the mapping is surjective.

Fifth, we show that the mapping is order-preserving. Suppose there are paths $g$ to $r$ and $g'$ to $r'$, both between $l_s$ and $l_t$. Without loss of generality, let $y(g) < y(g')$. Suppose these monotone paths map to $g''$ to $r''$ and $g'''$ to $r'''$ respectively. We show that $y(g'') < y(g''')$. Suppose for the sake of contradiction that $y(g'') \geq y(g''')$. If $g'' = g'''$, by minimality we would have $r'' = r'''$, contradicting the injectivity of our mapping. So $y(g'') > y(g''')$. Therefore, the path from $g$ to $g''$ and the path from $g'$ to $g'''$ must cross at some point $u$. But now we have a monotone path from $g$ to $u$ to $g'''$, which contradicts the fact that $g''$ is the highest point on $l_{s'}$ such that there is a path from $g$ to $g''$ to $r$. Hence, $y(g'') < y(g''')$ as required. Finally, we show that $y(g') < y(r)$ if and only if $y(g''') < y(r'')$. Suppose for the sake of contradiction that $y(g') < y(r)$ and $y(g''') > y(r'')$. Take the path from $r$ to $r''$ that is a horizontal line to the right, until it reaches the boundary between free and non-free space, and take the path along this boundary. Similarly, take the path from $g'''$ to $g'$ that is a horizontal line to the left, until it reaches the boundary between free and non-free space, and take the path along this boundary. For any pair of vertical lines $l_a$ between $l_s$ and $l_{s'}$ and $l_b$ between $l_{t}$ and $l_{t'}$, such that $l_b - l_a = \ell$, define $a$ to be the intersection of the path between $g'''$ and $g'$ with $l_a$, and define $b$ to the intersection of the path between $r$ and $r''$ with $l_b$. When $l_a = l_s$, we have $y(a) < y(b)$. When $l_s = l_{s'}$, we have $y(a) > y(b)$. So by the intermediate value theorem, there must be a point where $y(a) = y(b)$. This cannot occur when both $a$ and $b$ are on the horizontal portions of their respective paths. If they are both on the boundary between free and non-free space, then $a$ and $b$ are interior critical points, where $a$ is horizontally $\ell$ units to the right $b$. If one of $a$ or $b$ is on the boundary between free and non-free space, and the other is on the horizontal portion, then we also have an internal critical point, as there is a point on the boundary between free and non-free space sharing a $y$-coordinate with an external critical point. Putting this all together, we yield a contradiction in the case where $y(g') < y(r)$ and $y(g''') > y(r'')$. We yield a similar contradiction in that case where $y(g') > y(r)$ and $y(g''') < y(r'')$. This shows that our bijection is indeed order preserving. This completes the proof of the inductive case where the sweepline does not pass through an internal critical point.

For the inductive case where the sweepline passes through an internal critical point, we show that the bijection between the set of $y$-intervals and the set of minimal monotone paths is preserved. If the internal critical point is the leftmost free point in the cell, there is one extra starting point for a minimal monotone path to be considered. If an additional minimal monotone path exists due to this point, our algorithm adds it to our data structure. Similarly, if the internal critical point is the rightmost free point in the cell, there is one fewer ending point for minimal monotone paths, and a $y$-interval is deleted from our data structure if necessary. If our internal critical point is a point on the boundary between free and non-free space that shares a $y$-coordinate with an external critical point, we consider the two potential modifications to the set of minimal monotone paths. First, if the internal critical point is on $l_s$, we consider whether this point is the starting point of a new minimal monotone path. If the internal critical point is on $l_t$, we consider whether this point is the last valid point for a minimal monotone path that needs to be removed. Second, we consider whether there may be a swap in relative positions of the $y$-coordinates in minimal monotone paths. If this is the case, we replace the previous $y$-coordinates with new $y$-coordinates in the data structure. Finally, if the internal critical point is a pair of points on the boundary between free and non-free space that are $\ell$ units horizontally away from one another, we consider whether this internal critical point changes the relative positions of the $y$-coordinates of minimal monotone paths, and update the data structure accordingly. Note that if the internal critical point is a leftmost free point or rightmost free point, this only adds or deletes minimal monotone paths, since the critical point is local to either $l_s$ or $l_t$. In contrast, if the internal critical point is a pair of points on the boundaries of free and non-free space that are $\ell$ units horizontally from one another, this only affects the relative positions of $y$-coordinates of minimal monotone paths, so it suffices to consider swapping these relative positions in the data structure. To summarise, in all cases we update the data structure and preserve the relative positions of the $y$-intervals of the minimal monotone paths in our data structure. This completes the proof of our lemma.
\end{proof}

\begin{lemma}
\label{lem:cfd.a2a.correctness}
There exist $m-1$ monotone paths satisfying the conditions of \problemtwo under the continuous Fr\'echet distance if and only if our algorithm returns a set of $m-1$ monotone paths.
\end{lemma}

\begin{proof}
Suppose our algorithm returns a set of $m-1$ monotone paths. We have already shown in Section~\ref{sec:cfd.jump.v2v} that the monotone paths computed by the greedy depth first search and link-cut data structure on $G'=(V',E')$ result in valid paths. These monotone paths are non-overlapping due to the order preserving bijection in Lemma~\ref{lem:cfd.a2a.order_preserving_bijection}. Hence, our reference subtrajectory plus our $m-1$ monotone paths forms a subtrajectory cluster that satisfies the conditions of \problemtwo.

Suppose there exist $m-1$ monotone paths satisfying the conditions of \problemtwo. By Lemma~\ref{lem:cfd.a2a.exists_minimal} there exist $m-1$ minimal monotone paths satisfying the conditions of \problemtwo. By Lemma~\ref{lem:cfd.a2a.order_preserving_bijection} there is an order preserving bijection from the $y$-intervals of these $m-1$ minimal monotone paths to the set of intervals in our dynamic monotonic interval data structure. Hence, when our algorithm reaches the reference subtrajectory that satisfies the conditions of \problemtwo, our algorithm will report that there are $m-1$ non-overlapping intervals in the data structure, as required.
\end{proof}

Finally, we perform a running time analysis on our algorithm. We begin by bounding the number of internal critical points.

\begin{lemma}
\label{lem:cfd.a2a.n^3_critical_points}
There are $O(n^3)$ internal critical points.
\end{lemma}

\begin{proof}
If the internal critical point is the leftmost free point in its cell, there is only one copy of it in its cell, so there are at most $O(n^2)$ internal critical points of this type.

If the internal critical point shares a $y$-coordinate with an external critical point, and is on the boundary between free and non-free space, there are only two copies of it per external critical point and per cell in the same row as it. In total, there are $O(n^2)$ external critical points, and $O(n)$ cells in the same row as it, so there are $O(n^3)$ internal critical points of this type.

If the internal critical point is $\ell$ unites horizontally to the right of a point that is on the boundary between free and non-free space, there is at most a constant number of copies per pair of cells in the same row. This is because, for any pair of cells in the same row as each other, a $\ell$-unit horizontal translation of its free space boundary would only intersect the free space boundary of the other cell a constant number of times. As there are $O(n^2)$ pairs of cells that share a row, there are $O(n^2)$ internal critical points of this type.
\end{proof}

Now we can analyse the running time of our algorithm. to obtain the main result of Section~\ref{sec:cfd}.

\theoremtwo*

\begin{proof}
Computing the graph $G'=(V',E')$ requires $O(n^2 \log n)$ time. Computing all internal critical points takes $O(n^3)$ time, by computing each critical point individually in constant time. Sorting them by $x$-coordinate requires $O(n^3 \log n)$ time. At each of the $O(n^3)$ events, we may insert the critical point into the graph $G'$. Inserting a single critical point inserts up to $O(\log n)$ edges into the graph $G'$. In total, there are $O(n^3 \log n)$ edges in the graph $G'$. Each edge can be linked or cut at most once, so by the same amortised analysis as in Section~\ref{sec:dfd.jump}, we spend at most $O(n^3 \log^2 n)$ time updating the link-cut data structure.

We maintain the dynamic monotonic interval data structure by inserting, deleting, or replacing (both a delete and an insert) an interval, which requires $O(\log n)$ time per event. In total, maintaining the dynamic monotonic interval data structure requires $O(n^3 \log n)$ time. 

In total, the running time is dominated by updating the link-cut data structure, and the running time of our algorithm is $O(n^3 \log^2 n)$.
\end{proof}

\section{Lower bound}
\label{sec:3ov}

The main theorem we will prove in this section is the following:

\theoremfour*

We reduce from the 3OV problem to \problemtwo. The formal definition of 3OV is as follows:

\addtocounter{problem}{-1}
\begin{problem}[3OV]
\label{def:p4}
We are given three sets of vectors $\mathcal X = \{X_1, X_2, \ldots, X_n\}$, $\mathcal Y = \{Y_1, Y_2, \ldots, Y_n\}$ and $\mathcal Z = \{Z_1, Z_2, \ldots, Z_n\}$. For $1 \leq i,j,k \leq n$, each of the vectors $X_i$, $Y_j$ and $Z_k$ are binary vector of length $W$. Our problem is to decide whether there exists a triple of integers $1 \leq i,j,k \leq n$ such that $X_i$, $Y_j$ and $Z_k$ are orthogonal. The three vectors are orthogonal if
$X_i[h] \cdot Y_j[h] \cdot Z_k[h] = 0$ for all $1 \leq h \leq W$.
\end{problem}

This section is structured as follows. In Section~\ref{sec:3ov_construction}, given an instance $(\mathcal X, \mathcal Y, \mathcal Z)$ of 3OV, we construct an \problemtwo instance $(T, m, \ell, d)$ of complexity $O(nW)$. In Section~\ref{sec:3ov_well_defined} we fill in some missing details in our construction. In Sections~\ref{sec:3ov_antipodal_property}-\ref{sec:3ov_subtrajectory_property} we prove several useful properties of our construction. In Section~\ref{sec:3ov_free_space_diagram} we describe the free space diagram of our construction. In Section~\ref{sec:3ov_paths}, we construct a set of monotone paths in the free space diagram. In Sections~\ref{sec:3ov_if}, we use this set of monotone paths to prove that if our input $(\mathcal X, \mathcal Y, \mathcal Z)$ is a YES instance for 3OV, then our construction $(T, m, \ell, d)$ is a YES instance for \problemtwo. In Section~\ref{sec:3ov_cuts}, we define and construct a set of sequences that show that there are is no monotone path between certain points. In Section~\ref{sec:3ov_only_if}, we use this set of sequences to show that if our input $(\mathcal X, \mathcal Y, \mathcal Z)$ is a NO instance for 3OV then our construction $(T, m, \ell, d)$ is a NO instance for \problemtwo. Finally, in Section~\ref{sec:3ov_putting_together}, we put it all together and prove Theorem~\ref{theoremfour}. 

\subsection{Construction}
\label{sec:3ov_construction}

Recall that as input we are given a 3OV instance, that is, three sets $\mathcal X$, $\mathcal Y$, $\mathcal Z$ of $n$ binary vectors of length $W$. Given this instance, we construct an instance $(T,m,\ell,d)$ for \problemtwo. We construct the trajectory $T$ by constructing two subtrajectories $T_1$ and $T_2$ and connecting them via a point.

Let $r$ and $r'$ be positive real numbers so that $10r < r'$. We use polar coordinates in the complex plane. Recall that $r \cis \theta$ is a point that is $r$ units away from the origin, at an angle of $\theta$ anticlockwise from the positive real axis. Define $\phi = \frac {\pi}{4W + 6}$ and $d = ||r \cis \pi - r' \cis \phi||$, where $||\cdot||$ denotes the Euclidean norm. Define $\delta = (r+r' - d)$. Let $\varepsilon > 0$ be arbitrarily small relative to $\delta$ and $r$. 

Define the following points which will be used to define $T_1$. See Figure~\ref{fig:lower_1}.
\[
    \begin{array}{l c l l}
    A_{h,u} &=& r \cis ((2W+3+u) \cdot \phi), \text{ if}\ u \neq 2h \\
    A_{h,u} &=& (r - \frac 2 3 \delta) \cis ((2W+3+u) \cdot \phi), \text{ if}\ u = 2h,\\
    A_{W+1,u} &=& r \cis ((2W+3+u) \cdot \phi) \\
    B_{k,h} &=& r \cis ((4W+5) \cdot \phi), \text{ if}\ Z_k[h] = 1\\
    B_{k,h} &=& (r-\varepsilon) \cis ((4W+5) \cdot \phi), \text{ if}\
    Z_k[h] = 0 \\
    B_{k,W+1} &=& r \cis ((4W+5) \cdot \phi) \\
    C &=& r \cis ((4W+6) \cdot \phi) \\
    D &=& r \cis 0 \\
    E &=& r \cis \phi \\
    F_{h,u} &=& r \cis ((u+1) \cdot \phi), \text{ if}\ u \neq 2h \\
    F_{h,u} &=& (r-\frac 2 3 \delta) \cis ((u+1) \cdot \phi), \text{ if}\ u = 2h \\
    G &=& r \cis ((2W+3) \cdot \phi) \\
    H_1 &=& 4 r' \cis ((3W+4) \cdot \phi) \\
    H_2 &=& 8 r' \cis ((2W+3) \cdot \phi).
    \end{array}
\]

Now we are ready to define the first subtrajectory $T_1$.
\begin{align*}
T_1 = &\bigcirc_{1 \leq k \leq n} \Big(
    \bigcirc_{1 \leq h \leq W} \big(
        G
        \circ \bigcirc_{1 \leq u \leq 2W+1} (A_{h,u})
        \circ B_{k,h}
        \circ C
        \circ D
        \circ C
        \circ D 
        \\ &\hspace{3.6cm}
        \circ E
        \circ \bigcirc_{1 \leq u \leq 2W+1} (F_{h,u})
    \big)
    \\ &\hspace{1.8cm}
    \circ G 
    \circ \bigcirc_{1 \leq u \leq 2W+1} (A_{W+1,u})
    \circ B_{k,W+1}
    \circ C
    \circ D
    \circ C
    \circ D
    \\ &\hspace{1.8cm}
    \circ C
    \circ H_1
    \circ H_2
\Big)
\end{align*}

\begin{figure}[ht]
    \centering
    \includegraphics[width=0.6\textwidth]{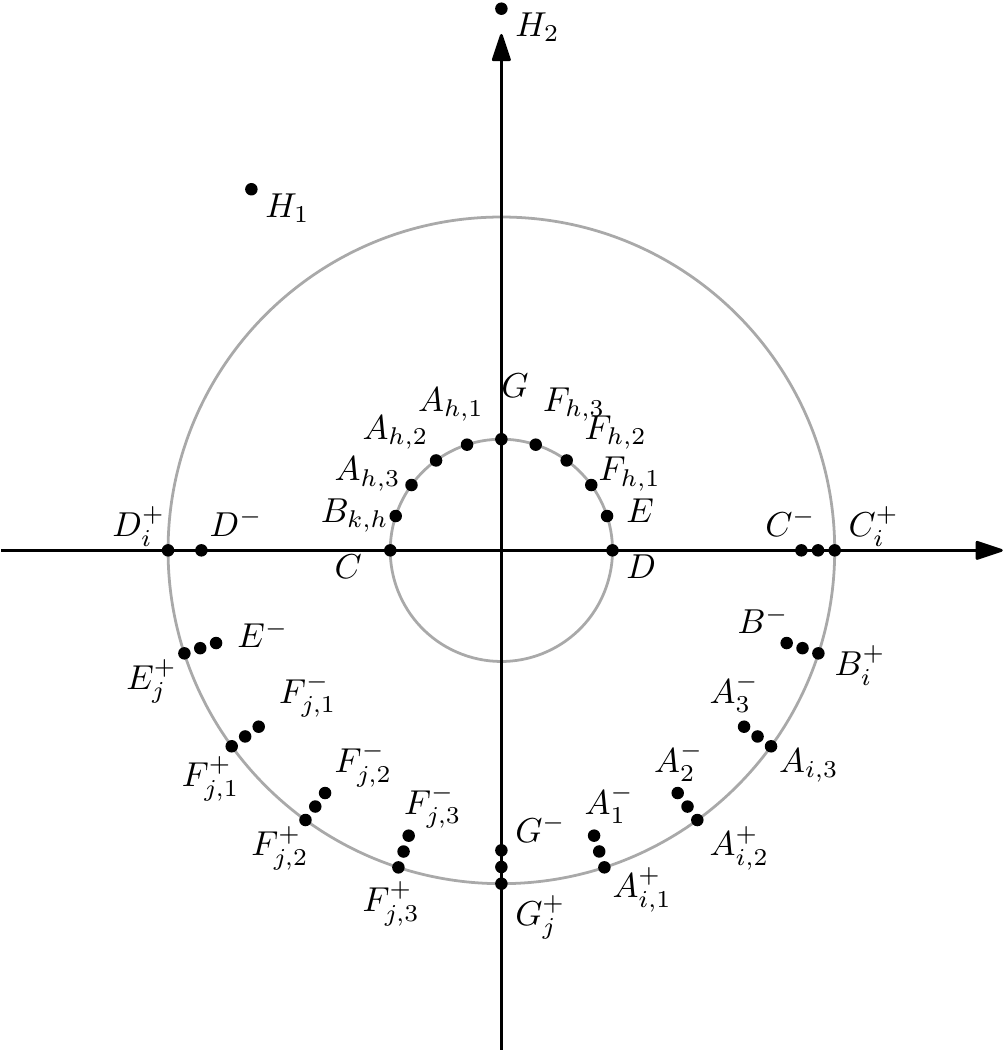}
    \caption{The vertices of $T_1$ and $T_2$, for $W=1$.}
    \label{fig:lower_1}
\end{figure}

Define the following points which will be used to define $T_2$. 
\[
    \begin{array}{l c l l}
    A^-_u &=& (r'-\delta) \cis (\pi + (2W+3+u) \cdot \phi) \text{ for all } 1 \leq u \leq 2W+1\\
    B^- &=& (r'-\delta) \cis (\pi + (4W+5) \cdot \phi) \\
    C^- &=& (r'-\delta) \cis (\pi + (4W+6) \cdot \phi) \\
    D^- &=& (r'-\delta) \cis (\pi) \\
    E^- &=& (r'-\delta) \cis (\pi + \phi) \\
    F^-_u &=& (r'-\delta) \cis (\pi + (u+1) \cdot \phi) \text{ for all } 1 \leq u \leq 2W+1 \\
    G^- &=& (r'-\delta) \cis (\pi + (2W+3) \cdot \phi) \\
    H^- &=& (r'-\delta) \cis ((2W+3) \cdot \phi) \\
    
    A^+_{i,u} &=& r' \cis (\pi + (2W+3+u) \cdot \phi), \text{ if $i$ or $u$ are odd} \\
    A^+_{i,u} &=& r' \cis (\pi + (2W+3+u) \cdot \phi), \text{ if $i$ and $u$ are even and } X_{\frac i 2}[\frac u 2] = 1 \\
    A^+_{i,u} &=& (r'- \frac 2 3 \delta) \cis (\pi + (2W+3+u) \cdot \phi), \text{ if $i$ and $u$ are even and } X_{\frac i 2}[\frac u 2] = 0 \\
    B_i^+ &=& b_i \cis (\pi + (4W+5) \cdot \phi), \text{ where $b_i$ is defined below} \\
    C_i^+ &=& c_i \cis ((4W+6) \cdot \phi), \text{ where $c_i$ is defined below} \\
    D^+_i &=& r' \cis (\pi) \\
    E_j^+ &=& e_j \cis (\pi + \phi) \text{ where $e_j$ is defined below}\\
    F^+_{j,u} &=& r' \cis (\pi + (u+1) \cdot \phi), \text{ if $j$ or $u$ are odd}  \\
    F^+_{j,u} &=& r' \cis (\pi + (u+1) \cdot \phi), \text{ if $j$ and $u$ are even and } Y_{\frac j 2}[\frac u 2] = 1 \\
    F^+_{j,u} &=& (r'-\frac 2 3 \delta) \cis (\pi + (u+1) \cdot \phi), \text{ if $j$ and $u$ are even and } Y_{\frac j 2}[\frac u 2] = 0 \\
    G^+_j &=& g_j \cis (\pi + (2W+3)\cdot \phi), \text{ where $g_j$ is defined below} \\
    \end{array}
\]

Now we define $b_i$, $c_i$, $g_j$, and $e_j$. 

Let $B_\varepsilon = (r-\varepsilon)\cis((4W+5) \cdot \phi)$. For $1 \leq i \leq 2n+3$, define $I^+_i$ so that $C,I^+_1,I^+_2,\ldots,I^+_{2n+3},B_\varepsilon$ are evenly spaced along the segment $CB_\varepsilon$. For $1 \leq i \leq 2n+3$ define $b_i$ so that $||B_i I_i|| = d$. For $1 \leq i \leq 2n+1$ define $c_i$ so that $||C_{i+2} I_i|| = d$. Define $c_1 = c_2 = c_3$. 

For all $1 \leq j \leq 2n+1$, define $J^+_{j} = (r' - \frac {j} {2n+2} \delta) \cis (\pi+(2W+4)\cdot \phi)$ to be a point on $A_1^-A_{1,1}^+$. Define $K_j$ to be the point on $GA_{1,1}$ such that $||K_jJ^+_{j}|| = d$. Define $g_j$ to be the positive real so that $||G^+_jK_{2n+2-j}|| = d$. 

For all $1 \leq j \leq 2n+1$, define $L^+_j = (r' - \frac {2n+4-j} {2n+2} \delta) \cis (\pi)$ to be a point on $D^-D_1^+$. Define $M_j$ to be the point on $DE$ such that $||M_jL^+_{j}|| = d$. Define $e_j$ to be the positive real so that $||E_j^+M_{2n+2-j}|| = d$. 

Now we are ready to define the second subtrajectory $T_2$.
\begin{align*}
T_2 = &\bigcirc_{1 \leq i \leq 2n+1} \big(
            G^-
            \circ G_i^+
            \circ G^-
            \circ \bigcirc_{1 \leq u \leq 2W+1}( A_u^- \circ A^+_{i,u} \circ A_u^-)
            \circ B^-
            \circ B^+_i
            \circ B^-
            \circ \bigcirc_{1 \leq u \leq 3}
            (
                P_i
                \circ Q_i
            )
        \big)
        \\
        &\circ
        \bigcirc_{1 \leq j \leq 2n} \big(
            E^-
            \circ E^+_j
            \circ E^-
            \circ \bigcirc_{1 \leq u \leq 2W+1} (F^-_u \circ F^+_{j,u} \circ F^-_{u})
            \circ G^-
            \circ G^+_j
            \circ G^-
        \big)
        \\
        &
        \circ C^-
        \circ C^+_1
        \circ C^-
        \circ D^-
        \circ D^+_1
        \circ D^-
        \circ C^-
        \circ C^+_2
        \circ C^-
        \circ D^-
        \circ D^+_2
        \circ D^-
        \circ H^-
        \\
        &
        \circ E^-
        \circ E^+_{2n+1}
        \circ E^-
        \circ \bigcirc_{1 \leq u \leq 2W+1} (F^-_u \circ F^+_{2n+1,u} \circ F^-_{u})
        \circ G^-
        \circ G^+_{2n+1}
        \circ G^-
        \\
        &\circ
        \bigcirc_{3 \leq i \leq 2n+2} \big(
            C^-
            \circ C^+_i
            \circ C^-
            \circ G^-
            \circ D^-
            \circ D^+_i
            \circ D^-
        \big)
        \\
        &\circ C^-
        \circ C^+_{2n+3}
        \circ C^-
        \circ G^-
        \circ D^-
        \circ D^+_{2n+3}
\end{align*}

where $P_i$ and $Q_i$ are points on $B^-G^-$ so that, for all $1 \leq i \leq n$, the subtrajectory from $A_{2i-1,1}^+$ to $D^+_{2i+3}$ has total length $\lambda$, for some constant $\lambda$. Define $T = T_1 \circ D^- \circ T_2$. Define $m = 2nW + 2$. Define $\ell = \lambda - \delta$. Recall that $d = ||r \cis \pi - r' \cis \phi||.$ Then our constructed instance is $(T,m,\ell,d)$.

We show that $B_i$, $E_i$, $G_j$ and $K_j$ are well defined in Section~\ref{sec:3ov_well_defined}, and show that $P_i$ and $Q_i$ are well defined in Section~\ref{sec:3ov_subtrajectory_property}.

\subsection{Points $B_i$, $E_i, G_j$ and $K_j$ are well defined}
\label{sec:3ov_well_defined}

We start with a useful lemma for constructing a point $X$ on a segment so that, for some other point $P$, the distance from $P$ to $X$ is exactly $d$. 

\begin{lemma}
\label{lem:well_defined_ivt}
Suppose $P$, $A$, $B$ are arbitrary points in the plane and let $d$ be a given constant. Suppose further that $|PA| < d < |PB|$. Then there exists a point $X$ on the segment $AB$ so that $|PX| = d$.
\end{lemma}

\begin{proof}
As $X$ varies along the segment $AB$, the distance of $P$ to $X$ is a continuous function. This function starts at a value less than $d$ (when $X = A$) and ends at a value greater than $d$ (when $X=B$). By the intermediate value theorem, there is a point so that $|PX| = d$.
\end{proof}

Next, we apply this lemma repeatedly to show that $B_i$, $E_i$, $G_j$ and $K_j$ are well defined. 

\begin{lemma}
The points $B_i$, $C_i$, $E_j$ and $G_j$ are well defined.
\end{lemma}

\begin{proof}
Let $O$ be the origin. The points $B_i$ and $C_i$ can be shown to be well defined in a similar way, so we focus only on $B_i$. The definition of $B_i$ is $b_i \cis (\pi + 4W+6 \cdot \phi)$ so that $||I_iB_i|| = d$. Let $B^+ = r' \cis (\pi + 4W+6 \cdot \phi)$. To show that $B_i$ is well defined, we apply Lemma~\ref{lem:well_defined_ivt} on segment $OB^+$. We know $||I_iO|| < r < d$, and $||I_iB^+|| > ||OB^+|| > r' > d$ since $\angle I_iOB^+ > \frac \pi 2$. Therefore, there exists a point $B_i$ on $OB^+$ so that $||I_iB_i|| = d$. 

The points $E_j$ and $G_j$ can be shown to be well defined in a similar way, so we focus only on $E_j$. For $E_j$, we first show that $M_j$ is well defined, then we show $E_j$ is well defined. To show that $M_j$ is well defined, we apply Lemma~\ref{lem:well_defined_ivt} on segment $DE$. We know $||DL_j^+|| > ||DD^-|| = d$, and $||EL_j^+|| < ||ED^+|| = d$. Hence, there exists $M_j$ on $DE$ so that $||M_jL_j^+|| = d$. Let $E^+ = r' \cis (\pi + \phi)$. To show that $E_j$ is well defined, we apply Lemma~\ref{lem:well_defined_ivt} on segment $OE^+$. We know $||M_jO|| < r < d$, and $||M_jE^+|| >||OE^+|| > r' > d$ since $\angle M_jOE^+ > \frac \pi 2$. Therefore, there exists a point $E_j$ so that $||E_jM_j|| =d$.
\end{proof}

\subsection{Antipodal property}
\label{sec:3ov_antipodal_property}

The subtrajectories $T_1$ and $T_2$ were constructed in such a way that most points are within distance $d$ of one another, whereas a few pairs of points are of distance greater than $d$ from one another. We call these pairs of points antipodes, which we formally define below.

\begin{definition}
\label{defn:antipodes}
Define the following pairs of vertices of $T$ to be antipodes: 
$(A_{i,u}^+, A_{h,u})$, 
$(B_i^+,B_{k,h})$,
$(C_i^+,C)$,
$(D^+,D)$,
$(E_j^+,E)$,
$(F_{j,u}^+,F_{h,u})$,
$(G_j^+,G)$, and
$(v^+,H_u)$, where $u \in \{1,2\}$ and $v^+$ is any vertex of $T_2$.
\end{definition}

Next, we prove that these antipodes are indeed the only pairs of points with distance greater than $d$ from one another.

\begin{lemma}
\label{lem:3ov_antipodal_property_of_vertices}
Suppose $v$ is a vertex of $T_1$ and $v^+$ is a vertex of $T_2$, so that the distance between $v^+$ and $v$ is greater than $d$. Then $v^+$ and $v$ are antipodes.
\end{lemma}

\begin{proof}
Note that if $v = H_u$ for some $u \in \{1,2\}$, then we immediately get that $v^+$ and $v$ are antipodes. So for the remainder of this proof we assume $v \neq H_u$.

Let $O$ be the origin. We will show that $v^+$, $O$ and $v$ are collinear, in that order. Suppose the contrary. Then $\angle v^+Ov \leq \pi - \phi$, since all vertices in our construction lie on rays emanating from the origin that are $\phi$ radians apart. Now, $||Ov|| \leq r$, $||Ov^+|| \leq r'$ and $\angle v^+Ov \leq \pi - \phi$. Therefore, $||v^+v|| \leq ||r \cis \pi - r' \cis \phi|| = d$, which is a contradiction. Hence, $v^+$, $O$, $v$ are collinear, in that order. Moreover, $||v^+v|| > d$. The only pairs of vertices $(v^+,v)$ satisfying these two properties are the antipodes listed in Definition~\ref{defn:antipodes}.
\end{proof}

Finally, we prove a property of these antipodal pairs, which will help us better describe the structure of the free space diagram in Section~\ref{sec:3ov_free_space_diagram}. 

\begin{lemma}
\label{lem:3ov_antipodal_property_for_segments}
Suppose $a$ is a point on $T_1$ and $b$ is a point on $T_2$, not necessarily vertices. Suppose that $|ab| > d$. Then there exists an endpoint $E_a$ of the segment containing $a$, and an endpoint $E_b$ of the segment containing $b$, so that $(E_b, E_a)$ are antipodes. 
\end{lemma}

\begin{proof}
Suppose for the sake of contradiction that there does not exist an endpoint of the segment containing $a$ and an endpoint of the segment containing $b$ that are antipodes. Let $a$ be on $a_1a_2$ and $b$ be on $b_1b_2$. Since none of $(a_u,b_w)$ are antipodes for $1 \leq u,w \leq 2$, we have by Lemma~\ref{lem:3ov_antipodal_property_of_vertices} that $||a_ub_w|| \leq d$. The distance of point $a_1$ to $b_1b_2$ is maximised at the endpoints $b_1$ or $b_2$. So $||a_1b|| \leq d$, $||a_2b|| \leq d$. But now, the distance of point $b$ to $a_1a_2$ is maximised at the endpoints. So $||ab|| \leq d$. This contradiction means that our initial assumption cannot hold, and there is a pair of endpoints $(E_b,E_a)$ which are antipodes.
\end{proof}

\subsection{Subtrajectory length property}
\label{sec:3ov_subtrajectory_property}

The main lemma in this section is Lemma~\ref{lem:3ov_p_q_well_defined}, which states that $P_i$ and $Q_i$ are well defined. Recall that $P_i$ and $Q_i$ are points on $B^-G^-$ so that, for all $1 \leq i \leq n$, the subtrajectory from $A_{2i-1,1}^+$ to $D_{2i+3}^+$ has total length $\lambda$, for some constant $\lambda$. We first show a useful lemma that we will use to prove the main lemma.

\begin{lemma}
\label{lem:3ov_subtrajectory_lemma}
Suppose $A$ and $B$ are arbitrary points in the plane and $\lambda$ is a positive real satisfying $|AB| < \lambda < 7 \cdot |AB|$. Then there exist points $P$ and $Q$ on $AB$ so that the curve $A \circ \bigcirc_{1 \leq u \leq 3} (P \circ Q) \circ B$ has total length $\lambda$.
\end{lemma}

\begin{proof}
Start with $P=Q=A$. At first, the total length of $A \circ \bigcirc_{1 \leq u \leq 3} (P \circ Q) \circ B$ is $|AB|$. Now, move $P$ continuously from $A$ to $B$. When $P$ reaches $B$, the total length is $7|AB|$. By the intermediate value theorem, there is a position of $P$ so that the total length of $A \circ \bigcirc_{1 \leq u \leq 3} (P\circ Q) \circ B$ is $\lambda$.
\end{proof}

Now we are ready to prove the main lemma of this section.

\begin{lemma}
\label{lem:3ov_p_q_well_defined}
The points $P_i$ and $Q_i$ are well defined for $1 \leq i, j \leq 2n+1$.
\end{lemma}

\begin{proof}
Define $P_{2n+1} = Q_{2n+1} = B^-$. Similarly, define $P_{2n} = Q_{2n} = P_{2n-1} = Q_{2n-1} = B^-$. Define $\lambda$ to be the length of the subtrajectory from $A_{2n-1,1}^+$ to $D_{2n+3}^+$. We define $P_i$ and $Q_i$ inductively, starting at $i = 2n-1$ and working down to $i=1$, so that the subtrajectory from $A_{2i-1,1}^+$ to $D_{2i+3}^+$ has length $\lambda$. In the base case of $i=2n-1$, the subtrajectory has length $\lambda$ by definition.

In the inductive case, assume that the subtrajectory from $A_{2i+1,1}^+$ to $D_{2i+5}^+$ has length $\lambda$. We would like to define $P_i$ and $Q_i$ so that the subtrajectory from $A_{2i-1,1}^+$ to $D_{2i+3}^+$ also has length $\lambda$. This is equivalent to the subtrajectory from $A_{2i-1,1}^+$ to $A_{2i+1,1}^+$ and the subtrajectory from $D_{2i+3}^+$ to $D_{2i+5}^+$ having the same length. We will approximate the lengths of these subtrajectories, and show that $P_i$ and $Q_i$ can be defined using Lemma~\ref{lem:3ov_subtrajectory_lemma} to make the subtrajectories the same length.

First, we approximate the length of the subtrajectory from $A_{2i-1,1}^+$ to $A_{2i+1,1}^+$. We ignore $P_i$ and $Q_i$ initially, and add its contribution later. The length is dominated by the distances between the vertices $A_{2i-1,u}^+$ for $1 \leq u \leq 2W-1$, $B_{2i-1}^+$, $G_{2i-1}^+$, $A_{2i,u}^+$ for $1 \leq u \leq 2W-1$, $B_{2i}^+$, $G_{2i}^+$ and finally $A_{2i+1,u}^+$. Formally, we can set $r' >> r$ so that $\delta$ approaches zero, so $B^-$ is arbitrarily close to $B^+$. With this simplification in mind, the subtrajectory is approximately a closed loop that visits $B^+$ twice and $G^+$ twice. So a lower bound for the length of the subtrajectory is $4||B^+G^+|| \approx 4 \sqrt 2 r'.$ An upper bound for the length of the subtrajectory is to replace the segments of the subtrajectory with an arc of the circle centered at the origin with radius $r'$. This upper bound is a closed loop on the circumference that visits $B^+$ twice and $G^+$ twice. So an upper bound for the length of the subtrajectory is four times the arc from $B^+$ to $G^+$, which is approximately $2 \pi r'.$ Let the length of $A_{2i-1,1}^+$ to $A_{2i+1,1}^+$, ignoring the contributions of $P_i$ and $Q_i$, be $\lambda_i$, so that $4 \sqrt 2 r' \leq \lambda_i \leq 2 \pi r'$.

Next, we approximate the length of the subtrajectory from $D_{2i+3}^+$ to $D_{2i+5}^+$. The length is dominated by the distances between the vertices $D^+$, $G^-$, $C^+$, $G^-$, $D^-$, $G^-$, $C^+$, $G^-$ and back to $D^+$. The length of this closed loop is approximately $8||D^+G^-|| = 8 \sqrt 2 r'$. 

Now we use the same idea as Lemma~\ref{lem:3ov_subtrajectory_lemma}. Start with $P_i = Q_i = B^-$. Then the length of the subtrajectory from $A_{2i-1,1}^+$ to $A_{2i+1,1}^+$ is simply $\lambda_i$. Now, move $P_i$ continuously from $B^-$ to $G^-$. When $P_i$ reaches $G^-$, then the length of the subtrajectory from $A_{2i-1,1}^+$ to $A_{2i+1,1}^+$ is $\lambda_i + 6 |B^-G^-| \approx \lambda_i + 6 \sqrt 2 r'$. Therefore, as $P_i$ moves continuously, the initial length of $A_{2i-1,1}^+$ to $A_{2i+1,1}^+$ is at most $2 \pi r'$, and its final length is at least $10 \sqrt 2 r'$. By the intermediate value theorem, there exists a position of $P_i$ so that the lengths of the subtrajectories from $A_{2i-1,1}^+$ to $A_{2i+1,1}^+$ and from $D_{2i+3}^+$ to $D_{2i+5}^+$ are the same.
\end{proof}

Next, we add indices to the trajectory $T$. The trajectory below is exactly the same as the trajectory defined in Section~\ref{sec:3ov_construction}, in that it is the same vertices in the same order. The additional indices, in the form of subscripts, helps us describe the rows and columns in the free space diagram later on. For example, the point $G$ appears as a vertex in the trajectory $T$ many times, and we let $G_{k,h}$ denote the $(kW + k + h)^{th}$ time $G$ appears in $T$. Applying a similar set of indices to all vertices, we yield the following:

\begin{align*}
T = &\bigcirc_{1 \leq k \leq n} \Big(
    \bigcirc_{1 \leq h \leq W} \big(
        G_{k,h}
        \circ \bigcirc_{1 \leq u \leq 2W+1} (A_{k,h,u})
        \circ B_{k,h}
        \circ C_{k,h,1}
        \circ D_{k,h,1}
        \circ C_{k,h,2}
        \circ D_{k,h,2}
        \\ &\hspace{4cm}
        \circ E_{k,h}
        \circ \bigcirc_{1 \leq u \leq 2W+1} (F_{k,h,u})
    \big) 
        \\ &\hspace{1.8cm}
    \circ G_{k,W+1}
    \bigcirc_{1 \leq u \leq 2W+1} (A_{k,W+1,u})
    \circ B_{k,W+1}
    \circ C_{k,W+1,1}
    \circ D_{k,W+1,1}
    \circ C_{k,W+1,2}
        \\ &\hspace{1.8cm}
    \circ D_{k,W+1,2}
    \circ C_{k,W+1,3}
    \circ H_{k,1}
    \circ H_{k,2}
\Big) \circ D^-_0
\\
&\circ \bigcirc_{1 \leq i \leq 2n+1} \big(
            G^-_{i,1}
            \circ G_{i,1}^+
            \circ G^-_{i,1}
            \circ \bigcirc_{1 \leq u \leq 2W+1}( A_{i,u}^- \circ A^+_{i,u} \circ A_{i,u}^-)
            \circ B^-_i
            \circ B^+_i
            \circ B^-_i
        \\ &\hspace{1.8cm}
            \circ \bigcirc_{1 \leq u \leq 3}
            (
                P_{i,u}
                \circ Q_{i,u}
            )
        \big) 
        \\
        &\circ
        \bigcirc_{1 \leq j \leq 2n} \big(
            E^-_j
            \circ E^+_j
            \circ E^-_j
            \circ \bigcirc_{1 \leq u \leq 2W+1} (F^-_{j,u} \circ F^+_{j,u} \circ F^-_{j,u})
            \circ G^-_{j,2}
            \circ G^+_{j,2}
            \circ G^-_{j,2}
        \big) 
        \\
        &\circ C^-_1
        \circ C^+_1
        \circ C^-_1
        \circ D^-_1
        \circ D^+_1
        \circ D^-_1
        \circ C^-_2
        \circ C^+_2
        \circ C^-_2
        \circ D^-_2
        \circ D^+_2
        \circ D^-_2
        \circ H^-_1
        \\
        &
        \circ E^-_{2n+1}
        \circ E^+_{2n+1}
        \circ E^-_{2n+1}
        \circ \bigcirc_{1 \leq u \leq 2W+1} (F^-_{2n+1,u} \circ F^+_{2n+1,u} \circ F^-_{2n+1,u})
        \circ G^-_{2n+1,2}
        \circ G^+_{2n+1,2}
        \\
        &\circ G^-_{2n+1,2}
        \circ
        \bigcirc_{3 \leq i \leq 2n+3} \big(
            C^-_i
            \circ C^+_i
            \circ C^-_i
            \circ G^-_{i,3}
            \circ D^-_i
            \circ D^+_i
            \circ D^-_i
            \circ G^-_{i,4}
        \big)
        \\
        &\circ C^-_{2n+3}
            \circ C^+_{2n+3}
            \circ C^-_{2n+3}
            \circ G^-_{2n+3,3}
            \circ D^-_{2n+3}
            \circ D^+_{2n+3}
\end{align*}

Similarly, the points $J^+_j$ and $L_j^+$ appear many times in the trajectory $T$. In particular, the segment $A^+_{2i-1,1}A^-_{2i-1,1}$ passes through $J^+_j$ for all $1 \leq i,j \leq n$, and the segment $D^-_{2i+3}D^+_{2i+3}$ passes through $L_j^+$ for all $1 \leq i,j \leq n$. Define $J^+_{2i-1,j}$ to coincide with $J^+_j$ in the complex plane, so that $A^+_{2i-1,1}$, $J^+_{2i-1,j}$ and $A^-_{2i-1,1}$ are in that order along the trajectory $T$. Similarly, define $L^+_{2i+3,j}$ to coincide with $L_j^+$ in the complex plane, so that $D^-_{2i+3}$, $L^+_{2i+3,j}$ and $D^+_{2i+3}$ are in that order along the trajectory $T$.

As a direct result of the main lemma of this section, and the definitions of $P_i$, $Q_i$, $J^+_{2i-1,j}$ and $L^+_{2i+3,j}$, we obtain the following subtrajectory length property. As mentioned in our third key component in Section~\ref{sec:overview_3ov}, we consider $n^2$ reference subtrajectories in our 3OV reduction. In particular, our $n^2$ reference subtrajectories will be the pairs of vertical lines where $l_s$ passes through $J^+_{2i-1,2j-1}$, and $l_t$ passes through $L^+_{2i+3,2j+1}$, for $1 \leq i,j, \leq n$. Lemma~\ref{lem:sub_length_j_to_l} states that if we pick $l_s$ and $l_t$ in this way, we obtain reference subtrajectories of length $\ell$ for all $1 \leq i,j \leq n$. 
\begin{lemma}
\label{lem:sub_length_j_to_l}
For all $1 \leq i,j \leq n$, the length of the subtrajectory from $J^+_{2i-1,2j-1}$ to $L^+_{2i+3,2j+1}$ is $\ell$. 
\end{lemma}

\begin{proof}
We already have that the subtrajectory from $A_{2i-1,1}^+$ to $D_{2i+3}^+$ has length $\lambda$, for all $1 \leq i \leq n$. The distance from $A_{2i-1,1}^+$ to $J_j^+$ is $\frac {2j-1} {2n+2} \delta$. The distance from $L_j^+$ to $D_{2n+3}^+$ is $\frac {2n+4-2j-1}{2n+2} \delta$. Therefore, the length of the subtrajectory from $J_j^+$ to $L_j^+$ is $\lambda - \frac {2j-1} {2n+2} \delta - \frac {2n+4-2j-1}{2n+2} \delta = \lambda - \delta = \ell$ as required.
\end{proof}

We show another subtrajectory length property. This bound shows that $l_s$ and $l_t$ cannot be between the vertical lines through the $x$-coordinates corresponding to $H_{k,1}$ and $H_{k+1,2}$. It involves modifying the construction of the base case in Lemma~\ref{lem:3ov_p_q_well_defined}.

\begin{lemma}
\label{lem:sub_length_h_h}
For all $1 \leq k \leq n$, the length of the subtrajectory from $H_{k,1}$ to $H_{k+1,2}$ is less than $\ell$.
\end{lemma}

\begin{proof}
Consider the subtrajectory from $H_{k,1}H_{k+1,2}$. All points on this subtrajectory, other than $H_{k,1}$, $H_{k,2}$, $H_{k+1,1}$ and $H_{k+1,2}$ lie within a circle of radius $r$ centered at the origin. We can set $r'$ to be arbitrarily large relative to $r$, so that the length of the subtrajectory is dominated by the distances between the origin, $H_{k,1}$, $H_{k,2}$, $H_{k+1,1}$ and $H_{k+1,2}$. Therefore, the length of the subtrajectory is at most $40 r'$, ignoring terms involving~$r$. It suffices to show that  $\ell > 40r'$.

Recall that the subtrajectory from $A_{2n-1,1}^+$ to $D^+_{2n+3}$ is defined to have length $\lambda = \ell + \delta$. Recall that for all $1 \leq i \leq 2n+3$, $P_i$ and $Q_i$ were defined so that the subtrajectory from $A_{2i-1,1}^+$ to $D_{2i+3}^+$ also has length $\lambda$. We will modify the subtrajectory $A_{2n-1,1}^+$ to $D^+_{2n+3}$ so that $\lambda > 41r'$, and therefore, $\ell > 40r'$.

To do this, instead of initialising our induction in Lemma~\ref{lem:3ov_p_q_well_defined} with $P_{2n+1} = Q_{2n+1} = B^-$ we place $P_{2n+1}$ and $Q_{2n+1}$ at $B^-$ and $G-$ respectively. Moreover, instead of placing three copies of each of $P_{2n+1}$ and $Q_{2n+1}$, we place 41 copies. Each subtrajectory $P_{2n+1} \circ Q_{2n+1} \circ P_{2n+1}$ has length at least $r'$, so $\lambda > 41r'$, as required.
\end{proof}

\subsection{Free space diagram}
\label{sec:3ov_free_space_diagram}

In this section, we describe the free space diagram $F_d(T,T)$, in other words, the free space diagram between the trajectory $T$ and itself, with distance parameter~$d$. For any pair of points $x$ and $y$ on the trajectory $T$, not necessarily vertices of $T$, define $f(x,y)$ to be the point in $F_d(T,T)$ with $x$-coordinate associated with point $x$ and $y$-coordinate associated with point $y$. We introduce the following notation to describe the regions of non-free space in the free space diagram.

\begin{definition}
Suppose $f(x,y)$ is not in the free space of $F_d(T,T)$. Define $\Diamond(x,y)$ to be the region of non-free space that contains $f(x,y)$. 
\end{definition}

Next, we show the non-free space $\Diamond(x,y)$ is always associated with an antipodal pair.

\begin{lemma}
All regions of non-free space in $F_d(T_2,T_1)$ are $\Diamond(v^+,v)$ where the pair $(v^+,v)$ are antipodes.
\end{lemma}

\begin{proof}
Suppose $f(x,y)$ is not in free space. Then by Lemma~\ref{lem:3ov_antipodal_property_for_segments}, there exists an endpoint of the segment containing $x$ and an endpoint of the segment containing $y$ that are antipodes. Let $(v^+, v)$ be this pair of antipodes. Let $x$ be on the segment $v^-v^+$ and let $y$ be on the segment $uv$. We can verify that for all $(v^+,v)$ that are antipodes in construction in Section~\ref{sec:3ov_construction}, by moving from $v^-$ to $v^+$ we move further away from (any points on) segment $uv$, and similarly, by moving from $u$ to $v$ we move further away from (any point on) segment $v^-v^+$. 

Using this fact, we can prove that desired lemma. Since $f(x,y)$ is in non-free space, we have $||xy|| > d$. By moving from $x$ to $v^+$, we move further away from $y$, which is on $uv$. Hence, $||v^+y|| >d$, and all points between $f(x,y)$ and $f(v^+,y)$ are in non-free space. Finally, by moving from $y$ to $v$, we move further away from $v^+$. Hence, all points between $f(v^+,y)$ and $f(v^+,v)$ are non-free space. Therefore, we have constructed a path of non-free space connecting $f(x,y)$ and $f(v^+,v)$, so $f(x,y)$ must be contained in $\Diamond(v^+,v)$.
\end{proof}

Now we show that the non-free space $\Diamond(v^+,v)$ is in fact a diamond. In particular, it is diamond shaped, with concave (inwards-curved) sides.  See Figure~\ref{fig:lower_diamonds}.

\begin{figure}[ht]
    \centering
    \includegraphics{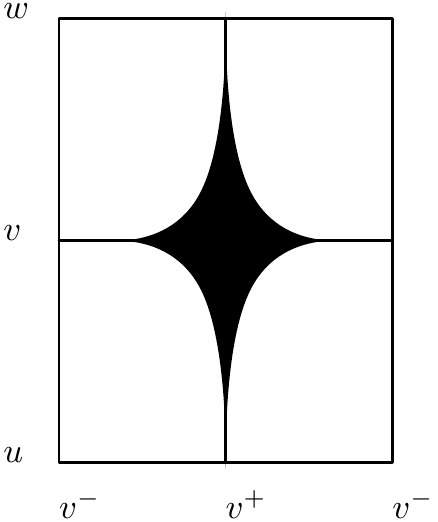}
    \caption{$\Diamond(v^+,v)$ is diamond shaped, with concave (inwards-curved) sides.}
    \label{fig:lower_diamonds}
\end{figure}

\begin{lemma}
Suppose $(v^+,v)$ are antipodes, $v \neq H_{u,1}$ and $v \neq H_{u,2}$ for any $1 \leq u \leq n$. Then $\Diamond(v^+,v)$ is diamond shaped, with concave (inwards-curved) sides. Moreover, the $x$-coordinates spanned by $\Diamond(v^+,v)$ corresponds to a subset of the neighbouring segments of $v^+$; the $y$-coordinates spanned by $\Diamond(v^+,v)$ corresponds to a subset of the neighbouring segments of $v$.
\end{lemma}

\begin{proof}
Lemma~\ref{lem:3ov_antipodal_property_for_segments} immediately implies that the $x$ and $y$-coordinates of $\Diamond(v^+,v)$ is spanned by the neighbouring segments of $v^+$ and $v$ respectively. It suffices to show that the shape of this non-free space is diamond shaped, with concave (inwards-curved) sides.

Draw a vertical line $\Diamond(v^+,v)$ through the $x$-coordinate~$v^+$, and a horizontal line through the $y$-coordinate~$v$. This divides $\Diamond(v^+,v)$ into four quadrants. We will show that each quadrant consists of a continuous, $xy$-monotone side. Then we will show that the sides are concave.

Let the neighbouring segments of $v^+$ in $T_2$ be $v^-v^+$ and $v^+v^-$. Let the neighbouring segments of $v$ in $T_1$ be $uv$ and $vw$. We can verify that for all $(v^+,v)$ that are antipodes in our construction, that the point $v^-$ is strictly closer to $u,v,w$ than $v^+$. We can also verify for all antipodes that $u$ and $w$ are closer to $v^+$ than $v$. 

Let us focus on the top right quadrant. This is the free space with $x$-coordinates associated with the segment $v^+v^-$, and $y$-coordinates associated with the segment $vw$. Along the edge $v^+v^-$, the distance to any fixed point on $vw$ is strictly increasing. Similarly, along the edge $vw$, the distance to any fixed point on $v^+v^-$ is strictly increasing. Therefore, the boundary of $\Diamond(v^+,v)$, which is the set of all $x \in v^+v^-$ and $y \in vw$ so that $||xy||=d$, is a continuous curve starting at the same $y$-coordinate as $v^+$, and moving down and to the right. This gives one of the four curved sides of the diamond. Moreover, we know that in the free space diagram, the free space must be the intersection of an ellipse with the cell. Hence, this $xy$-monotone side is concave, inward-facing, and the boundary of an ellipse. Repeating this for all four quadrants gives the four continuous, $xy$-monotone and concave (inwards-curved) sides of $\Diamond(v^+,v)$.
\end{proof}

Finally, we describe the free space and non-free space on the vertical line through the $x$-coordinate associated with the point $H_{k,2}$ for some $1 \leq k \leq n$. We the vertical line alternates between free and non-free space, where the free space corresponds with $H_{u,2}$ for some $u$, as shown in Figure~\ref{fig:lower_bound3}.

\begin{figure}[ht]
    \centering
    \includegraphics[width=0.4\textwidth]{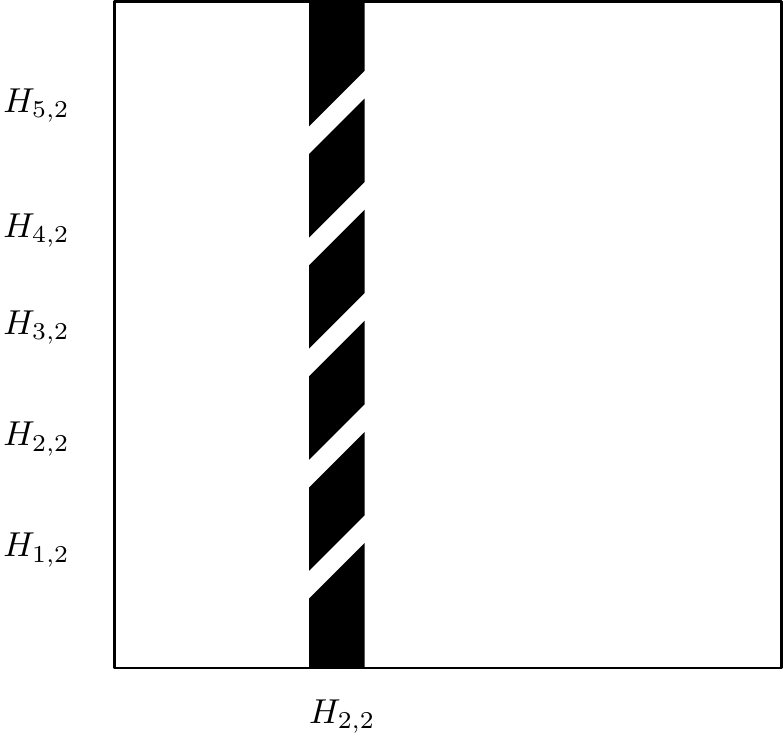}
    \caption{An example of the vertical line from $H_{k,2}$ passing through alternating regions of free and non-free space, where $k=2$ and $n=5$.}
    \label{fig:lower_bound3}
\end{figure}

\begin{lemma}
\label{lem:3ov_fs_Hk}
Consider the vertical line in the free space diagram $F_d(T,T)$ with $x$-coordinate corresponding to $H_{k,2}$. This vertical line consists of alternating regions of free space and non-free space, with $n$ regions of contiguous free space and $n+1$ regions of contiguous non-free space. Moreover, each of the $n$ regions of contiguous free space contains the point $f(H_{k,2}, H_{u,2})$ for some $1 \leq u \leq n$. 
\end{lemma}

\begin{proof}
Consider the vertices of $T$. Divide the vertices into two sets, $H_2 = \{H_{u,2} \text{ for some } 1 \leq u \leq n \}$, and the other vertices $T \setminus H_2$. For all $v \in T \setminus H_2$, the vertex $v$ is within the ball of radius $4r'$ centered at the original, so the distance from $v$ to $H_{k,2}$ is greater than $d$. However, for all $v \in H_2$, the distance from $v$ to $H_{k,2}$ is zero.

Let $x$ be an arbitrary point on $T$. If $x$ is on a segment where both endpoints are in $T \setminus H_2$, then $f(H_{k,2},x)$ is non-free space. If one of the endpoints are in $T \setminus H_2$, then $x$ will only be free space if it is within distance $d$ to $H_{k,2}$. These regions of close proximity only occur when one of the endpoints is $H_{u,2}$. So there are exactly $n$ segments of free space, and $n+1$ segments of non-free space on the vertical line with $x$-coordinate $H_{k,2}$.
\end{proof}

\subsection{Paths in free space}
\label{sec:3ov_paths}

In the following section (Section~\ref{sec:3ov_if}), we will show that if $(\mathcal X, \mathcal Y, \mathcal Z)$ is a YES instance for 3OV, then $(T,m,\ell,d)$ is a YES instance for \problemtwo. This involves constructing a pair of vertical lines $l_s$ and $l_t$, and $m-1$ monotone paths from $l_s$ to $l_t$, with the property that the $m-1$ monotone paths intersect in their $y$-coordinates in at most one point, and the each of the $m-1$ monotone paths intersect with the $y$-interval from $s$ to $t$ in at most one point.

In this section, we compile a list of monotone paths, and prove their correctness. Some of these monotone paths are conditioned on one of the booleans $X_i[h]$, $Y_j[h]$ or $Z_k[h]$ being equal to zero. The list of monotone paths from this section will be used extensively in the subsequent section.

To aid with our description of monotone paths in both this section and the subsequence section, we use the following notation. Suppose $(v^+,v)$ are antipodes, $v \neq H_{u,1}$ and $v \neq H_{u,2}$ for any $1 \leq u \leq n$. Define the top, bottom, left and right corners of $\Diamond(v^+,v)$ to be $\Diamond_T(v^+,v)$, $\Diamond_B(v^+,v)$, $\Diamond_L(v^+,v)$ and $\Diamond_R(v^+,v)$ respectively. 

We say $f(x,y) \to f(w,z)$ if there is a monotone path from $f(x,y)$ to $f(w,z)$. 

\begin{lemma}
\label{lem:3ov_fs}
For $1 \leq i,j,k \leq n$ and $1 \leq h \leq W$:
\begin{enumerate}[label=(\alph*)]
    \item \label{lem:3ov_fs_aa_to_bb1}
$\Diamond_R(A_{2i-1,1}^+,A_{k,h,1}) \to \Diamond_T(B_{2i}^+,B_{k,h})$,
    \item \label{lem:3ov_fs_aa_to_bb2}
$\Diamond_R(A_{2i-1,1}^+,A_{k,h,1}) \to \Diamond_T(B_{2i+1}^+,B_{k,h})$ if $X_i[h] = 0$,
    \item \label{lem:3ov_fs_bb_to_cc1}
$\Diamond_T(B_{i}^+,B_{k,h}) \to \Diamond_B(C_{i+1}^+,C_{k,h,1})$,
    \item \label{lem:3ov_fs_bb_to_cc3}
$\Diamond_T(B^+_{2i},B_{k,h}) \to \Diamond_B(C_{2i+2}^+,C_{k,h,1})$ if $Z_k[h] = 0$,
    \item \label{lem:3ov_fs_cc_to_dd1}
$\Diamond_B(C_{i+2}^+,C_{k,h,1}) \to \Diamond_L(D_{i+3}^+,D_{k,h,1})$,
    \item \label{lem:3ov_fs_aa_to_dd1}
$\Diamond_R(A_{2i-1,1}^+,A_{k,h,1}) \to \Diamond_L(D_{2i+3}^+,D_{k,h,2})$.
	\item \label{lem:3ov_fs_aa_to_dd2}
$\Diamond_R(A_{2i-1,1}^+,A_{k,h,1}) \to \Diamond_L(D_{2i+3}^+,D_{k,h,1})$ if $X_i[h] \cdot Z_k[h] = 0$.
	\item \label{lem:3ov_fs_aa_to_dd3}
$\Diamond_L(A_{2i-1,2}^+,A_{k,h,2}) \to \Diamond_T(D_{2i+2}^+,D_{k,h,2})$.
	\item \label{lem:3ov_fs_ee_to_gg1}
$\Diamond_B(E^+_{2j-1},E_{k,h}) \to \Diamond_T(G^+_{2j,2}, G_{k,h})$.
	\item \label{lem:3ov_fs_ee_to_gg2}
$\Diamond_B(E^+_{2j-1}, E_{k,h}) \to \Diamond_T(G^+_{2j+1,2}, G_{k,h})$ if $Y_j[h] = 0$.
	\item \label{lem:3ov_fs_ed_to_gg1}
$\Diamond_L(E_1^+, E_{k,h}) \to \Diamond_T(G_{1,2}^+,G_{k,h})$.
	\item \label{lem:3ov_fs_ed_to_gg2}
$f(E_{1}^+,D_{k,h,2}) \to \Diamond_T(G_{2n+1,2}^+,G_{k,h})$.
\end{enumerate}
\end{lemma}

\begin{proof}\quad

\begin{itemize}
\item[\emph{\ref{lem:3ov_fs_aa_to_bb1}}]
$
\Diamond_R(A_{2i-1,1}^+,A_{k,h,1}) 
\to \Diamond_L(A_{2i,1}^+,A_{k,h,1})
\to \Diamond_L(A_{2i,2}^+,A_{k,h,2}) 
\to \ldots 
\to \Diamond_L(A_{2i,2W+1}^+,A_{k,h,2W+1}) \\
\to \Diamond_T(B_{2i}^+,B_{k,h})
$

\item[\emph{\ref{lem:3ov_fs_aa_to_bb2}}]
Since $X_i[h] = 0$, we have that $\Diamond(A_{2i,2h}^+,A_{k,h,2h})$ is non-existent by construction. Therefore, $\Diamond_T(A_{2i,2h-1}^+,A_{k,h,2h-1})
\to \Diamond_B(A_{2i,2h+1}^+,A_{k,h,2h+1})$. Hence, we have $
\Diamond_R(A_{2i-1,1}^+,A_{k,h,1}) 
\to \Diamond_L(A_{2i,1}^+,A_{k,h,1})
\to \Diamond_L(A_{2i,2}^+,A_{k,h,2}) 
\to \ldots
\to \Diamond_L(A_{2i,2h-1}^+,A_{k,h,2h-1}) 
\to \Diamond_T(A_{2i,2h-1}^+,A_{k,h,2h-1}) \\
\to \Diamond_B(A_{2i,2h+1}^+,A_{k,h,2h+1})
\to \Diamond_L(A_{2i+1,2h}^+,A_{k,h,2h}) 
\to \Diamond_L(A_{2i+1,2h+1}^+,A_{k,h,2h+1}) 
\to \ldots 
\to \Diamond_L(A_{2i+1,2W+1}^+,A_{k,h,2W+1})
\to \Diamond_L(B_{2i+1}^+,B_{k,h})
\to \Diamond_T(B_{2i+1}^+,B_{k,h})
$

\item[\emph{\ref{lem:3ov_fs_bb_to_cc1}}]
Define $I_{i-0.5}^+$ to be the midpoint of $I_{i}^+I_{i-1}^+$. From the definition of $I_i$, we have $||B_i^+I_{i-0.5}^+|| < d$ and $||C_{i+1}^+I_{i-0.5}^+|| < d$. The point $I_{i-0.5}^+$ lies on $CB$, however, since $\varepsilon$ can be chosen to be arbitrarily small, we can place $I_{i-0.5}^+$ on $B_{k,h}C_{k,h,1}$ while maintaining $||B_i^+I_{i-0.5}^+|| < d$ and $||C_{i+1}^+I_{i-0.5}^+|| < d$. Let $I_{i-0.5,k,h}^+$ be the copy of $I_{i-0.5}^+$ on the segment $B_{k,h}C_{k,h,1}$. Hence,
$
\Diamond_T(B_{i}^+,B_{k,h})
\to f(B_{i}^+,I_{i-0.5,k,h}^+) 
\to f(C_{i+1}^+,I_{i-0.5,k,h}^+) 
\to \Diamond_B(C_{i+1}^+,C_{k,h,1})
$.

\item[\emph{\ref{lem:3ov_fs_bb_to_cc3}}]
Since $Z_k[h] = 0$, we have $B_{k,h} = B_{k,W+1}$. Let $I_{i,k,h}$ be the copy of the point $I_i$ on the segment $B_{k,h}C_{k,h,1}$. Then by our definition of $I_i$, we have $||B_iI_i|| = d$ and $||C_{i+2}I_i|| = d$. Hence,
$
\Diamond_T(B^+_{2i},B_{k,h})
= f(I_{2i}, B_{k,h})
\to f(I_{2i}, C_{k,h,1})
= \Diamond_B(C_{2i+2}^+,C_{k,h,1})
$.

\item[\emph{\ref{lem:3ov_fs_cc_to_dd1}}] 
$
\Diamond_B(C_{i+2}^+,C_{k,h,1})
\to \Diamond_R(C_{i+2}^+,C_{k,h,1})
\to \Diamond_L(C_{i+3}^+,C_{k,h,1})
\to \Diamond_L(D_{i+3}^+,D_{k,h,1})
$.

\item[\emph{\ref{lem:3ov_fs_aa_to_dd1}}]
By \emph{\ref{lem:3ov_fs_aa_to_bb1}}, 
$\Diamond_R(A_{2i-1,1}^+,A_{k,h,1}) \to \Diamond_T(B_{2i}^+,B_{k,h})$. By \emph{\ref{lem:3ov_fs_bb_to_cc1}},
$\Diamond_T(B_{2i}^+,B_{k,h}) \to \Diamond_B(C_{2i+1}^+,C_{k,h,1})$. By \emph{\ref{lem:3ov_fs_cc_to_dd1}},
$\Diamond_B(C_{2i+1}^+,C_{k,h,1}) \to \Diamond_L(D_{2i+2}^+,D_{k,h,1})$. Putting this together,
$
\Diamond_R(A_{2i-1,1}^+,A_{k,h,1}) \to \\
\Diamond_T(B_{2i}^+,B_{k,h})
\to \Diamond_B(C_{2i+1}^+,C_{k,h,1})
\to \Diamond_L(D_{2i+2}^+,D_{k,h,1})
\to \Diamond_L(D_{2i+3}^+,D_{k,h,2})
$.

\item[\emph{\ref{lem:3ov_fs_aa_to_dd2}}]
We modify the proof of \emph{\ref{lem:3ov_fs_aa_to_dd1}} in the case where $X_i[h] \cdot Z_k[h] = 0$. If $X_i[h] = 0$, we replace \emph{\ref{lem:3ov_fs_aa_to_bb1}} with \emph{\ref{lem:3ov_fs_aa_to_bb2}} to yield 
$
\Diamond_R(A_{2i-1,1}^+,A_{k,h,1})  
\to \Diamond_T(B_{2i+1}^+,B_{k,h})
\to \Diamond_B(C_{2i+2}^+,C_{k,h,1})
\to \Diamond_L(D_{2i+3}^+,D_{k,h,1})
$. If $Z_k[h] = 0$, we replace \emph{\ref{lem:3ov_fs_bb_to_cc1}} with \emph{\ref{lem:3ov_fs_bb_to_cc3}} to yield 
$
\Diamond_R(A_{2i-1,1}^+,A_{k,h,1}) 
\to \Diamond_T(B_{2i}^+,B_{k,h})
\to \Diamond_B(C_{2i+2}^+,C_{k,h,1})
\to \Diamond_L(D_{2i+3}^+,D_{k,h,1})
$.

\item[\emph{\ref{lem:3ov_fs_aa_to_dd3}}]
$
\Diamond_L(A_{2i-1,2}^+,A_{k,h,2}) \to
\Diamond_T(B_{2i-1}^+,B_{k,h})
\to \Diamond_B(C_{2i}^+,C_{k,h,1})
\to \Diamond_L(D_{2i+1}^+,D_{k,h,1})
\to \\ \Diamond_L(D_{2i+2}^+,D_{k,h,2})
$.

\item[\emph{\ref{lem:3ov_fs_ee_to_gg1}}] 
$
\Diamond_B(E^+_{2j-1},E_{k,h}) 
\to \Diamond_L(E^+_{2j-1},E_{k,h})
\to \Diamond_L(E^+_{2j},E_{k,h})
\to \Diamond_L(F^+_{2j,1},F_{k,h,1}) 
\to \Diamond_L(F^+_{2j,2},F_{k,h,2}) \\
\to \ldots
\to \Diamond_L(F^+_{2j,2W+1},F_{k,h,2W+1}) 
\to \Diamond_T(G^+_{2j,2}, G_{k,h})
$

\item[\emph{\ref{lem:3ov_fs_ee_to_gg2}}]
Since $Y_k[h]=0$, we have that $\Diamond(F_{2j,2h}^+,F_{k,h,2h})$ is non-existent. Therefore, $\Diamond_T(F_{2j,2h-1}^+,F_{k,h,2h-1}) \\ \to \Diamond_B(F_{2j,2h-1}^+,F_{k,h,2h-1})$. Hence,
$
\Diamond_B(E^+_{2j-1},E_{k,h}) 
\to \Diamond_L(E^+_{2j-1},E_{k,h})
\to \Diamond_L(E^+_{2j},E_{k,h})
\to \Diamond_L(F^+_{2j,1},F_{k,h,1}) 
\to \Diamond_L(F^+_{2j,2},F_{k,h,2}) 
\to \ldots
\to \Diamond_L(F^+_{2j,2h-1},F_{k,h,2h-1}) 
\to \Diamond_T(F^+_{2j,2h-1},F_{k,h,2h-1}) \\
\to \Diamond_B(F_{2j,2h-1}^+,F_{k,h,2h-1})
\to \Diamond_L(F_{2j+1,2h}^+,F_{k,h,2h}) 
\to \Diamond_L(F_{2j+1,2h+1}^+,F_{k,h,2h+1})
\to \ldots
\to \Diamond_L(F_{2j+1,2W+1}^+,F_{k,h,2W+1})
\to \Diamond_T(G^+_{2j+1,2}, G_{k,h})
$

\item[\emph{\ref{lem:3ov_fs_ed_to_gg1}}]
$
\Diamond_L(E_1^+, E_{k,h})
\to \Diamond_L(F_{1,1}^+, F_{k,h,1})
\to \Diamond_L(F_{1,2}^+, F_{k,h,2}) 
\to \ldots
\to \Diamond_L(F_{1,2W+1}^+, F_{k,h,2W+1})
\to \Diamond_T(G_{1,2}^+,G_{k,h})
$

\item[\emph{\ref{lem:3ov_fs_ed_to_gg2}}]
$
f(E_{1}^+,D_{k,h,2}) 
\to \Diamond_L(D_1^+,D_{k,h,2})
\to \Diamond_L(E_{2n+1}^+,E_{k,h})
\to \Diamond_L(F_{2n+1,1}^+,F_{k,h,1}) 
\to \Diamond_L(F_{2n+1,2}^+,F_{k,h,2}) \\
\to \ldots
\to \Diamond_L(F_{2n+1,2W+1}^+,F_{k,h,2W+1})
\to \Diamond_T(G_{2n+1,2}^+,G_{k,h})$
\end{itemize}
\vspace{-20pt}
\end{proof}

\subsection{YES instances}
\label{sec:3ov_if}

In this section, we show that if our input $(\mathcal X, \mathcal Y, \mathcal Z)$ to 3OV is a YES instance, then our constructed instance $(T,m,\ell,d)$ for \problemtwo is a YES instance. 

Since $(\mathcal X, \mathcal Y, \mathcal Z)$ is a YES instance, there exists a triple of integers $1 \leq \bar i,\bar j,\bar k \leq n$, so that for all $1 \leq h \leq W$, we have $X_{\bar i}[h] \cdot Y_{\bar j}[h] \cdot Z_{\bar k}[h] = 0$. To show that $(T,m,\ell,d)$ is a YES instance, it suffices to construct a pair of vertical lines $l_s$ and $l_t$ in $F_d(T,T)$ so that, there are $m-1$ distinct monotone paths starting at $l_s$ and ending at $l_t$, where the $y$-coordinates of any two monotone paths overlap in at most one point, and where the subtrajectory from $s$ to $t$ has length at least $\ell$. Moreover, each of the monotone paths intersect the $y$-interval from $s$ to $t$ in at most one point.

Choose the starting point $s = J_{2\bar i-1,2\bar j-1}^+$ and ending point $t = L_{2\bar i+3,2\bar j+1}^+$. By Lemma~\ref{lem:sub_length_j_to_l} the subtrajectory from~$s$ to~$t$ has length exactly $\ell$. Next, we construct a set of $m-1 = 2nW+1$ monotone paths from $l_s$ to~$l_t$. 

Define $K_{j,k,h}$ to be the point on $G_{k,h} A_{k,h,1}$ so that $||K_{j,k,h}J_{i,j}^+|| = d$ and $||K_{j,k,h}G_{2n+2-j,2}^+|| = d$. Define $M_{j,k,h}$ to be the point on $D_{k,h,2} E_{k,h}$ so that $||M_{j,k,h}L_{i,j}^+|| = d$ and $||M_{j,k,h}E_{2n+2-j}^+|| = d$. Define $M_{j,k,W+1}$ to be the point on $D_{k,W+1,2}C_{k,W+1,3}$ so that $||M_{j,k,W+1}L_{i,j}^+|| = d$.

Our set of constructed paths is as follows:

\begin{itemize}[noitemsep]
    \item For $1 \leq k \leq n$, $k \neq \bar k$, $1 \leq h \leq W$, construct $f(J_{2\bar i-1,2\bar j-1}^+, A_{k,h,2}) \to f(L_{2\bar i+3,2\bar j+1}^+, E_{k,h})$.  

    \item For $1 \leq k \leq n$, $k \neq \bar k$, $1 \leq h \leq W$ construct $f(J_{2\bar i-1,2\bar j-1}^+, E_{k,h}) \to f(L_{2\bar i+3,2\bar j+1}^+, A_{k,h+1,1})$. 

    \item For $k = \bar k$, $1 \leq h \leq W$,  $Y_{\bar j}[h] = 0$, construct $f(J_{2\bar i-1,2\bar j-1}^+, K_{2\bar j-1,\bar k,h}) \to f(L_{2\bar i+3,2\bar j+1}^+, M_{2\bar j+1, \bar k, h})$.

    \item For $k = \bar k$, $1 \leq h \leq W$,  $Y_{\bar j}[h] = 0$, construct $f(J_{2\bar i-1,2\bar j-1}^+, M_{2\bar j+1, \bar k, h}) \to f(L_{2\bar i+3,2\bar j+1}^+, K_{2\bar j-1,\bar k,h+1})$.

    \item For $k = \bar k$, $1 \leq h \leq W$,  $Y_{\bar j}[h] \neq 0$, construct $f(J_{2\bar i-1,2\bar j-1}^+, K_{2\bar j-1,\bar k,h}) \to f(L_{2\bar i+3,2\bar j+1}^+, C_{k,h,2})$.

    \item For $k = \bar k$, $1 \leq h \leq W$,  $Y_{\bar j}[h] \neq 0$, construct $f(J_{2\bar i-1,2\bar j-1}^+, D_{k,h,2}) \to f(L_{2\bar i+3,2\bar j+1}^+, K_{2\bar j-1,\bar k,h+1})$.

    \item Construct $f(J_{2\bar i-1,2\bar j-1}^+, K_{2\bar j-1,\bar k,W+1}) \to f(L_{2\bar i+3,2\bar j+1}^+, M_{2\bar j+1, \bar k, W+1})$
\end{itemize}

We prove the correctness of our constructed paths. First, we note that $f(J_{2\bar i-1,2\bar j-1}^+, K_{2\bar j-1,\bar k,h})$ and $f(L_{2\bar i+3,2\bar j+1}^+, M_{2\bar j+1, \bar k, h})$ are free points since $||K_{2 \bar j - 1,k,h}J_{i,2 \bar j - 1}^+|| = d$ and $||M_{2 \bar j + 1,k,h}L_{i,2 \bar j + 1}^+|| = d$, respectively. In total, we have $2nW + 1$ paths. Next, we prove that these $2nW + 1$ paths are indeed valid.

\begin{lemma}
$f(J_{2\bar i-1,2\bar j-1}^+, A_{k,h,2}) \to f(L_{2\bar i+3,2\bar j+1}^+, E_{k,h})$
\end{lemma}

\begin{proof}
$f(J_{2\bar i-1,2\bar j-1}^+, A_{k,h,2}) 
\to \Diamond_L(A_{2\bar i-1,2}^+,A_{k,h,2})
\to \Diamond_T(D_{2\bar i+2}^+,D_{k,h,2})
\to f(L_{2\bar i+3,2\bar j+1}^+, E_{k,h})$, where the second $\to$ in the chain is given by Lemma~\ref{lem:3ov_fs}\emph{\ref{lem:3ov_fs_aa_to_dd3}}.
\end{proof}

\begin{lemma}
$f(J_{2\bar i-1,2\bar j-1}^+, E_{k,h}) \to f(L_{2\bar i+3,2\bar j+1}^+, A_{k,h+1,1})$
\end{lemma}

\begin{proof}
$f(J_{2\bar i-1,2\bar j-1}^+, E_{k,h})
\to \Diamond_L(E_1^+, E_{k,h}) 
\to \Diamond_T(G_{1,2}^+,G_{k,h})
\to f(L_{2\bar i+3,2\bar j+1}^+, A_{k,h+1,1})$, where the second $\to$ in the chain is given by Lemma~\ref{lem:3ov_fs}\emph{\ref{lem:3ov_fs_ed_to_gg1}}.
\end{proof}

\begin{lemma}
\label{lem:3ov_yes_jk_lm}
$f(J_{2\bar i-1,2\bar j-1}^+, K_{2\bar j-1,\bar k,h}) \to f(L_{2\bar i+3,2\bar j+1}^+, M_{2\bar j+1, \bar k, h})$
\end{lemma}

\begin{proof}
$f(J_{2\bar i-1,2\bar j-1}^+, K_{2\bar j-1,\bar k,h}) 
\to \Diamond_R(A_{2\bar i-1,1}^+,A_{k,h,1})
\to \Diamond_L(D_{2\bar i+3}^+,D_{k,h,2})
\to f(L_{2\bar i+3,2\bar j+1}^+, M_{2\bar j+1, \bar k, h})$, where the second $\to$ in the chain is given by Lemma~\ref{lem:3ov_fs}\emph{\ref{lem:3ov_fs_aa_to_dd1}}.
\end{proof}

\begin{lemma}
$f(J_{2\bar i-1,2\bar j-1}^+, M_{2\bar j+1, \bar k, h}) \to f(L_{2\bar i+3,2\bar j+1}^+, K_{2\bar j-1,\bar k,h+1})$ if $Y_{\bar j}[h]=0$.
\end{lemma}

\begin{proof}
$f(J_{2\bar i-1,2\bar j-1}^+, M_{2\bar j+1, \bar k, h}) 
\to \Diamond_B(E^+_{2n+1-2\bar j}, E_{k,h})
\to \Diamond_T(G^+_{2n+3-2\bar j,2}, G_{k,h}) 
\to \\ f(L_{2\bar i+3,2\bar j+1}^+, K_{2\bar j-1,\bar k,h+1})$, where the first $\to$ is given by $||M_{2\bar j+1,k,h}E_{2n+1-2 \bar j}^+|| = d$, second $\to$ in the chain is given by Lemma~\ref{lem:3ov_fs}\emph{\ref{lem:3ov_fs_ee_to_gg2}}, and the third $\to$ is given by $||K_{2\bar j -1,k,h}G_{2n+3-2 \bar j,2}^+|| = d$.
\end{proof}

\begin{lemma}
$f(J_{2\bar i-1,2\bar j-1}^+, K_{2\bar j-1,\bar k,h}) \to f(L_{2\bar i+3,2\bar j+1}^+, C_{k,h,2})$ if $Y_{\bar j}[h] \neq 0$.
\end{lemma}

\begin{proof}
$f(J_{2\bar i-1,2\bar j-1}^+, K_{2\bar j-1,\bar k,h}) 
\to \Diamond_R(A_{2i-1,1}^+,A_{k,h,1}) 
\to \Diamond_L(D_{2i+3}^+,D_{k,h,1})
\to f(L_{2\bar i+3,2\bar j+1}^+, C_{k,h,2})$, 
where the second $\to$ in the chain is given by Lemma~\ref{lem:3ov_fs}\emph{\ref{lem:3ov_fs_aa_to_dd2}} (note $Y_{\bar j}[h] \neq 0 \implies X_i[h] \cdot Z_k[h] = 0$).
\end{proof}

\begin{lemma}
$f(J_{2\bar i-1,2\bar j-1}^+, D_{k,h,2}) \to f(L_{2\bar i+3,2\bar j+1}^+, K_{2\bar j-1,\bar k,h+1})$
\end{lemma}

\begin{proof}
$f(J_{2\bar i-1,2\bar j-1}^+, D_{k,h,2}) 
\to f(E_{1}^+,D_{k,h,2}) 
\to \Diamond_T(G_{2n+1,2}^+,G_{k,h})
\to f(L_{2\bar i+3,2\bar j+1}^+, K_{2\bar j-1,\bar k,h+1})$, where the second $\to$ in the chain is given by Lemma~\ref{lem:3ov_fs}\emph{\ref{lem:3ov_fs_ed_to_gg2}}.
\end{proof}

\begin{lemma}
$f(J_{2\bar i-1,2\bar j-1}^+, K_{2\bar j-1,\bar k,W+1}) \to f(L_{2\bar i+3,2\bar j+1}^+, M_{2\bar j+1, \bar k, W+1})$
\end{lemma}

\begin{proof}
The proof is essentially the same as the one for Lemma~\ref{lem:3ov_yes_jk_lm}, except $h$ is set to $W+1$. One subtlety is that~$M_{j,k,W+1}$ is a point on~$D_{k,W+1,2}C_{k,W+1,3}$. But the same monotone path persists since $\Diamond_L(D_{2\bar i+3}^+,D_{k,W+1,2})$ is the same shaped diamond as $\Diamond_L(D_{2\bar i+3}^+,D_{k,h,2})$ for $1 \leq h \leq W$. 
\end{proof}

Hence, we have $m-1$ valid monotone paths that start at $l_s$ and end at $l_t$. No pair of monotone paths overlap in $y$-coordinate in more than one point. No monotone path overlaps in $y$-coordinate with the $y$-interval from $s = J_{2\bar i-1,2\bar j-1}^+$ to $t = L_{2\bar i+3,2\bar j+1}^+$. These statements are true in both the $Y_{\bar j}[h] = 0$ and $Y_{\bar j}[h] \neq 0$ cases. Hence, we have shown that our constructed instance $(T,m,\ell,d)$ is a YES instance, yielding the following theorem.

\begin{theorem}
\label{theorem:3ov_yes_to_yes}
If our input $(\mathcal X, \mathcal Y, \mathcal Z)$ is a YES instance for 3OV, then our constructed instance $(T,m,\ell,d)$ is a YES instance for \problemtwo. 
\end{theorem}

\subsection{Cuts in free space}
\label{sec:3ov_cuts}

In the following section (Section~\ref{sec:3ov_only_if}), we will show that if $(\mathcal X, \mathcal Y, \mathcal Z)$ is a NO instance for 3OV, then $(T,m,\ell,d)$ is a NO instance for \problemtwo. This involves showing that for any pair of vertical lines~$l_s$ and~$l_t$, there cannot be $m-1$ monotone paths from $l_s$ to $l_t$ with the property that the $y$-coordinate of the $m-1$ monotone paths overlap in at most one point, and the $y$-coordinates of the $m-1$ monotone paths intersect with the $y$-interval from $s$ to $t$ in at most one point. In actuality, we show the contrapositive of the above statement in Section~\ref{sec:3ov_only_if}, but the core idea of upper bounding the number of monotone paths remains the same. 

To upper bound the number of monotone paths that can exist in our free space diagram, we require a significantly different (or in fact, opposite) strategy to Sections~\ref{sec:3ov_paths}. In particular, we need a way to show that, for a point $(s,z_1)$ on $l_s$ and a point $(t,z_2)$ on $l_t$, that there is no monotone path from $(s,z_1) \to (t,z_2)$. To show that there is no such monotone path, we construct a cutting sequence from $(s,z_1)$ to $(t,z_2)$. This is very similar in spirit to cuts constructed in Buchin~\etal's~\cite{DBLP:conf/soda/BuchinOS19} lower bound. 

First, let us define a cutting sequence. We say that $f(x,y)$ dominates $f(w,z)$ if $f(x,y)$ is strictly above and to the left of $f(w,z)$. In other words, $x > w$ and $y > z$. We say that $f(x,y)$ undermines $f(w,z)$ if $f(x,y)$ is vertically below $f(w,z)$, and the vertical segment between the two points is entirely in non-free space. We define a dominating sequence to be a sequence of points so that each point either dominates or undermines the next point in the sequence. We define a cutting sequence to be a dominating sequence where the first point dominates the second point, and the second to last point dominates the last point. 

To simplify the description of our dominating and cutting sequences in this and subsequent sections, we use the following notation. Let $f(x,y) \dominates f(w,z)$ denote that $f(x,y)$ dominates $f(w,z)$. Let $f(x,y) \undermines f(w,z)$ denote that $f(x,y)$ undermines $f(w,z)$. Let $f(x,y) \domseq f(w,z)$ denote that there is a dominating sequence from $f(x,y)$ to $f(w,z)$. See Figure~\ref{fig:dominating_sequence}.

\begin{figure}[ht]
    \centering
    \includegraphics{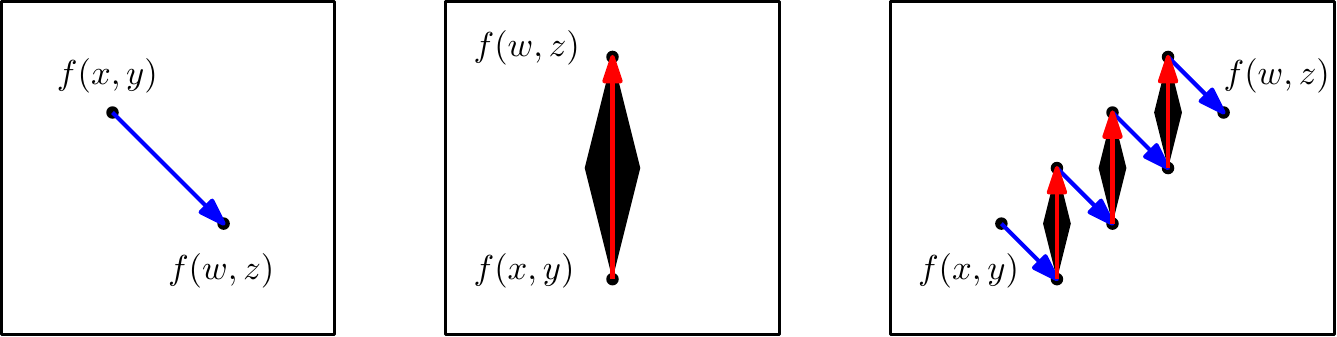}
    \caption{Left: $f(x,y) \dominates f(w,z)$. Middle: $f(x,y) \undermines f(w,z)$. Right: $f(x,y) \domseq f(w,z)$.}
    \label{fig:dominating_sequence}
\end{figure}

With this notation in mind, we are ready to prove that if there is a cutting sequence from $f(x,y)$ to $f(w,z)$ that there cannot be a monotone path $f(x,y) \to f(w,z)$.

\begin{lemma}
\label{lem:3ov_cutting_sequence}
If there is a cutting sequence from $f(x,y)$ to $f(w,z)$, then there is no monotone path from $f(x,y)$ to $f(w,z)$. 
\end{lemma}

\begin{proof}
Let the cutting sequence be $f(x_1, y_1), f(x_2, y_2), \ldots, f(x_{i-1}, y_{i-1}), f(x_i, y_i)$, where $f(x_1, y_1) = f(x,y)$ and $f(x_i, y_i) = f(w,z)$. Then $f(x_1,y_1) \dominates f(x_2,y_2) \domseq f(x_{i-1},y_{i-1}) \dominates f(x_i,y_i)$. Suppose for the sake of contradiction that there is a monotone path from $f(x_1, y_1)$ to $f(x_i, y_i)$. See Figure~\ref{fig:cutting_sequence}.

\begin{figure}[ht]
    \centering
    \includegraphics{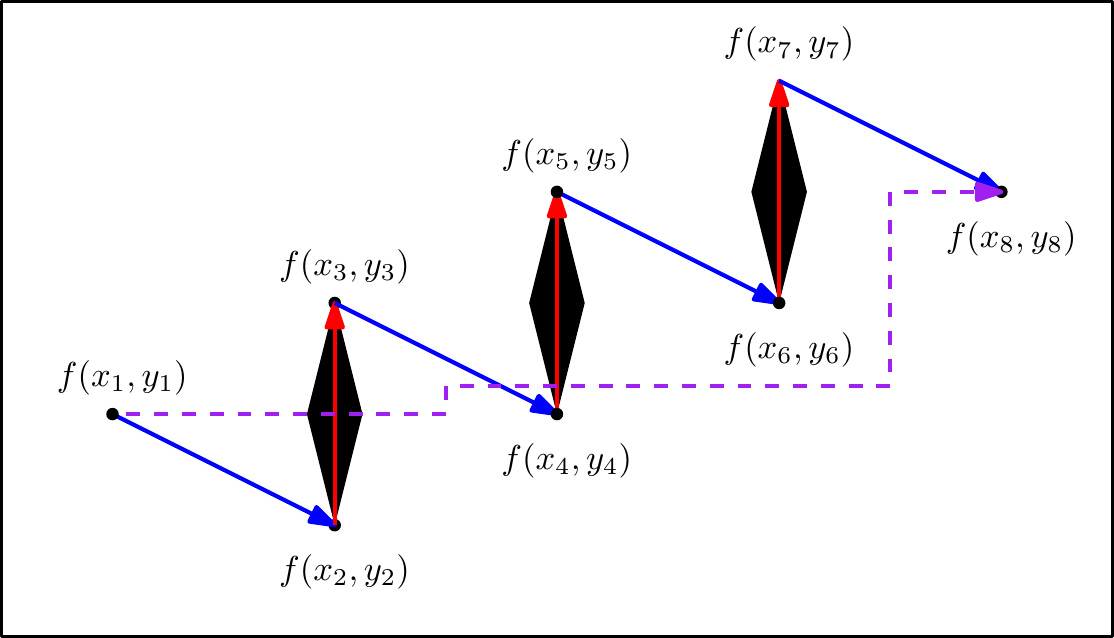}
    \caption{A cutting sequence from $f(x_1,y_1)$ to $f(x_i,y_i)$, for $i=8$.}
    \label{fig:cutting_sequence}
\end{figure}

Since $f(x_1,y_1) \dominates f(x_2,y_2)$, we know $f(x_2,y_2)$ is strictly below the monotone path $f(x_1, y_1) \to f(x_i, y_i)$. Since $f(x_{i-1}, y_{i-1}) \dominates (x_i, y_i)$, we know $f(x_{i-1}, y_{i-1})$ is strictly above the monotone path $f(x_1, y_1) \to f(x_i, y_i)$. Therefore, a dominating sequence $f(x_2,y_2) \domseq f(x_{i-1},y_{i-1})$ starts below the monotone path and ends above the monotone path. 

If we ignore the $y$-coordinates and focus only on the $x$-coordinates, the dominating sequence $f(x_2,y_2) \domseq f(x_{i-1}, y_{i-1})$ and the monotone path $f(x_1,y_1) \to f(x_i,y_i)$ are both weakly increasing. Therefore, there must be an $x$-coordinate where the dominating sequence $f(x_2,y_2) \domseq f(x_{i-1}, y_{i-1})$ crosses from being below to being above the monotone path $f(x_1,y_1) \to f(x_i,y_i)$. At this point, the dominating sequence is increasing in $y$-coordinate. Therefore, we must have an undermining step at this $x$-coordinate. However, an undermining step only traverses non-free space, so it cannot cross the monotone path. This yields a contradiction.
\end{proof}

In this section, we compile a list of dominating sequences, and pairs of undermining points. We prove the correctness of these sequences. For now, we will not show any cutting sequences, but in Section~\ref{sec:3ov_only_if}, we will combine our dominating sequences and undermining points to form cutting sequences. Some of these dominating sequences are conditioned on one of the booleans $X_i[h]$, $Y_j[h]$ or $Z_k[h]$ being equal to one. 

For several of our dominating sequences and undermining points, we require the following additional points. Let $N_1$, $N_2$ and $N_3$ be points on $H_1H_2$ so that the points $H_1, N_1, N_2, N_3, H_2$ are evenly spaced. Let $N_{k,1}$, $N_{k,2}$ and $N_{k,3}$ be the copies of $N_1$, $N_2$ and $N_3$ on the segment $H_{k,1}H_{k,2}$. 

\begin{lemma}
\label{lem:3ov_cuts}
For $1 \leq i,j,k,u \leq n$ and $1 \leq h \leq W$:
\begin{enumerate}[label=(\alph*)]
    \item \label{lem:3ov_cuts_gg_cc_1}
$\Diamond_B(G_{1,1}^+, G_{k,h}) \domseq \Diamond_B(C_1^+,C_{k,h,1})$
	\item \label{lem:3ov_cuts_cc_gg_1}
$\Diamond_B(C_1^+,C_{k,h,1}) \domseq \Diamond_T(G_{2n+1,2}^+,G_{k,h})$
	\item \label{lem:3ov_cuts_aa_cc_1}
$\Diamond_B(A_{2i+1,1}^+,A_{k,h,1}) \domseq \Diamond_T(C_{2i+3}^+,C_{k,h,1})$
	\item \label{lem:3ov_cuts_ee_gg_1}
$\Diamond_B(E_1^+,E_{k,h}) \domseq \Diamond_T(G_{1,2}^+,G_{k,h})$
	\item \label{lem:3ov_cuts_aa_dd_1}
$\Diamond_B(A_{2i-1,2}^+,A_{k,h,2}) \domseq \Diamond_T(D_{2i+2}^+,D_{k,h,2})$.
	\item \label{lem:3ov_cuts_aa_cc_2}
$\Diamond_B(A_{2i,1}^+,A_{k,h,1}) \domseq \Diamond_T(C_{2i+2}^+,C_{k,h,2})$ if $X_i[h] \cdot Z_k[h] = 1$.
	\item \label{lem:3ov_cuts_ee_gg_2}
$\Diamond_B(E_{2j}^+,E_{k,h}) \domseq \Diamond_T(G_{2j,2}^+,G_{k,h})$ if $Y_j[h] = 1$.
	\item \label{lem:3ov_cuts_ng_nn}
$f(N_{k,3},G_{1,1}) \undermines f(N_{k,3}, N_{1,2})$
	\item \label{lem:3ov_cuts_hn_nn}
$f(H_{k,1},N_{u,1}) \domseq f(N_{k,3}, N_{u+1,2})$
	\item \label{lem:3ov_cuts_hn_hd}
$f(H_{k,1},N_{n,1}) \undermines f(H_{k,1}, D_{2n+3}^+)$
	\item \label{lem:3ov_cuts_gh_gg}
$f(G_{1,1}^+,H_{k-1,1}) \undermines f(G_{1,1}^+,G_{k,h})$
	\item \label{lem:3ov_cuts_hh_cg}
$f(H_1^-,H_{n,1}) \domseq f(C_3^+,G^+_{2n+1,2})$
	\item \label{lem:3ov_cuts_he_hd}
$f(H_1^-,E_{2n+1}^+) \undermines f(H_1^-,D_{2n+3}^+)$
\end{enumerate}
\end{lemma}

\begin{proof}
\quad
\begin{itemize}
    \item[\emph{\ref{lem:3ov_cuts_gg_cc_1}}]
$\Diamond_B(G_{1,1}^+, G_{k,h}) 
\undermines \Diamond_T(G_{1,1}^+, G_{k,h})
\dominates \Diamond_B(A_{1,1}^+, A_{k,h,1})
\undermines \Diamond_T(A_{1,1}^+, A_{k,h,1})
\dominates \Diamond_B(A_{1,2}^+, A_{k,h,2}) \\
\undermines \Diamond_T(A_{1,2}^+, A_{k,h,2})
\dominates
\ldots
\dominates \Diamond_B(A_{1,2W+1}^+, A_{k,h,2W+1}) 
\undermines \Diamond_T(A_{1,2W+1}^+, A_{k,h,2W+1}) \\
\dominates \Diamond_B(C_1^+,C_{k,h,1})
$
    \item[\emph{\ref{lem:3ov_cuts_cc_gg_1}}]
$
\Diamond_B(C_1^+,C_{k,h,1})
\undermines \Diamond_T(C_1^+,C_{k,h,1})
\dominates \Diamond_B(D_1^+,D_{k,h,1})
\undermines \Diamond_T(D_1^+,D_{k,h,1}) 
\dominates \Diamond_B(C_2^+,C_{k,h,2})\\
\undermines \Diamond_T(C_2^+,C_{k,h,2}) 
\dominates \Diamond_B(D_2^+,D_{k,h,2})
\undermines \Diamond_T(D_2^+,D_{k,h,2}) 
\dominates \Diamond_B(E_{2n+1}^+,E_{k,h}) 
\undermines  \\\Diamond_T(E_{2n+1}^+,E_{k,h})
\dominates \Diamond_B(F_{2n+1,1}^+,F_{k,h,1}) 
\undermines \Diamond_T(F_{2n+1,1}^+,F_{k,h,1})
\dominates \Diamond_B(F_{2n+1,2}^+,F_{k,h,2}) 
\undermines \Diamond_T(F_{2n+1,2}^+,F_{k,h,2})
\dominates \ldots
\dominates \Diamond_B(F_{2n+1,2W+1}^+,F_{k,h,2W+1}) 
\undermines \Diamond_T(F_{2n+1,2W+1}^+,F_{k,h,2W+1})
\dominates \Diamond_B(G_{2n+1,2}^+,G_{k,h+1})
\undermines \Diamond_T(G_{2n+1,2}^+,G_{k,h+1})
$
    \item[\emph{\ref{lem:3ov_cuts_aa_cc_1}}]
$
\Diamond_B(A_{2i+1,1}^+,A_{k,h,1})
\undermines \Diamond_T(A_{2i+1,1}^+,A_{k,h,1})
\dominates \Diamond_B(A_{2i+1,2}^+,A_{k,h,2})
\undermines \Diamond_T(A_{2i+1,2}^+,A_{k,h,2}) \\
\dominates \ldots 
\dominates \Diamond_B(A_{2i+1,2W+1}^+,A_{k,h,2W+1})
\undermines \Diamond_T(A_{2i+1,2W+1}^+,A_{k,h,2W+1})
\dominates \Diamond_B(B_{2i+1}^+,B_{k,h}) \\
\undermines \Diamond_T(B_{2i+1}^+,B_{k,h}) 
\dominates \Diamond_B(C_{2i+3}^+,C_{k,h,1})
\undermines \Diamond_T(C_{2i+3}^+,C_{k,h,1})
$
    \item[\emph{\ref{lem:3ov_cuts_ee_gg_1}}]
$
\Diamond_B(E_1^+,E_{k,h})
\undermines \Diamond_T(E_1^+,E_{k,h})
\dominates \Diamond_B(E_1^+,E_{k,h})
\undermines \Diamond_T(E_1^+,E_{k,h})
\dominates \Diamond_B(F_{1,1}^+,F_{k,h,1}) \\
\undermines \Diamond_T(F_{1,1}^+,F_{k,h,1})
\dominates \Diamond_B(F_{1,2}^+,F_{k,h,2})
\undermines \Diamond_T(F_{1,2}^+,F_{k,h,2})
\dominates \ldots 
\dominates \Diamond_B(F_{1,2W+1}^+,F_{k,h,2W+1}) \\
\undermines \Diamond_T(F_{1,2W+1}^+,F_{k,h,2W+1})
\dominates \Diamond_B(G_{1,2}^+,G_{k,h})
\undermines \Diamond_T(G_{1,2}^+,G_{k,h})
$
    \item[\emph{\ref{lem:3ov_cuts_aa_dd_1}}]
$
\Diamond_B(A_{2i-1,2}^+,A_{k,h,2})
\undermines \Diamond_T(A_{2i-1,2}^+,A_{k,h,2}) 
\dominates \Diamond_B(A_{2i-1,3}^+,A_{k,h,3})
\undermines \Diamond_T(A_{2i-1,3}^+,A_{k,h,3}) \\
\dominates \ldots 
\dominates \Diamond_B(A_{2i-1,2W+1}^+,A_{k,h,2W+1})
\undermines \Diamond_T(A_{2i-1,2W+1}^+,A_{k,h,2W+1}) 
\dominates \Diamond_B(B_{2i-1}^+,B_{k,h}) \\
\undermines \Diamond_T(B_{2i-1}^+,B_{k,h}) 
\dominates \Diamond_B(C_{2i+1}^+,C_{k,h,1})
\undermines \Diamond_T(C_{2i+1}^+,C_{k,h,1})
\dominates \Diamond_B(D_{2i+1}^+,D_{k,h,1}) \\
\undermines \Diamond_T(D_{2i+1}^+,D_{k,h,1}) 
\dominates \Diamond_B(C_{2i+2}^+,C_{k,h,2}) 
\undermines \Diamond_T(C_{2i+2}^+,C_{k,h,2})
\dominates \Diamond_B(D_{2i+2}^+,D_{k,h,2}) \\
\undermines \Diamond_T(D_{2i+2}^+,D_{k,h,2})
$
    \item[\emph{\ref{lem:3ov_cuts_aa_cc_2}}]
Since $X_i[h] = 1$, we have $\Diamond(A_{2i,2h}^+,A_{k,h,2h})$ is non-empty. Therefore, $
\Diamond_T(A_{2i,2h-1}^+,A_{k,h,2h})
\dominates \Diamond_B(A_{2i,2h}^+,A_{k,h,2h})
\undermines \Diamond_T(A_{2i,2h}^+,A_{k,h,2h})
\dominates \Diamond_B(A_{2i,2h+1}^+,A_{k,h,2h})
$. Since $Z_k[h] = 1$, $\\\Diamond_T(B_{2i}^+,B_{k,h}) 
\dominates \Diamond_B(C_{2i+1}^+,C_{k,h,1})$. Now, $
\Diamond_B(A_{2i,1}^+,A_{k,h,1})
\undermines \Diamond_T(A_{2i,1}^+,A_{k,h,1})
\dominates \\ \Diamond_B(A_{2i,2}^+,A_{k,h,2}) 
\undermines \Diamond_T(A_{2i,2}^+,A_{k,h,2}) 
\dominates \ldots 
\dominates \Diamond_B(A_{2i,2h}^+,A_{k,h,2h})
\undermines \Diamond_T(A_{2i,2h}^+,A_{k,h,2h})
\dominates \ldots 
\dominates \Diamond_B(A_{2i,2W+1}^+,A_{k,h,2W+1}) 
\undermines \Diamond_T(A_{2i,2W+1}^+,A_{k,h,2W+1}) 
\dominates \Diamond_B(B_{2i}^+,B_{k,h}) 
\undermines \Diamond_T(B_{2i}^+,B_{k,h}) \\
\dominates \Diamond_B(C_{2i+1}^+,C_{k,h,1}) 
\undermines \Diamond_T(C_{2i+1}^+,C_{k,h,1}) 
\dominates \Diamond_B(D_{2i+1}^+,D_{k,h,1}) 
\undermines \Diamond_T(D_{2i+1}^+,D_{k,h,1}) \\
\dominates \Diamond_B(C_{2i+2}^+,C_{k,h,2}) 
\undermines \Diamond_T(C_{2i+2}^+,C_{k,h,2}) 
$
    \item[\emph{\ref{lem:3ov_cuts_ee_gg_2}}]
Since $Y_j[h] = 0$, we have $\Diamond(F_{2j,2h}^+,F_{k,h,2h})$ is non-empty. Therefore, $
\Diamond_T(F_{2j,2h-1}^+,F_{k,h,2h-1})
\dominates \Diamond_B(F_{2j,2h}^+,F_{k,h,2h})
\undermines \Diamond_T(F_{2j,2h}^+,F_{k,h,2h})
\dominates \Diamond_B(F_{2j,2h+1}^+,F_{k,h,2h+1})
$. Now, $
\Diamond_B(E_{2j}^+,E_{k,h}) \\
\undermines \Diamond_T(E_{2j}^+,E_{k,h}) 
\dominates \Diamond_B(F_{2j,1}^+,F_{k,h,1})
\undermines \Diamond_T(F_{2j,1}^+,F_{k,h,1})
\dominates \Diamond_B(F_{2j,2}^+,F_{k,h,2})
\undermines \\ \Diamond_T(F_{2j,2}^+,F_{k,h,2}) 
\dominates \ldots 
\dominates \Diamond_B(F_{2j,2h}^+,F_{k,h,2h}) 
\undermines \Diamond_T(F_{2j,2h}^+,F_{k,h,2h})
\dominates \ldots
\dominates \\ \Diamond_B(F_{2j,2n+1}^+,F_{k,h,2n+1}) 
\undermines \Diamond_T(F_{2j,2n+1}^+,F_{k,h,2n+1}) 
\dominates \Diamond_B(G_{2j,2}^+,G_{k,h}) 
\undermines \Diamond_T(G_{2j,2}^+,G_{k,h})
$
    \item[\emph{\ref{lem:3ov_cuts_ng_nn}}]
Recall that $H_1 = 4r' \cis (3W+4) \cdot \phi$ and $H_2 = 8r' \cis (2W+3) \cdot \phi$. So the distance between consecutive pairs in $H_1, N_1, N_2, N_3, H_2$ is greater than $d$. The subtrajectory from $G_{1,1}$ to $N_{1,2}$ is closest to the point $N_3$ at its endpoint, but $||N_{1,2}N_3|| > d$. Hence, for all points $x$ on the subtrajectory between $G_{1,1}$ to $N_{1,2}$, we have $||xN_{k,3}|| > r' > d$. Hence, $f(N_{k,3},G_{1,1}) \undermines f(N_{k,3}, N_{1,2})$.

    \item[\emph{\ref{lem:3ov_cuts_hn_nn}}]
Let $M$ be a point so that $H_{u,2} \prec M$ and $||H_{u,2}M|| = 2d$. We show that $f(H_{k,1},N_{u,1}) \undermines f(H_{k,1},M) \dominates f(N_{k,3},H_{u,2}) \undermines f(N_{k,3},N_{u+1,2})$. The subtrajectory from $N_{u,1}$ to $M$ lies outside the circle of radius $2 r'$ centered at the origin, but $H_{k,1}$ is within distance $r'$ of the origin. Hence, $f(H_{k,1},N_{u,1}) \undermines f(H_{k,1},M)$. Similarly, $f(N_{k,3},H_{u,2}) \undermines f(N_{k,3},N_{u+1,2})$.

    \item[\emph{\ref{lem:3ov_cuts_hn_hd}}]
The distances between $H_{k,1}$ and the segments $N_{n,1}H_{2,n}$ and $H_{2,n}D^-_0$ is strictly greater than~$d$. Of all the points on the subtrajectory from $D^-_0$ to $D_{2n+3}^+$, the point closest to $H_{k,1}$ is $H_1^-$, as all other vertices lie below the real axis on the complex plane. Since $||H_1^-H_{k,1}|| > d$, we have, for every point $x$ on segment $N_{n,1}H_{2,n}$, that $f(H_{k,1},x)$ is in non-free space. Therefore, $f(H_{k,1},N_{n,1}) \undermines f(H_{k,1}, D_{2n+3}^+)$. 

    \item[\emph{\ref{lem:3ov_cuts_gh_gg}}]
Let $x$ be a point on the segment $G_{k,h} H_{k,1}$. Note that as $x$ moves from $G_{k,h}$ to $H_{k,1}$, the distance from $x$ to $G_{1,1,}^+$ increases. But $||G_{k,h}G_{1,1}^+|| > d$, so $||xG_{1,1}^+|| > d$. Therefore, $f(G_{1,1}^+,x)$ is in non-free space for all $x$ on the segment $G_{k,h} H_{k,1}$. Thus, $f(G_{1,1}^+,H_{k-1,1}) \undermines f(G_{1,1}^+,G_{k,h})$.

    \item[\emph{\ref{lem:3ov_cuts_hh_cg}}]
We show that $f(H_1^-,H_{n,1}) \undermines f(H_1^-,D_2^+) \dominates f(C_3^+,D_2^-) \undermines f(C_3^+,G^+_{2n+1,2})$. The point $H_1^-$ is distance $d$ away from the segments $H_{n,1}H_{n,2}$ and $H_{n,2}D_0^-$. The subtrajectory from $D_0^-$ to $D_2^+$ lies entirely below the real axis in the complex plane, so the distance from $H_1^-$ to any of these points is greater than $d$. Therefore, $f(H_1^-,H_{n,1}) \undermines f(H_1^-,D_2^+)$. The distance from $C_3^+$ to any point on the subtrajectory starting at $D_2^-$ and ending at $G_{2n+1,2}^+$ is greater than $d$. Therefore, $f(C_3^+,D_2^-) \undermines f(C_3^+,G^+_{2n+1,2})$. Putting this together yields the dominating sequence.

    \item[\emph{\ref{lem:3ov_cuts_he_hd}}]
The subtrajectory from $E_{2n+1}^+$ to $D_{2n+3}^+$ lies entirely below the real axis in the complex plane, so the distance from $H_1^-$ to any of these points is at least $d$. Therefore, $f(H_1^-,x)$ is in non-free space for all points $x$ on the segment $E_{2n+1}^+D_{2n+3}^+$. This implies the claimed lemma.
\end{itemize}
\vspace{-20pt}
\end{proof}

\subsection{NO instances}
\label{sec:3ov_only_if}

In this section, we show that if our input $(\mathcal X, \mathcal Y, \mathcal Z)$ to 3OV is a NO instance, then our constructed instance $(T,m,\ell,d)$ for \problemtwo is a NO instance. 

We do this by taking the contrapositive. We show that if our constructed instance $(T,m,\ell,d)$ is a YES instance for \problemtwo, then our input $(\mathcal X, \mathcal Y, \mathcal Z)$ is a YES instance for 3OV.

Since $(T,m,\ell,d)$ is a YES instance, there exists a pair of vertical lines $l_s$ and $l_t$ so that the subtrajectory from $s$ to $t$ has length at least $\ell$, there are $m-1 = 2nW+1$ monotone paths from $l_s$ to $l_t$, no two of these monotone paths overlap in $y$-coordinate, and no monotone path overlaps in $y$-coordinate with the $y$-interval from $s$ to $t$.

The remainder of this section consists of five steps. The first step is to narrow down the starting position $s$ of the reference subtrajectory. The second step is to show that there are $2W+1$ monotone paths from $l_s$ to $l_t$ between the $y$-coordinates of $H_{k-1,1}$ and $H_{k,1}$, for some $1 \leq k \leq n$. The third step is to show that $s$ is between $G^+_{2i-1,1}$ and $G^+_{2i+1,1}$, for some $1 \leq i \leq n$. The fourth step is to show that $s$ is between $J^+_{2i-1,2j-1}$ and $J_{2i-1,2j+1}$ for some $1 \leq j \leq n$. The fifth step is to show that, for the integers $i$, $j$ and $k$ in our second, third and fourth steps, that $X_i$, $Y_j$, $Z_k$ are orthogonal. Putting our five steps together shows that our input $(\mathcal X, \mathcal Y, \mathcal Z)$ is a YES instance, as required.

As mentioned above, our first step is to narrow down the position of the starting point $s$. We begin with some useful notation, and then prove that $s$ is between $H_{n,1}$ and $G_{2n+1,1}^+$ in Lemma~\ref{lem:prec_h_s_g}.

\begin{definition}
We say $x \prec y$ if $x$ precedes $y$ in the trajectory $T$.
\end{definition}

\begin{lemma}
\label{lem:prec_h_s_g}
$H_{n,1} \prec s \prec G_{2n+1,1}^+$
\end{lemma}

\begin{proof}
We first show $s \prec G_{2n+1,1}^+$. By our construction, the subtrajectory from $A_{2n-1,1}^+$ to $D_{2n+3}^+$ has length $\ell + \delta$. Hence, the subtrajectory from $G_{2n+1,1}^+$ to $D_{2n+3}^+$ has length less than $\ell$. Since $D_{2n+3}^+$ is the final vertex of $T$, we have $t \preceq D_{2n+3}^+$. Since the subtrajectory from $s$ to $t$ has length at least $\ell$, we have that $s \prec G_{2n+1,1}^+$.

Next, we show $H_{n,1} \prec s$. Suppose for the sake of contradiction that $s \preceq H_{n,1}$. Without loss of generality, let $H_{k-1,1} \preceq s \preceq H_{k,1}$ for some $1 \leq k \leq n$. By Lemma~\ref{lem:sub_length_h_h}, we must have $s \preceq H_k \prec H_{k,2} \prec t$. Let the first $n+1$ monotone paths from $l_s$ to $l_t$ be $f(s,y_i) \to f(t,z_i)$ for $1 \leq i \leq n+1$. 

We will prove by induction that $N_{i,2} \preceq z_i$ for all $1 \leq i \leq n$. In the base case, suppose for the sake of contradiction that $z_1 \prec N_{1,2}$. Then $f(s,y_1) \dominates f(N_{k,3},G_{1,1}) \undermines f(N_{k,3},N_{1,2}) \dominates f(t,z_1)$ by Lemma~\ref{lem:3ov_cuts}\emph{\ref{lem:3ov_cuts_ng_nn}}. There is a cutting sequence from $f(s,y_1)$ to $f(t,z_1)$, which yields a contradiction.

Now we prove the inductive case. Assume the inductive hypothesis that $N_{i,2} \preceq z_i$. Then $N_{i,2} \preceq y_{i+1}$, since our monotone paths do not overlap in $y$-coordinate. Suppose for the sake of contradiction that $z_{i+1} \prec N_{i+1,2}$. Then $f(s,y_{i+1}) \dominates f(H_{k,1}, N_{i,1}) \domseq f(N_{k,3},N_{i+1},2) \dominates f(t,z_{i+1})$ by Lemma~\ref{lem:3ov_cuts}\emph{\ref{lem:3ov_cuts_hn_nn}}. There is a cutting sequence from $f(s,y_{i+1})$ to $f(t,z_{i+1})$, which is a contradiction. This completes our induction.

Setting $i=n$ in our induction, we get that $N_{n,1} \preceq z_n$. Therefore, $N_{n,1} \preceq y_{n+1}$. Now, by Lemma~\ref{lem:3ov_cuts}\emph{\ref{lem:3ov_cuts_hn_hd}} we have the cutting sequence $f(s,y_{n+1}) \dominates f(H_{k,1}, N_{n,1}) \undermines f(H_{k,1},D_{2n+3}^+) \dominates f(t,z_{n+1})$, yielding the final contradiction. Therefore, our initial assumption that $s \preceq H_{n,1}$ cannot hold, and we are done.
\end{proof}

Next, we show that there is at most no monotone paths in the region of free space between $y$-coordinates of $H_{n,1}$ and $D_{2n+3}^+$ that do not intersect the $y$-interval from $s$ to $t$. This fact will be useful for the second step in this section, which was to show that there are $2W+1$ monotone paths associated with the region between $y$-coordinates of $H_{k-1,1}$ and $H_{k,1}$, for some $1 \leq k \leq n$.

\begin{lemma}
\label{lem:3ov_no_one_up_top}
In the free space with bottom left corner $f(s, H_{n,1})$ and top right corner $f(t, D_{2n+3}^+)$, any monotone path from $l_s$ to $l_t$ overlaps in $y$-coordinate with the $y$-interval from $s$ to $t$ in more than one point.
\end{lemma}

\begin{proof}
Let $f(s,y_1) \to f(t,y_2)$ be a monotone path in the free space with bottom left corner $f(s, H_{n,1})$ and top right corner $f(t, D_{2n+3}^+)$. Therefore, $H_{n,1} \preceq y_1 \preceq y_2 \preceq D_{2n+3}^+$. Suppose $y_2 \prec G^+_{2n+1,2}$. Then $f(s,y_1) \dominates f(H_1^-,H_{n,1}) \domseq f(C_3^+,G^+_{2n+1,2}) \dominates f(t,y_2)$ by Lemma~\ref{lem:3ov_cuts}\emph{\ref{lem:3ov_cuts_hh_cg}}. There is a cutting sequence from $f(s,y_1)$ to $f(t,y_2)$, which is a contradiction. Hence, $G^+_{2n+1,2} \preceq y_2$. Suppose, $E_{2n+1}^+ \prec y_1$. By Lemma~\ref{lem:3ov_cuts}\emph{\ref{lem:3ov_cuts_he_hd}},
$f(s,y_1) \dominates f(H_1^-,E_{2n+1}^+) \undermines f(H_1^-,D_{2n+3}^+) \dominates f(t,y_2)$, which contradicts the fact that $f(s,y_1) \to f(t,y_2)$ is a monotone path. Hence, $y_1 \preceq E_{2n+1}^+$. Therefore $H_{n,1} \preceq y_1 \preceq E_{2n+1}^+ \prec G^+_{2n+1,2} \preceq y_2 \preceq D_{2n+3}^+$. But we also have $H_{n,1} \preceq s \prec G_{2n+1,1}^+$, so $E_{2n+1}^+ \prec t \prec D_{2n+3}^+$. Hence, the $y$-interval of the monotone path $f(s,y_1) \to f(t,y_2)$ intersects the $y$-interval from $s$ to $t$ in more than one point. 
\end{proof}

Now we are ready to prove the second step of this section.

\begin{lemma}
\label{lem:3ov_no_2d+1}
Let~$H_{0,1} = G_{1,1}$. There exists $1 \leq k \leq n$ so that, in the free space with bottom left corner $f(s, H_{k-1,1})$ and top right corner $f(t, H_{k,1})$, there are at least $2W+1$ monotone paths with non-overlapping $y$-coordinates. 
\end{lemma}

\begin{proof}
There are $m-1 = 2nW+1$ monotone in the free space with bottom left corner $f(s,H_{0,1})$ and top right corner $f(t,D_{2n+3}^+)$. By Lemma~\ref{lem:3ov_no_one_up_top} there are no monotone paths in the free space with bottom left corner $f(s, H_{n,1})$ and top right corner $f(t, D_{2n+3}^+)$ that do not intersect the $y$-interval from $s$ to $t$. Therefore, all $2nW+1$ monotone paths are in the free space with bottom left corner $f(s,H_{0,1})$ and top right corner $f(t,H_{n,1})$. By pigeonhole principle, there exists $1 \leq k \leq n$ so that there are $2W+1$ monotone paths in the free space with bottom left corner $f(s, H_{k-1,1})$ and top right corner $f(t, H_{k,1})$.
\end{proof}

For the remainder of this section, assume that $1 \leq k \leq n$ is an integer so that there are at least $2W+1$ monotone paths in the free space with bottom left corner $f(s,H_{k-1,1})$ and top right corner $f(t,H_{k,1})$. Using these paths we can narrow down the position of $s$ further.

\begin{lemma}
\label{lem:no_a_s_a}
$G_{1,1}^+ \preceq s \prec G_{2n+1,1}$
\end{lemma}

\begin{proof}
Suppose for the sake of contradiction that $H_{n,1} \prec s \prec G_{1,1}^+$. Let the monotone paths in the free space with bottom left corner $f(s,H_{k-1,1})$ and top right corner $f(t,H_{k,1})$ be $f(s,y_i) \to f(t,z_i)$ for $1 \leq i \leq 2W+1$. 

We will prove that $G_{k,i+1} \prec z_i$ for $1 \leq i \leq W$ by induction. Suppose for the sake of contradiction that $z_1 \preceq G_{k,2}$. Then by Lemmas~\ref{lem:3ov_cuts_gh_gg}, \ref{lem:3ov_cuts_gg_cc_1} and~\ref{lem:3ov_cuts_cc_gg_1} we have the cutting sequence $f(s,y_1) \dominates f(G_{1,1}^+,H_{k-1,1}) \undermines f(G_{1,1}^+,G_{k,h}) \domseq \Diamond_B(C_1^+,C_{k,h,1} \domseq \Diamond_T(G_{2n+1,2}^+,G_{k,h+1}) \dominates f(t,z_1)$, contradicting $f(s,y_1) \to f(t,z_1)$. 

Now we prove the inductive case. Assume the inductive hypothesis $G_{k,i+1} \prec z_i$. Then $G_{k,i+1} \prec y_{i+1}$. Suppose for the sake of contradiction that $z_{i+1} \preceq G_{k,i+2}$. Then by Lemma~\ref{lem:3ov_cuts}\emph{\ref{lem:3ov_cuts_gg_cc_1}} and~\emph{\ref{lem:3ov_cuts_cc_gg_1}} we have the cutting sequence $f(s,y_{i+1}) \dominates \Diamond_B(G_{1,1}^+, G_{k,i+1}) \domseq \Diamond_B(C_1^+,C_{k,h,1} \domseq \Diamond_T(G_{2n+1,2}^+,G_{k,i+2}) \dominates f(t,z_{i+1})$, contradicting the fact that $f(s,y_{i+1}) \to f(t,z_{i+1})$. This completes the induction.

Setting $i=W$ yields $G_{k,W+1} \prec z_W$. So $G_{k,W+1} \prec y_{W+1}$. But now, by Lemma, $f(s,y_{W+1}) \dominates f(G_{1,1}^+,G_{k,W+1}) \domseq f(C_3^+,H_{k,1}) \dominates f(t,z_{W+1})$, contradicting $f(s,y_{W+1}) \to f(t,z_{W+1})$. This yields a contradiction on our initial assumption, so $G_{1,1}^+ \preceq s$ as required.
\end{proof}

We can now assume without loss of generality that $G_{2i-1,1}^+ \prec s \preceq G_{2i+1,1}^+$ for some $1 \leq i \leq n$. Recall that this was the third step of this section. Next, we prove a similar result towards the fourth step.

\begin{lemma}
$J_{2i-1,1}^+ \prec s \preceq J_{2i-1,2n+1}^+$
\end{lemma}

\begin{proof}
It suffices to rule out the cases $G_{2i-1,1}^+ \prec s \preceq J_{2i-1,1}^+$, and $J_{2i-1,2n+1}^+ \prec s \preceq G_{2i+1,1}^+$. We will focus on the $2W+1$ paths in the free space with bottom left corner $f(s,H_{k-1,1})$ and top right corner $f(t,H_{k,1})$ be $f(s,y_i) \to f(t,z_i)$. The reason we can do this is that, for our final lemma, Lemma~\ref{lem:3ov_no_jj_to_orthogonal}, we only require that $J_{2i-1,1}^+ \prec s \preceq J_{2i-1,2n+1}^+$ and that there are $2W+1$ paths.

We start with the $G_{2i-1,1}^+ \prec s \preceq J_{2i-1,1}^+$ case. We will modify the $2W+1$ monotone paths so that $s = J_{2i-1,1}^+$. By Lemma~\ref{lem:sub_length_j_to_l}, we have $C_{2n+3}^+ \prec t \preceq L_{2i+3,3}^+$. Let the monotone paths in the free space with bottom left corner $f(s,H_{k-1,1})$ and top right corner $f(t,H_{k,1})$ be $f(s,y_i) \to f(t,z_i)$ for $1 \leq i \leq 2W+1$. Define $y_i' = A_{k,h,2}$ if $G_{k,h} \prec y_i \prec A_{k,h,2}$, otherwise $y_i' = y_i$. Define $z_i' = C_{k,h,2}$ if $D_{k,h,1} \preceq z_i \prec D_{k,h,2}$, define $z_i' = K_{1,k,h}^+$ if $D_{k,h,2} \preceq z_i \preceq K_{1,k,h}^+$ otherwise $z_i' = z_i$. We claim that if $f(s,y_i) \to f(t,z_i)$, then $f(J_{2i-1,1}^+,y_i') \to f(L_{2i+3,3}^+,z_i')$. If $G_{k,h} \prec y_i \prec A_{k,h,2}$, then $f(s,y_i)$ is to the right of $\Diamond(A_{2i-1,1}^+,A_{k,h})$ and we can replace the starting point $f(s,y_i)$ with $f(J_{2i-1,1}^+,A_{k,h,2})$. Otherwise, the segment between $f(s,y_i)$ and $f(J_{2i-1,1}^+,y_i)$ is free space, so we can replace the starting point $f(s,y_i)$ with $f(J_{2i-1,1}^+,y_i)$. Hence, we have $f(J_{2i-1,1}^+,y_i') \to f(t,z_i)$. If $D_{k,h,1} \preceq z_i \prec D_{k,h,2}$, then $f(t,z_i)$ is to the right of $\Diamond(C_{2i+3}^+,C_{k,h,2})$, so the monotone path $f(J_{2i-1,1}^+,y_i') \to f(t,z_i)$ must have passed under $\Diamond_B(C_{2i+3}^+,C_{k,h,2})$. Hence, we can replace the ending point $f(t,z_i)$ with $f(L_{2i+3,3}^+,C_{k,h,2})$. If $D_{k,h,2} \preceq z_i \preceq K_{1,k,h}^+$ then we append the monotone path $f(t,z_i) \to f(L_{2i+3,3}^+,K_{1,k,h}^+)$ to obtain the monotone path $f(J_{2i-1,1}^+,y_i') \to f(L_{2i+3,3}^+,K_{1,k,h}^+)$. Otherwise, the horizontal segment $f(t,z_i) \to f(L_{2i+3,3}^+,z_i)$ is free space, so we can append this to our path. Hence, we have modified all $2W+1$ monotone paths so that $s = J_{2i-1,1}^+$ and $t = L_{2i+3,3}^+$ as required.

Next we consider the $J_{2i-1,2n+1}^+ \prec s \preceq G_{2i+1,1}^+$. First, note that the subtrajectory from $J_{2n-1,2n+1}^+$ to $D_{2n+3}^+$ has length less than $\ell$, so we only need to consider the case where $1 \leq i < n$. The remainder of this proof is very similar to the previous case. We modify the $2W+1$ monotone paths. Our modification moves the starting $x$-coordinate to $J_{2i-1,2n+1}^+$ and the ending $x$-coordinate to $L_{2i+3,2n+1}^+$. For the starting $x$-coordinate, we observe that since $s \leq G_{2i+1,1}^+$, the point $f(s,y_i)$ is to the left of $\Diamond(G_{2i+1,1}^+, G_{k,h})$. Therefore, the horizontal segment from $f(J_{2i-1,2n+1}^+,y_i)$ to $f(s,y_i)$ must be free space. Therefore, we can prepend this to the monotone path to obtain $f(J_{2i-1,2n+1}^+,y_i) \to f(t,z_i)$. We truncate the monotone path at the $y$-coordinate $L_{2i+3,2n+1}^+$ to obtain $f(J_{2i-1,2n+1}^+,y_i) \to f(L_{2i+3,2n+1}^+,z_i')$ for some $z_i' \prec z_i$. Therefore, we modified all $2W+1$ monotone paths so that $s = J_{2i-1,1}^+$ and $t = L_{2i+3,3}^+$ as required.

Putting this all together, we obtain that  there are $2W+1$ paths in the free space with bottom left corner $f(s,H_{k-1,1})$ and top right corner, and $J_{2i-1,1}^+ \prec s \preceq J_{2i-1,2n+1}^+$ for some $1 \leq i \leq n$.
\end{proof}

For the remainder of this section, we can assume that $J_{2i-1,2j-1} \prec s \preceq J_{2i-1,2j+1}$ for some $1 \leq j \leq n$. Recall that this was the fourth step of this section.

Now we are ready to prove the fifth and final step of this section. We use the $2W+1$ paths from Lemma~\ref{lem:3ov_no_2d+1} and the cutting sequences from Lemma~\ref{lem:3ov_cutting_sequence} to show that that $X_i$, $Y_j$ and $Z_k$ are orthogonal, and therefore that $(\mathcal X, \mathcal Y, \mathcal Z)$ is a YES instance.

\begin{lemma}
\label{lem:3ov_no_jj_to_orthogonal}
Suppose $J_{2i-1,2j-1} \prec s \preceq J_{2i-1,2j+1}$ and there are $2W+1$ monotone paths in the free space diagram with bottom left corner $f(s, H_{k-1})$ and top right corner $f(t, H_k)$. Then $$X_i[h] \cdot Y_j[h] \cdot Z_k[h] = 0 \text{ for all } 1 \leq h \leq W.$$
\end{lemma}

\begin{proof}
Let the $2W+1$ monotone paths be $f(s,y_i) \to f(t,z_i)$ for $1 \leq i \leq 2W+1$. Suppose for the sake of contradiction that there exists $1 \leq \bar h \leq W$ such that $X_i[\bar h] \cdot Y_j[\bar h] \cdot Z_k[\bar h] = 1$. We will prove the following  series of statements by induction:

\begin{enumerate}[topsep=0pt,itemsep=0ex,label=\emph{(\roman*)}]
    \item for $1 \leq h < \bar h$, we have $C_{k,h,1} \preceq z_{2h-1}$, $G_{k,h} \prec z_{2h}$,
    \item for $h = \bar h$, we have $M_{2j+1,k,\bar h} \prec z_{2\bar h-1}$, $K_{2j-1,k,\bar h+1} \prec z_{2 \bar h}$,
    \item for $\bar h < h \leq W$, we have $E_{k,h,2} \preceq z_{2h-1}$, $A_{k,h+1,1} \prec z_{2h}$.
\end{enumerate}

After proving these three statements, we will use them to yield a contradiction. But first, we will prove $(i)$, $(ii)$ and $(iii)$, in order:

\begin{enumerate}[topsep=0pt,label=\emph{(\roman*)}]
    \item Let $1 \leq u \leq 2 \bar h - 2$. We prove the following by induction on $u$: if $u$ is odd, let $u=2h-1$, and we show that $C_{k,h,1} \preceq z_{2h-1}$, whereas if $u$ is even, let $u=2h$, and we show that $G_{k,h} \prec z_{2h}$.
    
    We start with the base case $u=1$. Suppose for the sake of contradiction that $z_1 \prec C_{k,1,2}$. Then by Lemma~\ref{lem:3ov_cuts}\emph{\ref{lem:3ov_cuts_aa_cc_1}}, we have $f(s,y_1) \dominates f(G_{3,1}^+,H_{k-1,1}) \undermines \Diamond_T(G_{3,1}^+,G_{k,h}) \dominates \Diamond_B(A_{2i+1,1}^+,A_{k,h,1}) \domseq \Diamond_T(C_{2i+3}^+,C_{k,h,1}) \dominates f(t,z_1)$, contradicting the claim that $f(s,y_1) \to f(t,z_1)$. This completes the base case. Next, we prove the inductive case of $(i)$.
    
    We split the inductive case into two subcases, either $u$ is odd or $u$ is even. 
    
    In the inductive case, if $u$ is odd, we assume the inductive hypothesis for $u-1$, which is even. Let $u = 2h-1$. Our inductive hypothesis states that $G_{k,h-1,2} \prec z_{2h-2}$, from which we get $G_{k,h-1,2} \prec y_{2h-1}$. Suppose for the sake of contradiction that $z_{2h-1} \prec C_{k,h,1}$. Then by Lemma~\ref{lem:3ov_cuts}\emph{\ref{lem:3ov_cuts_aa_cc_1}}, we have $f(s,y_{2h-1}) \dominates \Diamond_B(A_{2i+1,1}^+,A_{k,h,1}) \domseq \Diamond_T(C_{2i+3}^+,C_{k,h,1}) \dominates f(t,z_{2h-1})$, contradicting $f(s,y_{2h-1}) \to f(t,z_{2h-1})$. So $ C_{k,h,1} \preceq z_{2h-1}$, as required.
    
    In the inductive case, if $u$ is even, we assume the inductive hypothesis for $u-1$, which is odd. Let $u = 2h$. Our inductive hypothesis states that $C_{k,h,1} \preceq z_{2h-1}$, from which we get $C_{k,h,1} \preceq y_{2h}$. Suppose for the sake of contradiction that $z_{2h} \preceq G_{k,h}$. Then by Lemma~\ref{lem:3ov_cuts}\emph{\ref{lem:3ov_cuts_cc_gg_1}}, $f(s,y_{2h}) \dominates \Diamond_B(C_1^+,C_{k,h,1}) \domseq \Diamond_T(G_{2n+1,2}^+,G_{k,h}) \dominates f(t,z_{2h})$, contradicting $f(s,y_{2h}) \to f(t,z_{2h})$. So $G_{k,h} \prec z_{2h}$, as required.
    
    This completes the proof by induction of $(i)$.
    
    \item For $h = \bar h$, recall the following definitions. Recall that $\bar h$ is defined so that $X_i[\bar h] \cdot Y_j[\bar h] \cdot Z_k[\bar h] = 1$. Recall that $K_{2j-1,k,h}$ is the point on $G_{k,h} A_{k,h,1}$ so that $||K_{2j-1,k,h}J_{i,2j-1}^+|| = d$ and $||K_{2j-1,k,h}G_{2n+3-2j,2}^+|| = d$. Recall that $M_{2j+1,k,h}$ is the point on $D_{k,h,2} E_{k,h}$ so that $||M_{2j+1,k,h}L_{i,2j+1}^+|| = d$ and $||M_{2j+1,k,h}E_{2n+1-2j}^+|| = d$. Recall that $J_{2i-1,2j-1} \prec s \preceq J_{2i-1,2j+1}$, so by Lemma~\ref{lem:sub_length_j_to_l}, we have $L_{2i+3,2j+1} \prec t \preceq L_{2i+3,2j+3}$.
    
    We begin by showing $M_{2j+1,k,\bar h} \prec z_{2\bar h-1}$. Assume for the sake of contradiction that $z_{2\bar h-1} \preceq M_{2j+1,k,\bar h}$. Since $L_{2i+3,2j+1} \prec t$, we have $f(t,C_{k,h,2}) \undermines f(t,M_{2j+1,k,\bar h})$. Therefore, $z_{2\bar h-1} \prec C_{k,h,2}$. By $(i)$ we have $G_{k,h} \prec z_{2 \bar h - 2} \preceq y_{2 \bar h - 1}$. So $f(s,y_{2 \bar h - 1}) \dominates \Diamond_B(A_{2i,1}^+,A_{k,h,1})$. Since $X_i[h] \cdot Z_k[h] = 1$, by Lemma~\ref{lem:3ov_cuts}\emph{\ref{lem:3ov_cuts_aa_cc_2}}, we have that $\Diamond_B(A_{2i,1}^+,A_{k,h,1}) \domseq \Diamond_T(C_{2i+2}^+,C_{k,h,2})$. Finally, since $z_{2\bar h-1} \prec C_{k,h,2}$, we have $\Diamond_T(C_{2i+2}^+,C_{k,h,2}) \dominates f(s,z_{2 \bar h - 1})$. Putting this together we obtain the cutting sequence $f(s,y_{2 \bar h - 1}) \dominates \Diamond_B(A_{2i,1}^+,A_{k,h,1}) \domseq \Diamond_T(C_{2i+2}^+,C_{k,h,2}) \dominates f(s,z_{2 \bar h - 1}$, contradicting the fact that $f(s,y_{2 \bar h - 1}) \to f(s,z_{2 \bar h - 1})$. Therefore, $M_{2j+1,k,\bar h} \prec z_{2\bar h-1}$, as required.
    
    Next, we show $K_{2j-1,k,\bar h+1} \prec z_{2 \bar h}$. Assume for the sake of contradiction that $z_{2\bar h} \preceq K_{2j-1,k,\bar h+1}$. We know that $M_{2j+1,k,\bar h} \prec z_{2 \bar h -1}$. Therefore, $M_{2j+1,k,\bar h} \prec y_{2 \bar h}$. Since $Y_j[h] = 1$, by Lemma~\ref{lem:3ov_cuts}\emph{\ref{lem:3ov_cuts_ee_gg_2}}, we have that $f(s,y_{2 \bar h}) \dominates \Diamond_B(E_{2n+1-2j}^+,E_{k,\bar h}) \domseq \Diamond_T(G_{2n+1-2j,2}^+,G_{k,\bar h}) \dominates f(t,z_{2 \bar h})$, contradicting the fact that $f(s,y_{2\bar h}) \to f(t,z_{2 \bar h})$. Therefore, $K_{2j-1,k,\bar h+1} \prec z_{2 \bar h}$, as required.
    
    \item Let $2 \bar h + 1 \leq u \leq 2W$. We prove the following by induction on $u$: if $u$ is odd, let $u = 2h-1$, and we show that $E_{k,h,2} \preceq z_{2h-1}$, whereas if $u$ is even, let $u=2h$, and we show that $A_{k,h+1,1} \prec z_{2h}$.
    
    We start with the base case $u = 2 \bar h + 1$. Suppose for the sake of contradiction that $z_{2\bar h + 1} \prec E_{k,\bar h + 1, 2}$. By $(ii)$ we know that $K_{2j-1,k,\bar h+1} \prec z_{2\bar h}$, which implies that $K_{2j-1,k,\bar h_1} \prec y_{2\bar h+1}$. Since $s \preceq J_{2i-1,2j+1}$, we have $f(s,K_{2j-1,k,\bar h+1}) \undermines f(s,A_{k,\bar h+1,2})$. Therefore, $A_{k,\bar h+1,2} \prec y_{2 \bar h + 1}$. Now, by Lemma~\ref{lem:3ov_cuts}\emph{\ref{lem:3ov_cuts_aa_dd_1}}, we get $f(s,y_{2 \bar h + 1}) \dominates \Diamond_B(A_{2i-1,2}^+,A_{k,\bar h+1,2}) \domseq \Diamond_T(D_{2i+2}^+,D_{k,\bar h+1,2}) \dominates f(t,z_{2 \bar h + 1})$, contradicting $f(s,y_{2 \bar h + 1}) \to f(t,z_{2 \bar h + 1})$. Hence, $E_{k,\bar h + 1, 2} \preceq z_{2\bar h + 1}$ as required. This completes the base case. Next, we prove the inductive case of $(iii)$.
    
    We split the inductive case into two subcases, either $u$ is odd or $u$ is even. 
    
    In the inductive case, if $u$ is odd, we assume the inductive hypothesis for $u-1$, which is even. Since the base case $u = 2 \bar h + 1$ is already handled, we can assume that $2 \bar h + 2 \leq u-1$. Let $u = 2h-1$. Our inductive hypothesis states that $A_{k,h,1} \prec z_{2h-2}$, which implies $A_{k,h,1} \prec y_{2h-1}$. Suppose for the sake of contradiction that $z_{2h-1} \prec E_{k,h,2}$. Then by Lemma~\ref{lem:3ov_cuts}\emph{\ref{lem:3ov_cuts_aa_dd_1}} we have $f(s,y_{2 h - 1}) \dominates \Diamond_B(A_{2i-1,2}^+,A_{k,h,2}) \domseq \Diamond_T(D_{2i+2}^+,D_{k,h,2}) \dominates f(t,z_{2h-2})$, contradicting the claim that $f(s,y_{2h-1}) \to f(t,z_{2h-1})$. Hence, $E_{k,h,2} \preceq z_{2h-1}$, as required.
    
    In the inductive case, if $u$ is even, we assume the inductive hypothesis for $u-1$, which is odd. Let $u = 2h$. Our inductive hypothesis states that $E_{k,h,2} \preceq z_{2h-1}$, which which we get $E_{k,h,2} \preceq y_{2h}$. Suppose for the sake of contradiction that $z_{2h} \preceq A_{k,h+1,1}$. Then by Lemma~\ref{lem:3ov_cuts}\emph{\ref{lem:3ov_cuts_ee_gg_1}} we get $f(s,y_{2h}) \dominates \Diamond_B(E_1^+,E_{k,h}) \domseq \Diamond_T(G_{1,2}^+,G_{k,h}) \dominates f(t,z_{2h})$, contradicting the fact that $f(s,y_{2h}) \to f(t,z_{2h})$. Therefore, $A_{k,h+1,1} \prec z_{2h}$, as required. 
    
    This completes the proof by induction of $(iii)$.

\end{enumerate}

Finally, set $h=W$ in statement $(iii)$ to get $A_{k,W+1,1} \prec z_{2W}$. Therefore, $A_{k,W+1,1} \prec y_{2W+1}$. By Lemma~\ref{lem:3ov_cuts}\emph{\ref{lem:3ov_cuts_aa_dd_1}}, $f(s,y_{2W+1}) \dominates \Diamond_B(A_{2i-1,2}^+,A_{k,W+1,2}) \domseq \Diamond_T(D_{2i+2}^+,D_{k,W+1,2}) \undermines f(D_{2i+2}^+,H_{k,1}) \dominates f(t,z_{2W+1})$. This contradicts the fact that $f(s,y_{2W+1}) \to f(t,z_{2W+1})$. 

Hence, our initial assumption that there exists $1 \leq \bar h \leq W$ such that $X_i[\bar h] \cdot Y_j[\bar h] \cdot Z_k[\bar h] = 1$ cannot hold. Therefore, $X_i$, $Y_j$ and $Z_k$ are orthogonal.
\end{proof}

Putting this all together yields the main theorem of Section~\ref{sec:3ov_only_if}.

\begin{theorem}
\label{thm:3ov_no_to_no}
If our input $(\mathcal X, \mathcal Y, \mathcal Z)$ is a NO instance for 3OV, then our constructed instance $(T,m,\ell,d)$ is a NO instance for \problemtwo.
\end{theorem}

\subsection{Putting it all together}

\label{sec:3ov_putting_together}

Finally, we combine Theorems~\ref{theorem:3ov_yes_to_yes} and~\ref{thm:3ov_no_to_no} to obtain our main theorem of Section~\ref{sec:3ov}.

\theoremfour*

\begin{proof}
Suppose for the sake of contradiction that there is an $O(n^{3-\varepsilon})$ time algorithm for \problemtwo under the continuous Fr\'echet distance, for some $\varepsilon > 0$. We will use this to construct an  $O(n^{3-\varepsilon}W^{O(1)})$ time algorithm for 3OV, which cannot occur unless SETH fails~\cite{williams2018some}.

Our algorithm for 3OV is as follows. We obtain our input $(\mathcal X, \mathcal Y, \mathcal Z)$ for 3OV. We construct the instance $(T,m,\ell,d)$ as described in Section~\ref{sec:3ov_construction}.  Next, we run our algorithm for \problemtwo under the continuous Fr\'echet distance. If our algorithm returns YES for $(T,m,\ell,d)$, we return YES for $(\mathcal X, \mathcal Y, \mathcal Z)$, likewise in the NO case. This completes the description of our algorithm.

The correctness of our algorithm follows directly from Theorems~\ref{theorem:3ov_yes_to_yes} and~\ref{thm:3ov_no_to_no}. Finally, we analyse the running time of our 3OV algorithm. Constructing the instance $(T,m,\ell,d)$ takes $O(nW)$ time, and the complexity of $T$ is $O(nW)$. Our \problemtwo algorithm for the input $(T,m,\ell,d)$ takes $O(n^{3-\varepsilon}W^{3-\varepsilon})$ time. Therefore, the overall running time is $O(n^{3-\varepsilon}W^{O(1)})$, which cannot occur unless SETH fails.
\end{proof}

\bibliographystyle{plain}
\bibliography{main}

\newpage
\appendix

\section{Fr\'echet distance and Free space diagram}

\label{sec:appendix_frechet_freespace}
A trajectory $T$ is defined by a sequence of points $v_1, v_2, \ldots, v_n$ in $c$-dimensional Euclidean space, denoted $\mathbb R^c$. 
A trajectory $T$ of complexity $n$ is defined to be a continuous function $T:[0,n] \to \mathbb R^c$ satisfying $T(i) = v_i$ for all integers $1 \leq i \leq n$. Moreover, $T(i + \mu) = (1 - \mu) v_i + \mu v_{i+1}$ for all integers $1 \leq i \leq n-1$ and all reals $0 \leq \mu \leq 1$.

In Sections~\ref{sec:overview_3ov} and~\ref{sec:3ov} we make use of the following notation. If $T_1$ and $T_2$ are trajectories, we define $T_1 \circ T_2$ to be, first, the trajectory along $T_1$, then, a segment from the endpoint of $T_1$ to the start point of $T_2$, and finally, the trajectory $T_2$. In the special case of vertices, $v_1 \circ v_2$ is simply the segment from $v_1$ to $v_2$. In this way, we can write $T = v_1 \circ v_2 \circ \ldots \circ v_n$. Alternatively, we can use the notation $T = \bigcirc_{1 \leq i \leq n} v_i$.

Let $\Gamma(n)$ be the space of all continuous non-decreasing functions from $[0,1]$ to $[1,n]$. The continuous Fr\'echet distance between a pair of trajectories $T_1$ and $T_2$ of complexities $n_1$ and $n_2$ is defined to be:

\[
    d_F(T_1,T_2) = \inf_{\substack{\alpha_1 \in \Gamma(n_1) \\ \alpha_2 \in \Gamma(n_2)}} \ \max_{\mu \in [0,1]} ||T_1(\alpha_1(\mu)) - T_2(\alpha_2(\mu))||,
\]
where $|| \cdot ||$ denotes the Euclidean norm. Let $\Delta(n_1,n_2)$ be the set of all sequences of pairs of integers $(x_1,y_1), (x_2,y_2), \ldots (x_k,y_k)$, satisfying $(x_1, y_1) =  (1,1)$, $(x_k, y_k) = (n_1, n_2)$, and $(x_{i+1}, y_{i+1}) \in \{(x_i+1,y_i), (x_i+1,y_i+1), (x_i, y_i+1)\}$. The discrete Fr\'echet distance between a pair of trajectories $T_1$ and $T_2$ of complexities $n_1$ and $n_2$ is defined to be:

\[
    d_{DF}(T_1,T_2) = \min_{\alpha \in \Delta(n_1,n_2)} \max_{(x,y) \in \alpha} ||T_1(x) - T_2(y)||.
\]

The continuous free space diagram $F_d(T_1,T_2)$ between a pair of trajectories $T_1$ and $T_2$ of complexities $n_1$ and $n_2$ is defined to be the rectangle $[1,n_1] \times [1,n_2]$ separated into free and non-free space:
\[
    F_d(T_1,T_2) = 
    \begin{cases}
        (s,t) \in [1,n_1] \times [1,n_2] :  ||T_1(s) - T_2(t)|| \leq d, \quad \text{then $(s,t)$ is in free space}\\
        (s,t) \in [1,n_1] \times [1,n_2] :  ||T_1(s) - T_2(t)|| > d,
        \quad \text{then $(s,t)$ is in non-free space.}
    \end{cases}
\]
If we divide the $F_d(T_1,T_2)$ into $n_1$ columns and $n_2$ rows of unit width, this divides the space into $n_1n_2$ cells. See Figure~\ref{fig:pre_01_fsd}. In each cell, the free space is the intersection of an ellipse with the cell. The ellipse can intersect the boundary of the cell in at most eight points. We define these intersection points to be critical points of $F_d(T_1,T_2)$. We also define the corners of cells to be critical points of $F_d(T_1,T_2)$.

\begin{figure}[ht]
    \centering
    \includegraphics{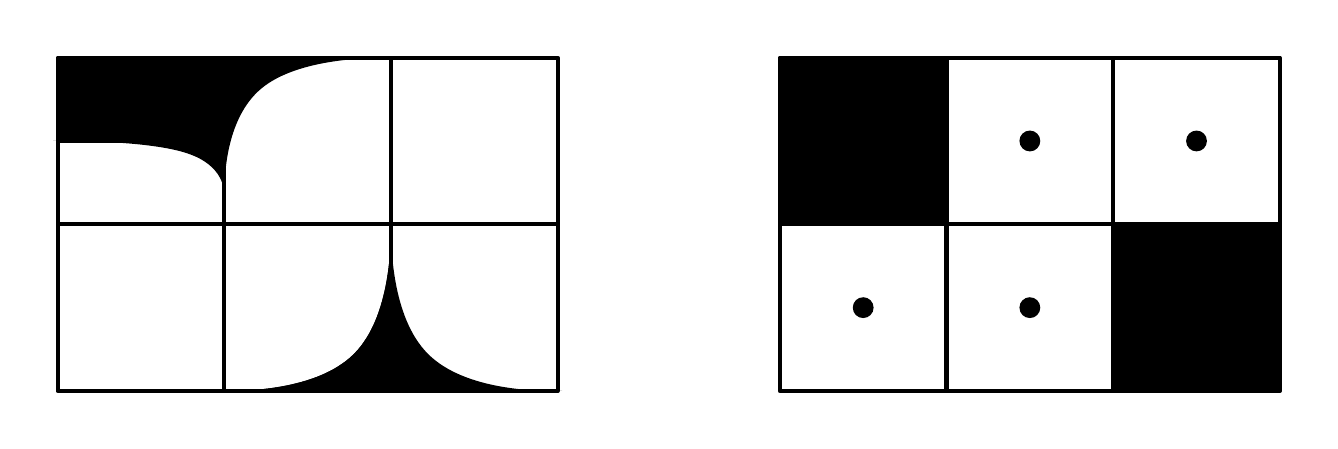}
    \caption{Examples of the continuous free space diagram (left) and the discrete free space diagram (right).}
    \label{fig:pre_01_fsd}
\end{figure}

The discrete free space diagram between a pair of trajectories $T_1$ and $T_2$ of complexities $n_1$ and $n_2$ is defined the set of grid points $\{1,2,\ldots,n_1\} \times \{1,2,\ldots,n_2\}$ separated into free and non-free space~\cite{DBLP:journals/ijcga/AltG95}. 

\[
    F_d(T_1,T_2) = 
    \begin{cases}
        (s,t) \in \{1,\ldots,n_1\} \times \{1,\ldots,n_2\}: ||T_1(s) - T_2(t)|| \leq d,  \text{ then $(s,t)$ is a free grid point}\\
        (s,t) \in \{1,\ldots,n_1\} \times \{1,\ldots,n_2\}: ||T_1(s) - T_2(t)|| > d,  \text{ then $(s,t)$ is a non-free point.}\\
    \end{cases}
\]

With slight abuse of notation, $F_d(T_1,T_2)$ is used to denote both the discrete and continuous free space diagrams, although it will be clear from context which similarity measure is being considered.

A monotone path in continuous free space is a continuous path that is monotone with respect to both its $x$- and $y$-coordinates. The continuous Fr\'echet distance between trajectories $T_1$ and $T_2$ (of complexities $n_1$ and $n_2$) is at most $d$ if and only if there exists a monotone path in $F_d(T_1,T_2)$ starting at $(1,1)$ and ending at $(n_1,n_2)$.

A monotone path in discrete free space is a sequence of grid points $(x_1,y_1), (x_2,y_2), \ldots (x_k,y_k)$ satisfying $(x_{i+1}, y_{i+1}) \in \{(x_i+1,y_i), (x_i+1,y_i+1), (x_i, y_i+1)\}$. The discrete Fr\'echet distance between $T_1$ and $T_2$ (of complexities $n_1$ and $n_2$) is at most $d$ if and only if there exists a monotone path in $F_d(T_1,T_2)$ starting at $(1,1)$ and ending at $(n_1,n_2)$.

\begin{figure}[ht]
    \centering
    \includegraphics{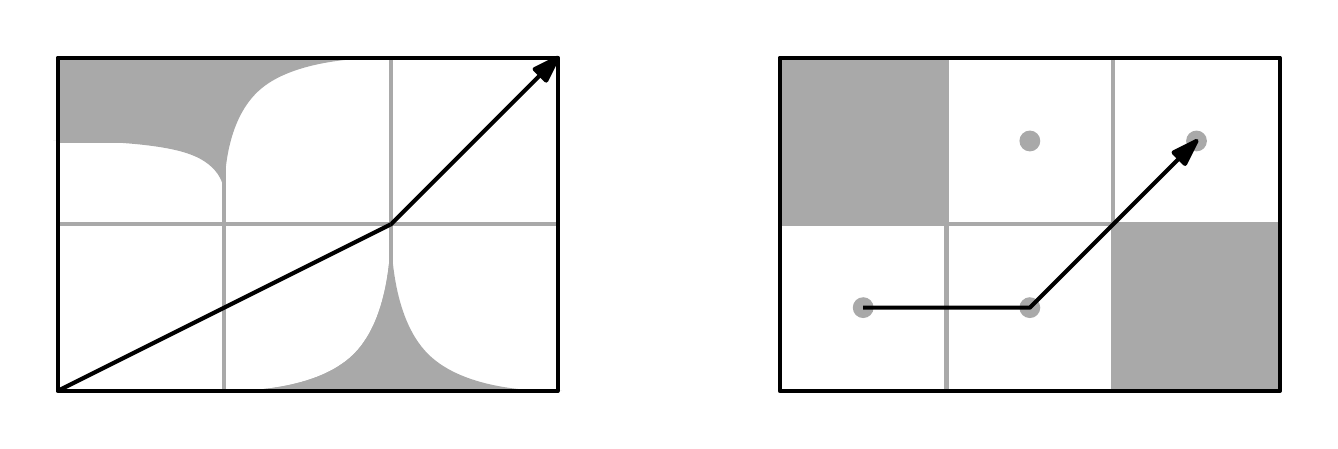}
    \caption{Examples of monotone paths in the continuous (left) and discrete free space (right).}
    \label{fig:pre_02_fsd}
\end{figure}

\end{document}